\documentclass[a4paper,11pt]{article}
\usepackage{tgpagella}
\usepackage[utf8]{inputenc}
\usepackage{xspace}
\usepackage{natbib}
\usepackage{graphicx,fullpage,paralist}
\usepackage{amsmath,amssymb,amsthm}
\usepackage{mathtools}
\usepackage{hyperref}
\usepackage{thmtools}
\usepackage{tikz}
\usepackage{tikz-network}
\usepackage{pgfplots}
\usepackage{varwidth}
\pgfplotsset{compat=1.10}
\usepgfplotslibrary{fillbetween}
\usetikzlibrary{backgrounds}
\usetikzlibrary{patterns}
\usetikzlibrary{arrows}
\usetikzlibrary{plotmarks}
\usetikzlibrary{fadings}
\usepackage{enumitem}
\usepackage{blkarray}
\usepackage{csquotes}
\usepackage{boldline}
\usepackage{geometry}
\usepackage{stmaryrd}
\usepackage{float}
\newcommand{\calG}{\mathcal{G}}
\newcommand{\calH}{\mathcal{H}}

\newcommand{\calF}{\mathcal{F}}

\newcommand{\calS}{\mathcal{S}}

\newcommand{\NN}{\mathbb{N}}

\newcommand{\ie}{i.e.,\xspace}
\newcommand{\eg}{e.g.,\xspace}

\newtheorem{theorem}{Theorem}
\newtheorem{lemma}{Lemma}

\newtheorem{claim}{Claim}

\usepackage{multirow}
\usepackage[ruled]{algorithm2e}
\usepackage{xcolor}

\newlength{\algofontsize}
\hypersetup{
	colorlinks,
	linkcolor={red!50!black},
	citecolor={blue!50!black},
	urlcolor={blue!80!black}
}

\vfuzz \hfuzz
\begin{document}

	\title{Impartial Selection with Additive Guarantees\\via Iterated Deletion}
	 
	\author{Javier Cembrano
	\thanks{Institut für Mathematik, Technische Universität Berlin, Germany}
	\and Felix Fischer			
	\thanks{School of Mathematical Sciences, Queen Mary University of London, UK}
	\and David Hannon		
	\thanks{School of Mathematical Sciences, Queen Mary University of London, UK}
	\and Max Klimm
	\thanks{Institut für Mathematik, Technische Universität Berlin, Germany}
	}
\date{\vspace{-1em}}
\maketitle

\begin{abstract}
Impartial selection is the selection of an individual from a group based on nominations by other members of the group, in such a way that individuals cannot influence their own chance of selection. For this problem, we give a deterministic mechanism with an additive performance guarantee of $O(n^{(1+\kappa)/2})$ in a setting with $n$ individuals where each individual casts $O(n^{\kappa})$ nominations, where $\kappa\in[0,1]$. For $\kappa=0$, \ie when each individual casts at most a constant number of nominations, this bound is $O(\sqrt{n})$. This matches the best-known guarantee for randomized mechanisms and a single nomination. For $\kappa=1$ the bound is $O(n)$. This is trivial, as even a mechanism that never selects provides an additive guarantee of $n-1$. We show, however, that it is also best possible: for every deterministic impartial mechanism there exists a situation in which some individual is nominated by every other individual and the mechanism either does not select or selects an individual not nominated by anyone.
\end{abstract}

\section{Introduction}
Many important decisions are made by groups of people who select one member of the group based on nominations from within the group. In fact, such a procedure is applied in a wide range of settings such as the papal conclaves~\cite{MacK20a}, the national captain's votes for the best FIFA Women's and Men's Players Awards, or the voting for spokespersons of committees or student representatives. A common feature of these situations is that there is at least a partial overlap between the set of voters and the set of candidates. This overlap may be a source of incentive issues as candidates who have a reasonable chance of winning may be motivated not to reveal their true opinion on who should win in order to increase their own chance of winning. 
Incentive issues of this kind were first studied in a systematic way by \citet{holzman2013impartial} and \citet{alon2011sum}, who formalized the problem in terms of a directed graph in which vertices correspond to voters and a directed edge from one voter to another indicates that the former nominates the latter. A (deterministic) selection mechanism then takes such a nomination graph as input and returns either one of its vertices or no vertex. In order to allow voters to express their true opinions about other voters without having to worry about their own chance of selection, an important property of a selection mechanism is its \emph{impartiality}: a mechanism is impartial if, for all nomination graphs, a change of the outgoing edges of some vertex $v$ does not change whether $v$ is selected or not. It is easy to see that a mechanism that selects a vertex with maximum indegree and breaks ties in some consistent way is not impartial: if ties are broken in favor of greater index, for example, a vertex with maximum indegree that currently nominates another vertex with maximum indegree but greater index has an incentive to instead nominate a different vertex.

\citeauthor{holzman2013impartial} have shown that impartial mechanisms are in fact much more limited even in a setting where each voter casts exactly one vote, \ie where each vertex has outdegree one: for every impartial mechanism that selects a vertex in every nomination graph, there exists a graph where the mechanism selects a vertex with indegree zero or a graph where it fails to select a vertex with indegree $n-1$, where $n$ is the number of voters. This shows in particular that the best multiplicative approximation guarantee for impartial mechanisms, \ie the worst case over all nomination graphs of the ratio between the maximum indegree and the indegree of the selected vertex is at least $n-1$. On the other hand, a multiplicative guarantee of $n-1$ is easy to obtain by always following the outgoing edge of a fixed vertex. 
As multiplicative guarantees do not allow for a meaningful distinction among deterministic impartial mechanisms, \citet{caragiannis2022impartial,caragiannis2023impartial} proposed to instead consider an additive guarantee, \ie the worst case over all nomination graphs of the difference between the maximum indegree and the indegree of the selected vertex. 
An adaptation provided by \citet{MacK20a} of a
mechanism due to \citeauthor{holzman2013impartial} achieves an additive guarantee of~$\lceil n/2 \rceil$.\footnote{The mechanism does not always select a vertex, but we will see that for our purposes there is little difference between mechanisms that do and do not always select.}
\citet{caragiannis2022impartial} further proposed a randomized mechanism with an additive guarantee of $O(\sqrt{n})$. It remained open, however, whether there exists a \emph{deterministic} mechanism with a \emph{sublinear} additive guarantee. 

The setting studied by \citeauthor{holzman2013impartial}, where each vertex has outdegree one, is commonly referred to as the plurality setting. The impossibility results of \citeauthor{holzman2013impartial} regarding multiplicative guarantees carry over to the more general approval setting, where outdegrees can be arbitrary, but here even less is known about possible additive guarantees. While \citet{caragiannis2022impartial} have shown that deterministic impartial mechanisms cannot provide a better additive guarantee than~$3$, no mechanism is known that improves on the trivial guarantee of $n-1$ achieved by selecting a fixed vertex. \citeauthor{caragiannis2022impartial} also gave a randomized mechanism for the approval setting with an additive guarantee of $\Theta(n^{2/3} \ln^{1/3} n)$.

\subsection{Our Contribution} 

We develop a new deterministic mechanism for impartial selection that is parameterized by a pair of thresholds on the indegrees of vertices in the graph. The mechanism seeks to select a vertex with large indegree, and to achieve impartiality it iteratively deletes outgoing edges from vertices in decreasing order of their indegrees, until only the outgoing edges of vertices with indegrees below the lower threshold remain. It then selects a vertex with maximum remaining indegree if that indegree is above the higher threshold, and otherwise does not select. Any ties are broken according to a fixed ordering of the vertices. We give a sufficient condition for choices of thresholds that guarantee impartiality. The iterative nature of the deletions requires a fairly intricate analysis but is key to achieving impartiality. The additive guarantee is then obtained for a good choice of thresholds, and the worst case is the one where the mechanism does not select.

For instances with $n$ vertices and maximum outdegree at most $O(n^{\kappa})$, where $\kappa\in[0,1]$, the mechanism provides an additive guarantee of $\smash{O(n^{\frac{1+\kappa}{2}})}$. This is the first sublinear bound for a deterministic mechanism and any $\kappa\in[0,1]$, and is sublinear for all $\kappa\in[0,1)$. 
For settings with constant maximum outdegree, which includes the setting of \citeauthor{holzman2013impartial} where all outdegrees are equal to one, our bound matches the best known bound of $O(\sqrt{n})$ for randomized mechanisms and outdegree one, due to \citet{caragiannis2022impartial}.

When the maximum outdegree is unbounded, the bound becomes $O(n)$. This is of course trivial, as even a mechanism that never selects a vertex provides an additive guarantee of $n-1$. For a setting without abstentions, \ie with minimum outdegree one, the guarantee can be improved slightly to $n-2$ by following the outgoing edge of a fixed vertex. We show that both of these bounds are best possible by giving matching lower bounds. This improves on the only lower bound known prior to our work, again due to \citeauthor{caragiannis2022impartial}, which is equal to~$3$ and applies to the setting with abstentions and mechanisms that select a vertex for every graph.

Just like the lower bounds regarding multiplicative guarantees for plurality, our lower bounds for approval are obtained through an axiomatic impossibility result. \citeauthor{holzman2013impartial} have shown that in the case of plurality, impartiality is incompatible with positive and negative unanimity. Here, positive unanimity requires that a vertex with the maximum possible indegree of $n-1$ must be selected, and negative unanimity that a vertex with indegree zero cannot be selected.\footnote{This result is formulated for mechanisms that always select a vertex, but not selecting a vertex clearly does not improve the situation.} We show that in the case of approval this impossibility can be strengthened even further: call a selection mechanism weakly unanimous if it selects a vertex with positive indegree whenever there exists a vertex with the maximum possible indegree of $n-1$; then weak unanimity and impartiality are incompatible, even if we are not always required to select.

This result is obtained by analyzing the behavior of impartial mechanisms on a restricted class of graphs with a high degree of symmetry among vertices. Like \citeauthor{holzman2013impartial}, we can assume that isomorphic vertices are selected with equal probabilities by a randomized relaxation of a mechanism. 
A suitable class of graphs for our purposes are those generated by partial orders on the set of ordered partitions. Our result is then obtained by combining counting results for ordered partitions and an argument similar to Farkas' Lemma.

\subsection{Related Work}

Impartiality as a formal property of social and economic mechanisms was first considered by \citet{de2008impartial} for the distribution of a divisible commodity among a set of individuals according to the individuals' subjective claims.
\citet{holzman2013impartial} and \citet{alon2011sum} studied impartial selection in two different settings, plurality and approval, and established strong impossibility results regarding the ability of deterministic mechanisms to approximate the maximum indegree in a multiplicative sense. \citet{alon2011sum} also proposed randomized mechanisms for the selection of one or more vertices. \citet{fischer2015optimal} then obtained a randomized mechanism with the best possible multiplicative guarantee of~$2$ for the selection of a single vertex in the approval setting, and \citet{bjelde2017impartial} gave improved deterministic and randomized mechanisms for the selection of more than one vertex.
\citet{niemeyer2023simple} studied strategyproof mechanisms for the selection of a single individual under more general preferences.

Starting from the observation that impossibility results for randomized mechanisms in particular are obtained from graphs with very small indegrees, \citet{bousquet2014near} developed a randomized mechanism that is optimal in the large-indegree limit, \ie that chooses a vertex with indegree arbitrarily close to the maximum indegree as the latter goes to infinity. \citet{caragiannis2022impartial,caragiannis2023impartial} used the same observation as motivation to study mechanisms with additive rather than multiplicative guarantees. They developed new mechanisms that achieve such guarantees, established a relatively small but nontrivial lower bound of~$3$ for deterministic mechanisms in the approval setting, and gave improved deterministic mechanisms for a setting with prior information.

The axiomatic study of \citeauthor{holzman2013impartial} has been refined and extended in a number of ways, for example with a focus on symmetric mechanisms~\citep{mackenzie2015symmetry} and to the selection of more than one vertex~\citep{TaOh14a}. \citet{MacK20a} provided a detailed axiomatic analysis of mechanisms used in the papal conclave
and showed that impartial mechanisms satisfying natural symmetry and anonymity axioms are exactly those that select a vertex whose indegree exceeds a supermajority threshold.
Various selection mechanisms have also been proposed that are tailored to applications like peer review and exploit the particular preference and information structures of those applications~\citep{kurokawa2015impartial,xu2019strategyproof,aziz2019strategyproof,mattei2021peernomination}.
Impartial mechanisms have finally been considered for other objectives, specifically for the maximization of progeny~\citep{babichenko2020incentive,zhang2021incentive} and for rank aggregation~\citep{kahng2018ranking}.
For an overview of mechanisms for peer selection and evaluation, the reader is referred to the survey of \citet{olckers2022manipulation}.

The proof of our impossibility result uses a class of graphs constructed from ordered partitions of the set of vertices. The class has been studied previously~\citep[\eg][]{insko2017ordered,diagana2017some}, and some of its known properties including its lattice structure and the number of graphs isomorphic to each graph within the class are relevant to us.

\section{Preliminaries}

We let~$\NN$ denote the positive integers. For $n\in \NN$, let~$[n]=\{1,\ldots,n\}$ denote the set of integers from~$1$ to~$n$ and~$[n]_0=\{0,1,\ldots,n\}$ denote the set of integers from~$0$ to~$n$. Let $\calG_n = \left\{(N, E): N = [n], E \subseteq (N \times N ) \setminus \bigcup_{v\in N}\{(v,v)\}\right\}$ be the set of directed graphs with $n$ vertices and no loops. Let $\calG = \bigcup_{n\in \NN} \calG_n$. For $G=(N,E)\in\calG$ and $v\in N$, let $N^+(v, G)=\{u\in N: (v,u) \in E\}$ be the out-neighborhood and $N^-(v, G)=\{u\in N: (u,v)\in E\}$ the in-neighborhood of~$v$ in~$G$. Let $\delta^+(v,G)=|N^+(v,G)|$ and $\delta^-(v,G)=|N^-(v,G)|$ denote the outdegree and indegree of~$v$ in~$G$, 
\[ 
    \delta^-_S(v,G)=|\{u\in S: (u,v)\in E\}|
\]
the indegree of~$v$ from a particular subset $S\!\subseteq\! N$ of the vertices, and $\Delta(G) \!=\! \max_{v\in N} \delta^-(v, G)$ the maximum indegree of any vertex in $G$. 
When the graph is clear from the context, we will sometimes drop $G$ from the notation and write $N^+(v)$, $N^-(v)$, $\delta^+(v)$, $\delta^-(v)$, $\delta^-_S(v)$, and~$\Delta$. 
For $n,k\in\NN$, let $\calG^+_n = \{(N, E)\in \calG_n:  \delta^+(v)\geq 1 \text{ for every } v\in N\}$ be the set of graphs in $\calG_n$ where all outdegrees are strictly positive, 
\[
    \calG_n(k)=\left\{(N, E)\in \calG_n:  \delta^+(v)\leq k \text{ for every } v\in N\right\}
\]
the set of graphs in $\calG_n$ where outdegrees are at most $k$, and $\calG^+_n(k) = \calG^+_n \cap \calG_n(k)$ for the set of graphs satisfying both conditions. Let $\calG^+ = \bigcup_{n\in \NN} \calG^+_n,\ \calG(k) = \bigcup_{n\in \NN} \calG_n(k)$, and $\calG^+(k) = \bigcup_{n\in \NN} \calG^+_n(k)$.

A (deterministic) \textit{selection mechanism} is then given by a family of functions $f \colon \calG_n \to 2^N$
that maps each graph to a subset of its vertices, where we require throughout that $|f(G)|\leq 1$ for all $G\in\calG$.
By slightly abusing notation, we will use $f$ to refer to both the mechanism and to individual functions from the family.
Mechanism $f$ is \textit{impartial} on $\calG'\subseteq \calG$ if on this set of graphs the outgoing edges of a vertex have no influence on its selection, \ie if for every pair of graphs $G = (N, E)$ and $G' = (N, E')$ in $\calG'$ and every $v\in N$, $f(G)\cap\{v\}=f(G')\cap\{v\}$ whenever $E\setminus(\{v\}\times N)=E'\setminus(\{v\}\times N)$.
Mechanism $f$ is \textit{$\alpha$-additive} on $\calG' \subseteq \calG$, for $\alpha \geq 0$, if for every graph in $\calG'$ the indegree of the choice of $f$ differs from the maximum indegree by at most $\alpha$, \ie if
\[
    \sup_{\substack{G\in \calG'}} \bigl\{ \Delta(G) - \delta^-(f(G), G) \bigr\} \leq \alpha,
\]
where for $S\subseteq N$ we let $\delta^-(S,G)=\sum_{v\in S}\delta^-(v,G)$ denote the sum of the indegrees of the vertices in $S$.

\section{Iterated Deletion of Nominations}

When outdegrees are at most one, the following simple mechanism is $\lfloor n/2\rfloor$-additive: if there is a vertex with indegree at least $\lfloor n/2\rfloor+1$, select it; otherwise, do not select. This mechanism is a slight modification of a mechanism \citet{holzman2013impartial} call majority with default and corresponds to the supermajority rule of \citet{MacK20a} with threshold $\lfloor n/2\rfloor+1$.
As there can be at most one vertex with degree $\lfloor n/2\rfloor+1$ or more and a vertex cannot influence its own indegree, the mechanism is clearly impartial.
We will borrow from this mechanism the idea of imposing a threshold on the minimum indegree a vertex needs to be selected, but will seek to lower the threshold in order to achieve a better additive guarantee and also to relax the constraint on the maximum outdegree. Of course, lower thresholds and larger outdegrees both mean that more and more vertices become eligible for selection, and we will no longer get impartiality for free.

As a first step, it is instructive to again consider the outdegree-one case but to use a lower threshold of $t = \lfloor n/3\rfloor+1$. This choice of threshold implies that there can be at most two vertices with indegrees equal to or higher than the threshold.
Analogously to the supermajority rule, the general idea is that only vertices with indegree at least $t$ are eligible for selection. 
In cases where there is only one such vertex impartiality would follow by a similar argument as for the supermajority rule, so the critical case is when there are two vertices~$u$ and~$v$ with indegree at least $t$.
If $u$ is selected, $\delta^-(u)=t$, and $(v,u)\in E$, then~$v$ could remove the edge $(v,u)$, cause $u$ to drop below the threshold, and leave~$v$ as the only vertex eligible for selection. To prevent impartiality from being compromised in this case, it makes sense to introduce another threshold $T = t+1 = \lfloor n/3\rfloor +2$ in the understanding that only vertices with indegree at least $T$ are eligible for selection, but the outgoing edges of all vertices with indegree at least $t = \lfloor n/3\rfloor +1$ will automatically be ignored. It turns out that in order to actually achieve impartiality, the outgoing edges of the vertices with indegree at least $t$ must be removed in decreasing order of indegrees, breaking ties according to some fixed ordering of the vertices. In the end we select the vertex with maximum indegree among those with indegree at least $T$, breaking ties as before, and call the resulting mechanism \emph{two contenders (TC)}.
It is not too difficult to convince ourselves that the mechanism is indeed impartial.\footnote{This will follow as a special case from a more general result in \autoref{lem:impartiality-additive-ub}, but we find it insightful to give the informal argument here.}
Indeed, if $\delta^-(u)\geq T$ and $\delta^-(v)\geq T$, neither $u$ nor $v$ can influence whether the respective other vertex ends up with an indegree of at least $t=T-1$; the outgoing edges of both $u$ and $v$ will be removed, and neither $u$ nor $v$ can influence which of the two vertices is selected. If $\delta^-(u)\geq T$ and $\delta^-(v)=t$, only $u$ is eligible for selection; the outgoing edge of $u$, which has the highest indegree, is removed first, so $u$ cannot influence whether $v$ will have indegree at least $t$ and thus whether a potential edge $(v,u)$ is removed or not. All other cases are either analogous or lead to no vertex being selected, and impartiality follows.

It is natural to ask whether the threshold can be lowered further while maintaining impartiality, and whether guarantees can be obtained in a similar fashion for graphs with larger outdegrees. The answer to the first question is not obvious, as the number of graphs that need to be considered to establish impartiality grows very quickly in the number of vertices eligible for selection. The obvious generalization of the mechanism with threshold $\lfloor n/2\rfloor+1$ to a setting with outdegrees at most~$k$, of selecting the unique vertex with indegree at least $\lfloor kn/2\rfloor+1$ if it exists and not selecting otherwise, is impartial and $\lfloor kn/2 \rfloor$-additive, but this guarantee is trivial when $k\geq 2$.

Our main result answers both questions in the affirmative. It applies to settings with outdegree at most $k=O(n^{\kappa})$ for $\kappa\in[0,1]$ and provides a nontrivial guarantee when $\kappa\in[0,1)$. When~$k$ is constant the guarantee is $O(\sqrt{n})$, which matches the best guarantee known for randomized mechanisms and outdegree one~\citep{caragiannis2022impartial}.
\begin{theorem}
\label{thm:additive-ub}
For every $n\in\NN$, $\kappa\in[0,1]$, and $k=O(n^{\kappa})$, there exists an impartial and $O(n^{\frac{1+\kappa}{2}})$-additive mechanism on $\calG_n(k)$. Specifically, for every $n\in\NN$, there exists an impartial and $\sqrt{7.25n}$-additive mechanism on $\calG_n(1)$.
\end{theorem}

The result is achieved by a mechanism of a family we call \emph{twin threshold mechanisms} and describe formally in \autoref{alg:TTM}. Each of these mechanisms iteratively deletes the outgoing edges from vertices with indegree above a first threshold $t$ from the highest to the lowest indegree and, in the end, selects the vertex with maximum remaining indegree as long as that indegree is above a second, higher threshold $T$. The parameters~$t$ and~$T$ characterize a specific mechanism of this family and will be chosen in order to achieve impartiality and obtain the desired bounds.
Throughout the mechanism ties are broken as before, in favor of greater index.
\begin{algorithm}[t]
	\SetAlgoNoLine
	\KwIn{Digraph $G=(N,E)\in \calG_n$.}
	\KwOut{Set $S\subseteq N$ of selected vertices with $|S|\leq 1$.}
	Initialize $i \xleftarrow{} 0$ and $d\xleftarrow{} \Delta$\;
	$D^i\xleftarrow{} \emptyset$ \tcp*{vertices with deleted outgoing edges in iteration $i$ or before}
	$\hat{\delta}^i(v) \xleftarrow{} \delta^-(v) \quad \forall v\in N$ \tcp*{indegree of $v$ omitting edges deleted up to~$i$}
	\While{$d\geq t$}{
	    \If{$\{u\in N\setminus D^i:  \hat{\delta}^i(u)=d\} = \emptyset$}{
	        update $d\xleftarrow{} d-1$\;
                {\bf continue}
        }
	    $v\gets\max\{u\in N\setminus D^i:  \hat{\delta}^i(u)=d\}$\;
	    update $\hat{\delta}^{i+1}(u) \xleftarrow{} \hat{\delta}^i(u)-1 \quad\forall u\in N^+(v)$\tcp*{delete outgoing edges of $v$}
     update $\hat{\delta}^{i+1}(u) \xleftarrow{} \hat{\delta}^i(u)\quad\forall u\in N\setminus N^+(v)$\;
	    update $D^{i+1} \xleftarrow{}D^i \cup \{v\}$ and $i \xleftarrow{} i+1$\;
	}
	$I\xleftarrow{} i$\;
        $\hat{\Delta} \gets \max_{v \in N}{\hat{\delta}^I(v)}$\tcp*{maximum indegree after deletion}
	\eIf{$\hat{\Delta} \geq T$}{
	    $S \gets \{\max\{v\in N:  \hat{\delta}^I(v)=\hat{\Delta}\}\}$
	}{
            $S\gets \emptyset$
        }
	{\bf return} $S$
	\caption{Twin threshold mechanism with thresholds $T$ and $t$}
	\label{alg:TTM}
\end{algorithm}

For any choice of the maximum outdegree~$k$ and the threshold parameters~$T$ and~$t$, the twin threshold mechanism with thresholds $T$ and $t$ achieves its worst additive performance guarantee in cases where it does not select, and this guarantee can be obtained in a straightforward way by bounding the maximum indegree. The proof of impartiality, on the other hand, uses a relatively subtle argument to show that for certain values of $T$ and $t$, a vertex above the higher threshold~$T$ cannot influence whether another vertex ends up above or below the lower threshold~$t$ when edges have been deleted. 
Vertices above~$T$ then have no influence on the set of edges taken into account for selection, and since these are the only vertices that can potentially be selected impartiality follows.

For the outdegree-one case and specific values of $n$, we can optimize over $t$ and $T$ in order to obtain the best guarantee; bounds for an optimal twin threshold mechanism and selected values of $n$ are shown in \autoref{fig:plot_alpha_n}, alongside bounds for the supermajority rule with threshold $\lfloor n/2\rfloor+1$, corresponding to a twin threshold mechanism with $T=t=\lfloor n/2\rfloor +1$, the two contenders mechanism, corresponding to a twin threshold mechanism with $T=t+1=\lfloor n/3\rfloor +2$, and lower bounds for any impartial mechanism.
It is interesting to note at this point that the supermajority rule with threshold $\lfloor n/2\rfloor+1$ possesses natural symmetry properties~\citep{MacK20a}, and that the optimal twin threshold mechanism does not possess these properties due to its use of tie-breaking and its distinction among edges depending on their origin. The improved additive guarantee provided by the optimal twin threshold mechanism and shown in the figure thus comes at the expense of a certain amount of fairness among vertices.
\definecolor{color1}{HTML}{1b9e77}
\definecolor{color2}{HTML}{d95f02}
\definecolor{color3}{HTML}{7570b3}
\definecolor{color4}{HTML}{e7298a}

\pgfdeclareplotmark{halfcircle}{%
\pgfpathcircle{\pgfpoint{0pt}{0pt}}{\pgfplotmarksize}
\pgfusepathqstroke
\pgfpathmoveto{\pgfpoint{\pgfplotmarksize}{0pt}}
\pgfpatharc{0}{180}{\pgfplotmarksize}
\pgfpathclose
\pgfusepathqfill
}
\pgfdeclareplotmark{thirdcircle}{%
\pgfpathcircle{\pgfpoint{0pt}{0pt}}{\pgfplotmarksize}
\pgfusepathqstroke
\pgfpathmoveto{\pgfpointorigin}
\pgfpathlineto{\pgfpoint{\pgfplotmarksize}{0pt}}
\pgfpatharc{0}{120}{\pgfplotmarksize}
\pgfpathclose
\pgfusepathqfill
}

\begin{figure}[t]
    \centering
\begin{tikzpicture}[xscale=0.2,yscale=0.2]
\draw[ultra thick,-latex] (-1,0) -- (55,0) node[right] {$n$};
\draw[ultra thick,-latex] (0,-1) -- (0,27) node[above] {$\alpha$};
\foreach \n in {2,10,20,30,40,50}{
    \draw (\n,0.5) -- (\n,-0.5) node[below] {$\n$};
}
\foreach \a in {0,10,20,20}{
    \draw (0.5,\a) -- (-0.5,\a) node[left] {$\a$};
}
\fill[lightgray,path fading=east] (5,2) -- (9,2) -- (10,3) -- (50,3) -- (50,16) -- (49,15) -- (45,15) -- (44,14) -- (40,14) -- (39,13) -- (36,13) -- (35,12) -- (32,12) -- (31,11) -- (28,11) -- (27,10) -- (24,10) -- (23,9) -- (21,9) -- (20,8) -- (18,8) -- (17,7) -- (15,7) -- (14,6) -- (12,6) -- (11,5) -- (10,5) -- (9,4) -- (8,4) -- (7,3) -- (6,3) -- cycle;
\draw[very thick,color1,dash pattern=on 0.25pt off 0.75pt] plot[mark=*, mark color=none,mark options={fill=color1,draw opacity=0},mark size=8pt] coordinates { (2,1) (3,1) (4,2) (5,2) (6,3) (7,3) (8,4) (9,4) (10,5) (11,5) (12,6) (13,6) (14,7) (15,7) (16,8) (17,8) (18,9) (19,9) (20,10) (21,10) (22,11) (23,11) (24,12) (25,12) (26,13) (27,13) (28,14) (29,14) (30,15) (31,15) (32,16) (33,16) (34,17) (35,17) (36,18) (37,18) (38,19) (39,19) (40,20) (41,20) (42,21) (43,21) (44,22) (45,22) (46,23) (47,23) (48,24) (49,24) (50,25)} node[right, align=left] {\small SR};
\draw[very thick,color4,dash pattern=on 0.25pt off 0.75pt, dash phase=0.25pt] plot[mark=*, mark color=none, mark options={color=color4,draw opacity=0},mark size=8pt] coordinates {(2,1) (3,1) (4,2) (5,2) (6,2) (7,2) (8,2) (9,2) (10,3) (11,3) (12,3) (13,3) (14,3) (15,3) (16,3) (17,3) (18,3) (19,3) (20,3) (21,3) (22,3) (23,3) (24,3) (25,3) (26,3) (27,3) (28,3) (29,3) (30,3) (31,3) (32,3) (33,3) (34,3) (35,3) (36,3) (37,3) (38,3) (39,3) (40,3) (41,3) (42,3) (43,3) (44,3) (45,3) (46,3) (47,3) (48,3) (49,3) (50,3)} node[right] {\small LB};
\draw[very thick,color2,dash pattern=on 0.25pt off 0.75pt, dash phase=0.5pt] plot[mark=*, mark color=none,mark options= {color=color2,draw opacity=0},mark size=8pt] coordinates { (2,2) (3,3) (4,3) (5,3) (6,4) (7,4) (8,4) (9,5) (10,5) (11,5) (12,6) (13,6) (14,6) (15,7) (16,7) (17,7) (18,8) (19,8) (20,8) (21,9) (22,9) (23,9) (24,10) (25,10) (26,10) (27,11) (28,11) (29,11) (30,12) (31,12) (32,12) (33,13) (34,13) (35,13) (36,14) (37,14) (38,14) (39,15) (40,15) (41,15) (42,16) (43,16) (44,16) (45,17) (46,17) (47,17) (48,18) (49,18) (50,18)} node[right] {\small TC};
\draw[very thick,color3, dash pattern=on 0.25pt off 0.75pt, dash phase=0.75] plot[mark=*, mark color=none, mark options = {color=color3,draw opacity=0}, mark size=8pt] coordinates {
(2,1) (3,1) (4,2) (5,2) (6,3) (7,3) (8,4) (9,4) (10,5) (11,5) (12, 6) (13, 6) (14, 6) (15, 7) (16, 7) (17, 7) (18, 8) (19, 8) (20, 8) (21, 9) (22, 9) (23, 9) (24, 10) (25, 10) (26, 10) (27, 10) (28, 11) (29, 11) (30, 11) (31, 11) (32, 12) (33, 12) (34, 12) (35, 12) (36, 13) (37, 13) (38, 13) (39, 13) (40, 14) (41, 14) (42, 14) (43, 14) (44, 14) (45, 15) (46, 15) (47, 15) (48, 15) (49, 15) (50, 16)
} node[right] {\small TT$^*$};

\path plot[mark=thirdcircle, mark options = {color=color4,draw opacity=0},mark size=8pt] coordinates { (2,1) (3,1) (4,2) (5,2)};
\path plot[mark=thirdcircle, mark options = {color=color1,draw opacity=0,rotate=120},mark size=8pt] coordinates { (2,1) (3,1) (4,2) (5,2) (8,4) (10,5) (11,5) (13,6)};
\path plot[mark=halfcircle, mark options = {color=color1,draw opacity=0,rotate=90}, mark size =8pt] coordinates {(6,3) (7,3) (9,4) (12,6) (15,7)};
\path plot[mark=thirdcircle, mark options = {color=color2,draw opacity=0},mark size=8pt] coordinates { (8,4) (10,5) (11,5) (13,6)};
\path plot[mark=halfcircle, mark options = {color=color2,draw opacity=0,rotate=-90}, mark size =8pt] coordinates {(14,6) (16,7) (17,7) (20,8)};
\end{tikzpicture}

~\\[0.2cm]

\centering
\begin{tabular}{>{\centering\arraybackslash}m{1cm} >{\centering\arraybackslash}m{1cm} >{\centering\arraybackslash}m{1cm} >{\centering\arraybackslash}m{1cm} >{\centering\arraybackslash}m{1cm} >{\centering\arraybackslash}m{1cm} >{\centering\arraybackslash}m{1cm} >{\centering\arraybackslash}m{1cm} >{\centering\arraybackslash}m{1cm} >{\centering\arraybackslash}m{1cm} >{\centering\arraybackslash}m{1cm}}
\hlineB{3}
\textbf{$n$}         & 2 & 3 & 4 & 5 & 10 & 20 & 50 & 100 & 500 & 1000 \\ \hline
\textcolor{color1}{SR\phantom{$^*$}} & 1 & 1 & 2 & 2 & 5  & 10  & 25 & 50  & 250  & 500   \\
\textcolor{color2}{TC\phantom{$^*$}} & 2 & 3 & 3 & 3 & 5  & 8  & 18 & 35 & 168 & 335 \\
\textcolor{color3}{TT$^*$} & 1 & 1 & 2 & 2 & 5  & 8  & 16 & 23  & 56  & 81   \\[-1.3ex]
\tiny{$(t, T)$\phantom{$^*$}} & \tiny{$(2, 2)$} & \tiny{$(2, 2)$} & \tiny{$(3, 3)$} & \tiny{$(3, 3)$} & \tiny{$(4, 5)$}  & \tiny{$(7, 8)$}  & \tiny{$(7, 11)$} & \tiny{$(13, 18)$}  & \tiny{$(28, 41)$}  & \tiny{$(36, 56)$}   \\
\textcolor{color4}{LB\phantom{$^*$}} & 1 & 1 & 2 & 2 & 3 & 3 & 3 & 3 & 3 & 3 \\\hlineB{3}
\end{tabular}
\caption{Values of $\alpha_n$, for certain values of $n$, such that the supermajority rule with threshold $\lfloor n/2\rfloor+1$ (SR), the two contenders mechanism (TC), and an optimal twin threshold mechanism (TT$^*$) are $\alpha_n$-additive on $\calG_n(1)$, together with a lower bound (LB) on $\alpha_n$ for any impartial $\alpha_n$-additive mechanism. Each bound for TT$^*$ in the table is accompanied by values of the thresholds $t$ and $T$ that achieve the bound, which may not be the unique such values.}
\label{fig:plot_alpha_n}
\end{figure}

A twin threshold mechanism maintains a set $D^i$ of vertices whose outgoing edges have been deleted up to iteration $i$, and indegrees $\hat{\delta}^i(v)$ where all edges deleted up to iteration $i$ are not counted. After the iterations during which outgoing edges are deleted, we add an artificial final iteration during which no edges are deleted; this iteration counter is denoted by $I$. 
In order to prove \autoref{thm:additive-ub}, we will compare runs of the mechanism with fixed thresholds for different graphs, and denote by $\hat{\delta}^i(v,G)$, $D^i(G)$, and $I(G)$ the respective values of $\hat{\delta}^i(v)$, $D^i$, and $I$ when the mechanism is invoked with input graph $G$. We use $\chi$ to denote the indicator function for logical propositions, \ie $\chi(p)=1$ when proposition $p$ is true and $\chi(p)=0$ otherwise.
For a graph~$G$ and a vertex~$v$ in~$G$ whose outgoing edges are deleted by the mechanism, we use~$i^{*}(v,G)$ to denote the iteration in which this deletion takes place, such that $D^{i^{*}(v,G)+1}(G)\setminus D^{i^{*}(v,G)}(G) = \{v\}$. We use the convention that $i^{*}(v,G)=I(G)$ if $v\not\in D^{I(G)}(G)$ to extend the function to all vertices. For a graph~$G$ and vertex~$v$, we write $\smash{\delta^{*}(v,G)=\hat{\delta}^{i^{*}(v,G)}(v,G)}$ for the indegree of $v$, not taking into account any incoming edges deleted previously, at the last moment before the outgoing edges of~$v$ are deleted. When the graph $G$ is clear from the context, we again drop $G$ from the notation.
It is clear from the definition of the mechanism that for any graph $G=(N,E)$ and vertex $v\in D^{I(G)}(G)$,
\begin{equation}\label{eq:vertex-descent-characterization-iterations}
    \delta^-(v,G)-\delta^{*}(v,G) = |\{u\in N^-(v,G):  i^{*}(u,G)< i^{*}(v,G)\}|.
\end{equation}
When comparing tuples of the form $(\delta^{*}(v), v)$, we use $\preceq$ and $\prec$ to denote lexicographic order and strict lexicographic order. These comparisons are relevant to our analysis because they determine the order in which edges are deleted. Specifically, for any graph $G=(N,E)$ and vertices $u,v\in D^{I(G)}(G)$, 
\begin{equation}\label{eq:equivalence-iterations-indegrees}
i^{*}(u,G)<i^{*}(v,G) \quad\text{if and only if}\quad (\delta^{*}(u,G), u) \succ (\delta^{*}(v,G), v).
\end{equation}

\begin{figure}[t]
\centering
\begin{tikzpicture}[scale=0.88]
\Vertex[y=1.95, Math, shape=circle, color=black, , size=.05, label=u_0, fontscale=1, position=above, distance=-.09cm]{A}
\Vertex[x=1.4, y=1.95, Math, shape=circle, color=black, size=.05, label=v, fontscale=1, position=above, distance=-.08cm]{B}
\Vertex[x=2.1, y=1.95, Math, shape=circle, color=black, size=.05, label=u_1, fontscale=1, position=above, distance=-.09cm]{C}
\Vertex[x=.7, y=.65, Math, shape=circle, color=black, size=.05, label=u_2, fontscale=1, position=above left, distance=-.16cm]{D}
\Edge[Direct, color=black, lw=1pt](A)(B)
\Edge[Direct, color=black, lw=1pt](C)(B)
\Edge[Direct, color=black, lw=1pt](D)(B)
\draw[] (-.2,-.2) -- (2.3,-.2);
\draw[] (-.2,.45) -- (2.3,.45);
\draw[] (-.2,1.1) -- (2.3,1.1);
\draw[] (-.2,1.75) -- (2.3,1.75);
\draw[] (-.2,2.4) -- (2.3,2.4);

\Vertex[x=4, y=1.95, Math, shape=circle, color=black, , size=.05, label=u_0, fontscale=1, position=above, distance=-.09cm]{E}
\Vertex[x=5.4, y=1.3, Math, shape=circle, color=black, size=.05, label=v, fontscale=1, position=above, distance=-.08cm]{F}
\Vertex[x=6.1, y=1.95, Math, shape=circle, color=black, size=.05, label=u_1, fontscale=1, position=above, distance=-.09cm]{G}
\Vertex[x=4.7, y=.65, Math, shape=circle, color=black, size=.05, label=u_2, fontscale=1, position=above left, distance=-.16cm]{H}
\Edge[Direct, color=black, lw=1pt](G)(F)
\Edge[Direct, color=black, lw=1pt](H)(F)
\draw[] (3.8,-.2) -- (6.3,-.2);
\draw[] (3.8,.45) -- (6.3,.45);
\draw[] (3.8,1.1) -- (6.3,1.1);
\draw[] (3.8,1.75) -- (6.3,1.75);
\draw[] (3.8,2.4) -- (6.3,2.4);

\Vertex[x=8, y=1.95, Math, shape=circle, color=black, , size=.05, label=u_0, fontscale=1, position=above, distance=-.09cm]{I}
\Vertex[x=9.4, y=.65, Math, shape=circle, color=black, size=.05, label=v, fontscale=1, position=above, distance=-.08cm]{J}
\Vertex[x=10.1, y=1.95, Math, shape=circle, color=black, size=.05, label=u_1, fontscale=1, position=above, distance=-.09cm]{K}
\Vertex[x=8.7, y=.65, Math, shape=circle, color=black, size=.05, label=u_2, fontscale=1, position=above, distance=-.09cm]{L}
\Edge[Direct, color=black, lw=1pt](L)(J)
\draw[] (7.8,-.2) -- (10.3,-.2);
\draw[] (7.8,.45) -- (10.3,.45);
\draw[] (7.8,1.1) -- (10.3,1.1);
\draw[] (7.8,1.75) -- (10.3,1.75);
\draw[] (7.8,2.4) -- (10.3,2.4);
\Vertex[x=12, y=1.95, Math, shape=circle, color=white, , size=.05, label=u_0, fontscale=1, position=above, distance=-.09cm]{M}
\Vertex[x=13.4, Math, shape=circle, color=black, size=.05, label=v, fontscale=1, position=above, distance=-.08cm]{N}
\Vertex[x=14.1, y=1.95, Math, shape=circle, color=black, size=.05, label=u_1, fontscale=1, position=above, distance=-.09cm]{O}
\Vertex[x=12.7, y=.65, Math, shape=circle, color=black, size=.05, label=u_2, fontscale=1, position=above, distance=-.09cm]{P}
\draw[] (11.8,-.2) -- (14.3,-.2);
\draw[] (11.8,.45) -- (14.3,.45);
\draw[] (11.8,1.1) -- (14.3,1.1);
\draw[] (11.8,1.75) -- (14.3,1.75);
\draw[] (11.8,2.4) -- (14.3,2.4);
\Text[x=1.05, y=-.6, fontsize=\scriptsize]{$i=0$}
\Text[x=5.05, y=-.6, fontsize=\scriptsize]{$i=1$}
\Text[x=9.05, y=-.6, fontsize=\scriptsize]{$i=2$}
\Text[x=13.05, y=-.6, fontsize=\scriptsize]{$i=3$}
\end{tikzpicture}
\caption{Illustration of the edge deletion process in \autoref{alg:TTM}. 
Here and in other figures, we arrange the vertices above~$t$ vertically according to their indegree, with larger indegrees above, and horizontally according to their index, with greater indices to the left. Vertices below $t$, as well as edges with one or both endpoints below $t$, are omitted for simplicity.}
\label{fig:example-mechanism}
\end{figure}
\autoref{fig:example-mechanism} shows an example of the edge deletion process over four iterations of \autoref{alg:TTM}. Observe that for $j\in \{1,2\}$, $(\delta^-(v),v)\succ(\delta^-(u_j),u_j)$ but $(\delta^{*}(v),v)\prec(\delta^{*}(u_j),u_j)$. This is caused by a drop in the indegree of~$v$ over the course of the algorithm, and this drop occurs before the algorithm considers possible outgoing edges of~$v$ for deletion. For our analysis, it will be important to bound how much the indegree of a vertex~$v$ can drop before~$v$ loses its outgoing edges. The following lemma, which we prove in~\autoref{app:indegree-changes}, characterizes this quantity in terms of the indegrees of the in-neighbors of~$v$.
\begin{lemma}
\label{lem:indegree-changes}
Let $G=(N,E)\in \calG_n,\ v\in N$ and $(d,z)\in \NN^2$ such that $$(\delta^-(v),v)\succ(d, z)\succeq (\delta^{*}(v), v).$$ Let $r =\delta^-(v)-d+\chi(v>z)$.
Then there exist vertices $u_0,\ldots,u_{r-1}$ such that for each $j\in [r-1]_0$, $(u_{j},v)\in E$ and $(\delta^{*}(u_{j}),u_j)\succ(\delta^-(v)-j, v)$. Moreover, if $(d,z)=(\delta^{*}(v),v)$, then for every vertex $u\in N^-(v)\setminus \{u_0,\ldots,u_{r-1}\}$, $(\delta^{*}(u),u)\prec(\delta^{*}(v),v)$.
\end{lemma}

If we take $(d,z)=(\delta^{*}(v),v)$, the lemma implies that for any vertex $v\in D^{I}$,
\begin{equation}\label{eq:vertex-descent-characterization-indegrees}
    \delta^-(v)-\delta^{*}(v) = |\{u\in N^-(v):  (\delta^{*}(u), u)\succ(\delta^{*}(v)+1,v)\}|.
\end{equation}
In other words, if the indegree of a vertex~$v$ drops from $\delta^-(v)$ to $\delta^{*}(v)=\delta^-(v)-r$ before the outgoing edges of~$v$ are deleted, there must be~$r$ vertices with edges to~$v$ that satisfy the following property: at least one of them, $u_0$, must have indegree high enough for its outgoing edges to be deleted before those of $v$, \ie $(\delta^{*}(u_0),u_0)\succ(\delta^-(v), v)$; another vertex, $u_1$, must have indegree high enough for its outgoing edges to be deleted before those of $v$ after its indegree is reduced by one, \ie $(\delta^{*}(u_1),u_1)\succ(\delta^-(v)-1, v)$; and so forth, up to $u_{r-1}$ with $(\delta^{*}(u_{r-1}),u_{r-1})\succ(\delta^-(v)-(r-1), v)=(\delta^{*}(v)+1, v)$.
The other vertices $u$ with edges to $v$ must satisfy $(\delta^{*}(u),u)\prec(\delta^{*}(v),v)$.
This is illustrated in \autoref{fig:lemma-1} for the case where the indegree of $v$ drops by $r=3$. When $(d,z)\succ(\delta^{*}(v),v)$ it is enough to carry out the analysis for the first $r=\delta^-(v)-d+\chi(v>z)$ vertices $u_0,\ldots,u_{r-1}$.

\begin{figure}[t]
\centering
\begin{tikzpicture}[scale=0.88]
\draw[] (-2,-.2) -- (2,-.2);
\draw[] (-2,.4) -- (2,.4);
\draw[] (-2,1) -- (2,1);
\draw[] (-2,1.6) -- (2,1.6);
\draw[] (-2,2.2) -- (2,2.2);
\fill[color=gray] (-2,2.2) rectangle (2,2.8);
\fill[color=gray] (-2,1.6) rectangle (-.1,2.2);
\fill[color=gray,pattern=dots] (.1,1.6) rectangle (2,2.2);
\fill[color=gray,pattern=dots] (-2,1) rectangle (-.1,1.6);
\fill[color=gray,pattern=vertical lines] (.1,1) rectangle (2,1.6);
\fill[color=gray,pattern=vertical lines] (-2,.4) rectangle (-.1,1);
\Text[x=0, y=2.5]{\Large{\bf A}}
\Text[x=-1, y=1.9]{\Large{\bf A}}
\Text[x=-1, y=1.3]{\Large{\bf B}}
\Text[x=1, y=1.9]{\Large{\bf B}}
\Text[x=1, y=1.3]{\Large{\bf C}}
\Text[x=-1, y=.7]{\Large{\bf C}}
\Vertex[y=1.8, Math, shape=circle, color=black, size=.05, label=v, fontscale=1, position=above, distance=-.08cm]{A}
\Vertex[Math, shape=circle, color=black, size=.05]{B}
\Edge[Direct, style=dashed, color=black, lw=1pt](A)(B)
\end{tikzpicture}
\caption{Illustration of \autoref{lem:indegree-changes} for $r=3$. If the indegree of $v$ drops as shown by the dashed arrow, there must be a vertex with an edge to~$v$ in~$A$, another vertex with an edge to~$v$ in $A\cup B$, and a third vertex with an edge to~$v$ in $A\cup B\cup C$. Note that this exact condition is satisfied for the example of \autoref{fig:example-mechanism}.}
\label{fig:lemma-1}
\end{figure}

To establish impartiality of the twin threshold mechanism with certain thresholds, we need to compare runs of the mechanism on graphs that differ in the outgoing edges of a single vertex. Intuitively, a change in the outgoing edges will make a difference to the outcome of the mechanism only if it affects the position of some other vertex relative to the lower threshold $t$ at the time that vertex is considered: If at that time the vertex is above the threshold its outgoing edges are deleted, otherwise the edges remain and are used in the decision of which vertex to select. We are thus interested in pairs of graphs $G_1$ and $G_2$ that differ only in the outgoing edges of a vertex $\tilde{v}$, and which contain a second vertex~$v\neq\tilde{v}$ such that $\delta^{*}(v,G_1)>\delta^{*}(v,G_2)$. 
Using \autoref{lem:indegree-changes}, we can derive conditions in terms of the indegrees of the in-neighbors of~$v$ under which this can happen. Moreover, we can show that one of two additional conditions must be satisfied: either (i)~$\tilde{v}$ has an edge to $v$ in $G_1$, or (ii)~there exists a vertex $u_0$ with an edge to~$v$ in both $G_1$ and $G_2$ such that $\delta^{*}(u_0,G_1)<\delta^{*}(u_0,G_2)$.
We obtain the following lemma.
\begin{lemma}
\label{lem:indegree-changes-two-graphs}
Let $G_1=(N,E_1),G_2=(N,E_2)\in \calG_n$, $v,\ \tilde{v} \in N$ with $v \not= \tilde{v}$ be such that $E_1\setminus(\{ \tilde{v} \}\times N)=E_2\setminus (\{ \tilde{v} \}\times N),\ \delta^{*}(v,G_1)>\delta^{*}(v,G_2)$, and $\delta^{*}(v,G_1)\geq t$.
Consider $(d,z)\in \NN^2$ such that $(\delta^{*}(v,G_1),v) \succ (d, z) \succeq (\delta^{*}(v,G_2), v)$,
and let $r=\delta^{*}(v,G_1)-d+\chi(v>z)$.
Then, there exist vertices $u_0,\ldots,u_{r-1}$ such that, for every $j\in [r-1]$,
\[
    (u_j,v)\in E_1\cap E_2,\quad (\delta^{*}(u_{j},G_2),u_j)\succ (\delta^{*}(v,G_1)-j,v),\quad (\delta^{*}(u_j,G_1),u_j)\prec (\delta^{*}(v,G_1),v),
\]
and one of the following holds:
\begin{enumerate}[label=(\roman*)]
    \item $u_0=\tilde{v}$ and $(u_0,v)\in E_1\setminus E_2$. Moreover, if $(\delta^{*}(v,G_1),v)\prec(\delta^-(\tilde{v},G_1),\tilde{v})$, taking $\tilde{r}= \delta^-(\tilde{v},G_1)-\delta^{*}(v,G_1)+\chi(\tilde{v}>v)$
    we have that there are vertices $\tilde{u}_0,\ldots,\tilde{u}_{\tilde{r}-1}$, none of them equal to $\tilde{v}$, such that $(\delta^{*}(\tilde{u}_j,G_1),\tilde{u}_{j})\succ (\delta^-(\tilde{v},G_1)-j,\tilde{v})$ for every $j\in[\tilde{r}-1]_0$; or \label{indegree-changes-two-graphs-alt1}
    \item $(\delta^{*}(u_0,G_2),u_0)\succ\dots \succ (\delta^{*}(u_{r-1},G_2),u_{r-1})$, $(u_0,v)\in E_1\cap E_2$, and $(\delta^{*}(u_0,G_2),u_0)\succ (\delta^{*}(v,G_1),v) \succ (\delta^{*}(u_0,G_1),u_0)$. \label{indegree-changes-two-graphs-alt2}
\end{enumerate}
\end{lemma}

We prove \autoref{lem:indegree-changes-two-graphs} in \autoref{app:indegree-changes-two-graphs} but provide some intuition for its correctness here. Assume that the indegree of a vertex $v$ drops from $\delta^-(v)$ to $\delta^{*}(v)=\delta^-(v)-r$ before its outgoing edges are deleted. 
Then, by \autoref{lem:indegree-changes}, there must be $r$ in-neighbors of~$v$ whose indegrees are at least $\delta^{*}(v)+1$ when their outgoing edges are deleted. 
Thus, when $v$, $G_1$, and $G_2$ are as in the statement of \autoref{lem:indegree-changes-two-graphs}, and defining $r_1=\delta^-(v, G_1)-\delta^{*}(v,G_1)$ and $r_2=\delta^-(v, G_2)-\delta^{*}(v,G_2)$, then $r_1$ in-neighbors of $v$ must have indegree at least $\delta^{*}(v,G_1)+1$ in $G_1$ upon deletion of their outgoing edges, and $r_2$ in-neighbors of $v$ must have indegree at least $\delta^{*}(v,G_2)+1$ in $G_2$ upon deletion of their outgoing edges. 
There are then two possible reasons for the difference between $\delta^{*}(v,G_1)$ and $\delta^{*}(v,G_2)$. The first, which can only occur when $r_1=r_2$, is given in Condition~\ref{indegree-changes-two-graphs-alt1} of the lemma: $\tilde{v}$ has an edge to $v$ in $G_1$ but not in $G_2$, while all the other indegrees remain the same. However, for this difference to have an impact, the outgoing edges of $\tilde{v}$ must be deleted after those of $v$, and \autoref{lem:indegree-changes} implies the existence of in-neighbors of $\tilde{v}$ with indegrees as shown. 
The other reason, which can happen when $r_1=r_2$ and necessarily happens otherwise, is that some in-neighbor $u_0$ of $v$ in both $G_1$ and $G_2$ loses its outgoing edges after $v$ when the input to the mechanism is $G_1$, but before $v$ when the input is $G_2$. This must happen due to a change in the indegree of~$u_0$ at the time its outgoing edges are deleted, \ie $$(\delta^{*}(u_0,G_2),u_0)\succ (\delta^{*}(v,G_1),v) \succ (\delta^{*}(u_0,G_1),u_0).$$ 
This is captured in Condition~\ref{indegree-changes-two-graphs-alt2}. 
In both cases, \autoref{lem:indegree-changes} implies the existence of $\delta^{*}(v,G_1)-\delta^{*}(v,G_2)-1$ further in-neighbors of $v$ in both graphs, denoted as $u_1,\ldots, u_{r-1}$, which lose their outgoing edges after $v$ in $G_1$ but before $v$ in $G_2$. 
When $(d,z)\succ(\delta^{*}(v,G_2),v)$, it again suffices to carry out the analysis for a smaller subset of the vertices with edges to~$v$.

\autoref{lem:indegree-changes-two-graphs} implies that whenever $G_1$ and $G_2$ differ only in the outgoing edges of a single vertex $\tilde{v}$, and $v$ is a different vertex with $\delta^{*}(v,G_1)>\delta^{*}(v,G_2)$, then either (i)~$\tilde{v}$ has an edge to $v$ in $G_1$, or (ii)~there exists a vertex $u_0$ with  $\delta^{*}(u_0,G_1)<\delta^{*}(u_0,G_2)$. 
The fact that this relationship is the opposite of that for $v$ naturally suggests an iterative analysis, where the roles of $G_1$ and $G_2$ are exchanged in each iteration as long as Condition~\ref{indegree-changes-two-graphs-alt2} holds. 
Such an analysis leads to the following lemma, which establishes a sufficient condition for impartiality in terms of~$T$,~$t$, and~$k$. Impartiality for a particular choice of $T$ and $t$ that guarantees the bound of \autoref{thm:additive-ub} can then be obtained in a straightforward way.
\begin{lemma}
\label{lem:impartiality-additive-ub}
For every $n,k \in \NN$ with $k\in [n-1]$, the twin threshold mechanism with thresholds $T$ and $t$ such that
\[
    \frac{1}{2}(T^2+T+3t-t^2)-\min\{T-t,~ k\} > kn
\]
is impartial on $\calG_n(k)$.
\end{lemma}
The key insight for this lemma is that whenever the twin threshold mechanism with thresholds~$t$ and~$T$ is \emph{not} impartial on a pair of graphs~$G,G'$, there is a vertex~$v$ with indegree at least~$t$ in~$G$ and vertices~$w_j$ with indegree at least~$t+j$ in \emph{some} of the graphs~$G$ or~$G'$, for each value~\mbox{$j\in [T-t]_0$}. The result follows as the sum of their indegrees in~$G$ must be at least the left-hand side of the inequality in the statement and at most~$kn$. The proof of this key idea is based on the existence of vertices whose indegree change between~$G$ and~$G'$ at the point when their outgoing edges are deleted. We use \autoref{lem:indegree-changes-two-graphs} both to bound the indegrees of these vertices and their in-neighbors as described above and to show that all vertices involved are distinct. 

\begin{proof}[Proof of \autoref{lem:impartiality-additive-ub}]
Let~$n \in \NN$~$k\in [n-1]$, and~$t,T\in \NN$ be such that the twin threshold mechanism with thresholds~$t$ and~$T$ is \emph{not} impartial on~$\calG_n(k)$. 
Let \mbox{$G=(N,E)\in \calG_n(k)$} and \mbox{$G'=(N,E')\in \calG_n(k)$} be graphs where impartiality fails due to the existence of a vertex~$\tilde{v}\in N$ such that $E\setminus(\{\tilde{v}\}\times N)=E'\setminus(\{\tilde{v}\}\times N)$ and $\tilde{v}$ is selected only for one of these graphs. 
In particular, this implies that \mbox{$\delta^-(\tilde{v},G)=\delta^-(\tilde{v},G') \geq T$}. 
For~$\tilde{v}$ to be selected only for one of the graphs, there has to be a vertex~$v\in N\setminus \{\tilde{v}\}$ whose outgoing edges are deleted when the mechanism runs with input~$G$ but not when it runs with input~$G'$, or vice versa.
We suppose w.l.o.g.\ that the former holds, so denoting~$v^0=v$ we have that \mbox{$\delta^{*}(v^0,G) > t-1 \geq \delta^{*}(v^0,G')$}. 
Applying \autoref{lem:indegree-changes-two-graphs} with~$v=v^0$,~$G_1=G,\ G_2=G'$, and~$(d,z)=(t-1,v^0)$, we obtain that there are
\begin{equation}
    r^0=\delta^{*}(v^0,G)-(t-1)\label{eq:def-r0}    
\end{equation} 
vertices~$u^0_0,\ldots,u^0_{r^0-1}$ such that for every~$j \in [r^0-1]$,
\[
    (u^0_{j},v^0)\in E\cap E',\ \ \ (\delta^{*}(u^0_{j},G'),u^0_{j})\succ (\delta^{*}(v^0,G)-j,v^0),\ \ \ (\delta^{*}(u^0_{j},G),u^0_{j})\prec (\delta^{*}(v^0,G),v^0),
\]
and either Condition \ref{indegree-changes-two-graphs-alt1} or \ref{indegree-changes-two-graphs-alt2} of the lemma holds.
If Condition \ref{indegree-changes-two-graphs-alt1} holds, we denote $m=0$. Otherwise, we know that 
$(\delta^{*}(u^0_0,G'),u^0_0)\succ (\delta^{*}(v^0,G),v^0) \succ (\delta^{*}(u^0_0,G),u^0_0)$, thus we can define~$v^1=u^0_0$ and apply \autoref{lem:indegree-changes-two-graphs} with~$v=v^1$,~$G_1=G',\ G_2=G,$ and~$(d,z)=(\delta^{*}(v^0,G),v^0)$. 

The idea is to iterate this procedure until Condition \ref{indegree-changes-two-graphs-alt1} holds. For simplicity, for each~$\ell\in \NN \cup \{0\}$ we denote~$G^\ell=(N,E^\ell)$ with~$E^\ell=E$ if~$\ell$ is even and~$E^\ell=E'$ if~$\ell$ is odd. 
We consider~$v^0$,~$r^0$, and~$\{u^0_{\ell}: \ell\in [r^{\ell}-1]_0\}$ as described above. If Condition \ref{indegree-changes-two-graphs-alt1} of \autoref{lem:indegree-changes-two-graphs} holds, we define~$m=0$ and stop, omitting the procedure described in the next paragraph. Otherwise, we take~$\ell=1$.

We define~$v^{\ell}=u^{\ell-1}_0$. From the fact that Condition \ref{indegree-changes-two-graphs-alt2} of \autoref{lem:indegree-changes-two-graphs} holds for~$\ell-1$, we know that~$(\delta^{*}(v^{\ell},G^{\ell}),v^{\ell})\succ (\delta^{*}(v^{\ell-1},G^{\ell-1}),v^{\ell-1}) \succ (\delta^{*}(v^{\ell},G^{\ell-1}),v^{\ell})$. We apply \autoref{lem:indegree-changes-two-graphs} with~$v=v^{\ell}$,~$G_1=G^{\ell},\ G_2=G^{\ell-1}$, and~$(d,z)=(\delta^{*}(v^{\ell-1},G^{\ell-1}),v^{\ell-1})$, denoting 
\begin{equation}
    r^{\ell} = \delta^{*}(v^{\ell},G^{\ell}) - \delta^{*}(v^{\ell-1},G^{\ell-1}) + \chi(v^{\ell} > v^{\ell-1}).\label{eq:def-r-ell}
\end{equation}
If Condition \ref{indegree-changes-two-graphs-alt1} holds, we fix~$m=\ell$ and stop. Otherwise, we update~$\ell$ to~$\ell+1$ and repeat the procedure described in this paragraph.

We start by showing that~$m$ is finite, \ie that Condition \ref{indegree-changes-two-graphs-alt1} of \autoref{lem:indegree-changes-two-graphs} holds for some iteration. To see this, we claim that for every~$\ell$ such that Condition \ref{indegree-changes-two-graphs-alt2} holds and every~$\ell'\in [\ell-1]_0$, we have that 
\begin{equation}
    (\delta^*(u^\ell_j,G^{\ell+1}),u^\ell_j) \succ (\delta^*(u^{\ell'}_{j'},G^{\ell+1}),u^{\ell'}_{j'}) \qquad \text{for all } j\in [r^\ell-1]_0,\ j'\in [r^{\ell'}-1]_0.\label{eq:ineq-vertices-distinct-u-ell}
\end{equation}
If true, this implies that for every distinct~$\ell$ and~$\ell'$ for which Condition \ref{indegree-changes-two-graphs-alt2} holds,~$u^\ell_j \neq u^{\ell'}_{j'}$ for every~$j\in [r^\ell-1]_0$ and~$j'\in [r^{\ell'}-1]_0$. This fact also holds trivially when~$\ell=\ell'$ and~$j\neq j'$. Since~$n$ is finite, Condition \ref{indegree-changes-two-graphs-alt2} can thus hold for a finite number of iterations and~$m$ is also finite.\footnote{We prove a slightly more general version of \eqref{eq:ineq-vertices-distinct-u-ell} to show that even the vertices~$u^m_j$ with~$j\geq 1$ are distinct from all other vertices later in the proof (\autoref{claim:vertices-distinct}).}

We prove \eqref{eq:ineq-vertices-distinct-u-ell} by induction over~$\ell$, so let~$\ell$ be an iteration for which Condition \ref{indegree-changes-two-graphs-alt2} of \autoref{lem:indegree-changes-two-graphs} holds and note that, for every~$j\in [r^{\ell}-1]_0$,
\begin{align*}
    (\delta^{*}(u^{\ell}_{j},G^{\ell+1}),u^{\ell}_{j}) & \succ(\delta^{*}(v^{\ell},G^\ell)-(r^\ell-1),v^{\ell}) \\
    & = (\delta^{*}(v^{\ell-1},G^{\ell+1})+1-\chi(v^\ell>v^{\ell-1}),v^{\ell})\\
    & \succ(\delta^{*}(v^{\ell-1},G^{\ell+1}),v^{\ell-1}),
\end{align*}
where the first inequality follows from \autoref{lem:indegree-changes-two-graphs}, the equality from \eqref{eq:def-r-ell}, and the last inequality from a simple calculation.
By \autoref{lem:indegree-changes-two-graphs}, we additionally know that~$(\delta^{*}(v^{\ell-1},G^{\ell+1}),v^{\ell-1})\succ(\delta^{*}(u^{\ell-1}_{j'},G^{\ell+1}),u^{\ell-1}_{j'})$ for each~$j'\in [r^{\ell-1}-1]_0$.
We obtain
\begin{equation}
    (\delta^{*}(u^{\ell}_{j},G^{\ell+1}),u^{\ell}_{j}) \succ(\delta^{*}(u^{\ell-1}_{j'},G^{\ell+1}),u^{\ell-1}_{j'}) \qquad \text{for all } j'\in [r^{\ell-1}-1]_0.\label{eq:ineq-degree-ell-ell1-mfinite}
\end{equation}
Furthermore, if~$\ell\geq 2$ we also know from Condition \ref{indegree-changes-two-graphs-alt2} of \autoref{lem:indegree-changes-two-graphs} that
\[
(\delta^{*}(v^{\ell-1},G^{\ell+1}),v^{\ell-1}) = (\delta^{*}(u^{\ell-2}_{0},G^{\ell+1}),u^{\ell-2}_{0}) \succeq (\delta^{*}(u^{\ell-2}_{j'},G^{\ell+1}),u^{\ell-2}_{j'})
\]
for each~$j'\in [r^{\ell-2}-1]_0$.
Therefore,
\begin{equation}
    (\delta^{*}(u^{\ell}_{j},G^{\ell+1}),u^{\ell}_{j}) \succ(\delta^{*}(u^{\ell-2}_{j'},G^{\ell+1}),u^{\ell-2}_{j'}) \qquad \text{for all } j'\in [r^{\ell-2}-1]_0.\label{eq:ineq-degree-ell-ell2-mfinite}
\end{equation}
Inequalities \eqref{eq:ineq-degree-ell-ell1-mfinite} and \eqref{eq:ineq-degree-ell-ell2-mfinite} prove \eqref{eq:ineq-vertices-distinct-u-ell} for the base cases~$\ell\in \{1,2\}$, and give the necessary ingredient to prove it for larger values of~$\ell$: If we assume that \eqref{eq:ineq-vertices-distinct-u-ell} holds for some value of~$\ell$, \eqref{eq:ineq-degree-ell-ell1-mfinite} implies that it also holds for~$\ell+1$ and \eqref{eq:ineq-degree-ell-ell2-mfinite} implies that it also holds for~$\ell+2$.

In summary, the previous construction leads to the existence of vertices~$u^{\ell}_j$ for every \mbox{$\ell\in [m]_0$} and~$j\in [r^{\ell}-1]_0$ such that, for every~$\ell\in [m]_0$,
\begin{align}
    (\delta^{*}(u^{\ell}_{j},G^{\ell+1}),u^{\ell}_j) \succ (\delta^{*}(v^{\ell},G^{\ell})-j,v^{\ell}) & \ \ \text{for all } j\in \{\chi(\ell=m),\ldots,r^{\ell}-1\},\label{eq:inneighbors-lb}\\
    (\delta^{*}(u^{\ell}_{j},G^{\ell}),u^{\ell}_j) \prec (\delta^{*}(v^{\ell},G^{\ell}),v^{\ell}) & \ \ \text{for all } j\in \{\chi(\ell=m),\ldots,r^{\ell}-1\},\label{eq:inneighbors-ub}\\
    (\delta^{*}(u^{\ell}_0,G^{\ell+1}),u^{\ell}_0) \succ\dots \succ (\delta^{*}(u^{\ell}_{r^{\ell}-1},G^{\ell+1}),u^{\ell}_{r^{\ell}-1}) & \ \ \text{if } \ell\neq m.\label{eq:order-ell}
\end{align}
For~$\ell=m$, we know that~$(\tilde{v},v^m)\in E^m\setminus E^{m+1}$ and that, if~$(\delta^*(v^m,G^m),v^m)\prec (\delta^-(\tilde{v},G^m),\tilde{v})$, taking 
\begin{equation}
    \tilde{r}= \delta^-(\tilde{v},G^m)-\delta^{*}(v^m,G^m)+\chi(\tilde{v}>v^m),\label{eq:def-r-tilde}
\end{equation}
there are vertices~$\{\tilde{u}_j: j\in [\tilde{r}-1]_0\}$, none of them equal to~$\tilde{v}$, such that 
\begin{equation}
    (\delta^{*}(\tilde{u}_j,G^m),\tilde{u}_{j})\succ (\delta^-(\tilde{v},G^m)-j,\tilde{v}) \qquad \text{for all } j\in[\tilde{r}-1]_0.\label{eq:inneighbors-vtilde}
\end{equation}

The remainder of this proof aims to show the existence of a subset~$w_0,\ldots,w_{T-t}$ of the vertices defined above such that each vertex~$w_j$ has an indegree of at least~$t+j$ in some of the graphs~$G$ or~$G'$. Specifically, we want to show the following property:
\begin{equation}
    \begin{aligned} \text{There exists } \{w_j: j\in [T-t]_0\} \subseteq \{u^\ell_j: \ell\in [m]_0, j\in [r^\ell-1]_0\} \cup \{\tilde{u}_j: j\in [\tilde{r}-1]_0\} \\
    \text{s.t.\ }
    w_j\neq v^0 \text{ and }\max\{\delta^-(w_j,G_j): G_j\in \{G,G'\}\} \geq t+j \text{ for all } j\in [T-t]_0.
    \end{aligned}
    \label{eq:lem-impartiality-main}
\end{equation}
We now observe that the result we want to show immediately follows if this property holds, so let \mbox{$W=\{w_j: j\in [T-t]_0\}$} be as in Property \eqref{eq:lem-impartiality-main} and, for each~$j\in [T-t]_0$, $G_j\in \arg\max \{\delta^-(w_j,\bar{G}): \bar{G}\in \{G,G'\}\}$. We claim that
\begin{equation}
    \sum_{v\in N\setminus \{v^0\}} \delta^-(v,G) \geq \sum_{j=0}^{T-t} \delta^-(w_j,G_j)- \min\{T-t,~ k\}. \label{eq:indegrees-G-G'}
\end{equation}
To see this, note that since~$W$ is a subset of~$N\setminus \{v^0\}$ and~$G$ and~$G'$ only differ in the outgoing edges of~$\tilde{v}$, which are at most~$k$, it holds that 
\[
    \sum_{j=0}^{T-t} \delta^-(w_j,G_j) - \sum_{v\in N\setminus \{v^0\}} \delta^-(v,G) \leq \sum_{j=0}^{T-t} \delta^-(w_j,G_j) - \sum_{j=0}^{T-t} \delta^-(w_j,G) \leq k.
\]
Moreover, we know that~$\delta^-(\tilde{v},G)= \delta^-(\tilde{v},G') \geq T \geq 1$, and that~$\delta^-(w_j,G)\geq \delta^-(w_j,G_j)-1$ for every~$j\in [T-t]_0$. If~$\tilde{v}=w_{j'}$ for some~$j'\in [T-t]_0$, this yields
\[
    \sum_{v\in N\setminus \{v^0\}} \delta^-(v,G) \geq \delta^-(\tilde{v},G_{j'}) + \sum_{j\in [T-t]_0\setminus \{j'\}} (\delta^-(w_j,G_j)-1) = \sum_{j=0}^{T-t} \delta^-(w_j,G_j)-(T-t).
\]
If~$\tilde{v}\not\in W$, we have that
\[
    \sum_{v\in N\setminus \{v^0\}} \delta^-(v,G) \geq \delta^-(\tilde{v},G) + \sum_{j=0}^{T-t} (\delta^-(w_j,G_j)-1) \geq \sum_{j=0}^{T-t} \delta^-(w_j,G_j)-(T-t).
\]
This concludes the proof of \eqref{eq:indegrees-G-G'}. We obtain that
\begin{align*}
    \sum_{v\in N} \delta^-(v,G) & \geq t+ \sum_{j=0}^{T-t} \delta^-(w_j,G_j)-\min\{T-t,~k\} \\
    & \geq t+\sum_{j=t}^T j - \min\{T-t,k\}\\
    & = \frac{1}{2}(T^2+T+3t-t^2)-\min\{T-t,~ k\},
\end{align*}
where the first inequality follows from \eqref{eq:indegrees-G-G'} and the fact that~$\delta^-(v^0,G)\geq t$ and the second inequality from \eqref{eq:lem-impartiality-main}.
Since the sum of the indegrees in a graph in~$\calG_n(k)$ is at most the maximum number~$kn$ of edges, this implies that
\[
    \frac{1}{2}(T^2+T+3t-t^2)-\min\{T-t,~ k\} \leq kn
\]
and concludes the proof.

We now proceed with the proof of Property \eqref{eq:lem-impartiality-main} through a sequence of claims. 
The first one will allow us to derive a simple form of the bounds on the indegrees of the vertices $\{u^\ell_j: \ell\in [m]_0, j\in \{\chi(\ell=m),\ldots,r^\ell-1\} \}$ given by \eqref{eq:inneighbors-lb}.

\begin{claim}\label{claim:indegrees-vell}
    For every~$\ell\in [m]_0$, it holds that 
    \[
        \delta^*(v^\ell,G^\ell) = t-1+\sum_{i=0}^{\ell} r^i - \sum_{i=0}^{\ell-1} \chi(v^{i+1}>v^i).
    \]
\end{claim}

\begin{proof}
    We prove the claim by induction over~$\ell\in [m]_0$. 
    For the base case~$\ell=0$, we know from \eqref{eq:def-r0} that
    \[
        \delta^*(v^0,G^0) = t-1+r^0.
    \]
    Assume now that~$m\geq 1$ and that, for some~$\ell\in [m-1]_0$, it holds that 
    \begin{equation}
        \delta^*(v^\ell,G^\ell) = t-1+\sum_{i=0}^{\ell} r^i - \sum_{i=0}^{\ell-1} \chi(v^{i+1}>v^i).\label{eq:ineq-degrees-induction}
    \end{equation}
    This implies
    \[
        \delta^*(v^{\ell+1},G^{\ell+1}) = \delta^*(v^\ell,G^\ell) + r^{\ell+1} - \chi(v^{\ell+1}>v^\ell) = t-1+\sum_{i=0}^{\ell+1} r^i - \sum_{i=0}^{\ell} \chi(v^{i+1}>v^i),
    \]
    where the first equality follows from \eqref{eq:def-r-ell} and the second one from \eqref{eq:ineq-degrees-induction}.
\end{proof}

\autoref{claim:indegrees-vell}, along with \eqref{eq:inneighbors-lb}, implies that for every~$\ell\in [m]_0$ and every \mbox{$j\in \{\chi(\ell=m),\ldots, r^\ell-1\}$}, 
\[
    (\delta^{*}(u^{\ell}_{j},G^{\ell+1}),u^{\ell}_j) \succ \bigg(t-1+\sum_{i=0}^{\ell} r^i - \sum_{i=0}^{\ell-1} \chi(v^{i+1}>v^i) - j, v^\ell \bigg).
\]
Therefore, defining~$g \colon \NN\times \NN \to \NN$ as
\begin{equation}
    g(\ell,j) = t-1+\sum_{i=0}^{\ell} r^i - \sum_{i=0}^{\ell-1} \chi(v^{i+1}>v^i)-j,\label{eq:def-g-bound}
\end{equation}
we know that for every~$\ell\in [m]_0$ and~$j\in \{\chi(\ell=m),\ldots, r^\ell-1\}$, 
\begin{equation}
    (\delta^{*}(u^{\ell}_{j},G^{\ell+1}),u^{\ell}_j) \succ (g(\ell,j),v^\ell).\label{eq:ineq-lb-indegrees-uell}
\end{equation}
We now state some simple properties about~$g$. In particular, we show that the image via~$g$ of any~$(\ell,j)$ such that there exists a vertex~$u^\ell_j$ lies between~$t$ and~$\delta^*(v^m,G^m)-1$, and that for every~$\tau$ between these values there exists~$u^\ell_j$ such that~$g(\ell,j)=\tau$. We phrase these properties slightly more generally, replacing~$m$ in the previous statements with an arbitrary value~$\ell'\in [m]_0$.

\begin{claim}\label{claim:properties-g}
    For every~$\ell'\in [m]_0$, it holds that
    \begin{enumerate}[label=(\roman*)]
        \item $g(0,r^0-1) =t$ and~$g(\ell',\chi(\ell'=m)) = \delta^*(v^{\ell'},G^{\ell'})-\chi(\ell'=m)$;\label{claim:properties-g-i}
        \item for every~$\tau\in \{t,t+1,\ldots,\delta^*(v^{\ell'},G^{\ell'})-\chi(\ell'=m)$\}, there exist values~$\ell\in [\ell']_0$ and \mbox{$j\in \{\chi(\ell=m),\ldots,r^{\ell}-1\}$} such that~$g(\ell,j)=\tau$.\label{claim:properties-g-ii}
    \end{enumerate}
\end{claim}

\begin{proof}
    The equality~$g(0,r^0-1) =t$ follows immediately from \eqref{eq:def-g-bound}.
    The other equality of Part~\ref{claim:properties-g-i} also follows from \eqref{eq:def-g-bound}, since for every~$\ell'\in [m]_0$ we have
    \begin{align*}
        g(\ell',\chi(\ell'=m)) & = t-1+ \sum_{i=1}^{\ell'} (\delta^{*}(v^{i},G^{i}) - \delta^{*}(v^{i-1},G^{i-1}) + \chi(v^{i} > v^{i-1})) + {} \\
        & \phantom{{} = {}} \delta^{*}(v^0,G)-(t-1) - \sum_{i=0}^{\ell'-1} \chi(v^{i+1}>v^i)-\chi(\ell'=m)\\
        & = \delta^*(v^{\ell'},G^{\ell'})-\chi(\ell'=m),
    \end{align*}
    where we have used \eqref{eq:def-r0} and \eqref{eq:def-r-ell}. This concludes the proof of Part~\ref{claim:properties-g-i}.

    To prove Part~\ref{claim:properties-g-ii}, we actually show the following two properties:
    \begin{align}
        g(\ell,j) & = g(\ell,j+1)+1 \qquad\text{for all } \ell\in [m]_0,\ j\in [\chi(\ell=m),\ldots,r^\ell-2\},\label{eq:g-growth-same}\\
         g(\ell,r^\ell-1) & \in \{g(\ell-1,0), g(\ell-1,0)+1\} \qquad \text{for all } \ell\in [m-\chi(r^m=1)].\label{eq:g-growth-diff} 
    \end{align}
    These directly imply Part~\ref{claim:properties-g-ii}, as they state that~$g$ grows one unit at a time when keeping~$\ell$ constant and decreasing~$j$ by one, and that it outputs the same value or increases in one unit when switching from~$(\ell-1,0)$ to~$(\ell,r^\ell-1)$. Property \eqref{eq:g-growth-same} is again immediate from \eqref{eq:def-g-bound}.
    To see \eqref{eq:g-growth-diff} we observe that, for~$\ell\in [m-\chi(r^m=1)]$,
    \begin{align*}
        g(\ell,r^\ell-1) & = t+\sum_{i=0}^{\ell-1} r^i - \sum_{i=0}^{\ell-1} \chi(v^{i+1}>v^i),\quad 
        g(\ell-1,0) = t-1+\sum_{i=0}^{\ell-1} r^i - \sum_{i=0}^{\ell-2} \chi(v^{i+1}>v^i).
    \end{align*}
    Thus,~$g(\ell,r^\ell-1) = g(\ell-1,0)$ if~$v^\ell>v^{\ell-1}$ and~$g(\ell,r^\ell-1) = g(\ell-1,0)+1$ otherwise.
\end{proof}

The previous claim implies that the bounds given by~$g(\ell,j)$ cover all values between~$t$ and~$\delta^*(v^m,G^m)-1$. However, to prove Property \eqref{eq:lem-impartiality-main} we need, in addition, that the vertices with these bounds on their indegrees are distinct. The next claim states that, indeed, all vertices involved in the analysis, besides~$\tilde{v}=u^m_0$, are different from each other. This generalizes the fact that vertices~$u^\ell_j,\ u^{\ell'}_{j'}$ for~$\ell,\ell'\in [m-1]_0$ are different from each other, implied by \eqref{eq:ineq-vertices-distinct-u-ell}.
\begin{claim}\label{claim:vertices-distinct}
    For every~$\ell,\ell'\in [m]$,~$j\in \{\chi(\ell=m),\ldots,r^\ell-1\}$, and~$j'\in \{\chi(\ell'=m),\ldots,r^{\ell'}-1\}$ with~$(\ell,j)\neq (\ell',j')$, it holds that
    \[
        v^0 \neq u^\ell_j \neq u^{\ell'}_{j'} \neq \tilde{u}_j \neq v^0.
    \]
\end{claim}

\begin{proof}
We will show that, for every~$\ell\in [m]_0$ and~$\ell'\in [\ell-1]_0$, it holds that
    \begin{enumerate}[label=(\roman*)]
        \item $(\delta^*(u^\ell_j,G^{\ell+1}),u^\ell_j) \succ (\delta^*(u^{\ell'}_{j'},G^{\ell+1}),u^{\ell'}_{j'})$ for all~$j\in \{\chi(\ell=m),\ldots,r^\ell-1\}$,~$j'\in [r^{\ell'}-1]_0$;\label{claim:vertices-distinct-i}
        \item $(\delta^*(u^\ell_j,G^{\ell+1}),u^\ell_j) \succ (\delta^*(v^0,G^{\ell+1}),v^0)$ for all~$j\in \{\chi(\ell=m),\ldots,r^\ell-1\}$;\label{claim:vertices-distinct-ii}
        \item $(\delta^*(\tilde{u}_j,G^m),\tilde{u}_j) \succ (\delta^*(u^{\ell}_{j'},G^{m}),u^{\ell}_{j'})$ for all~$j\in [\tilde{r}-1]_0, j'\in \{\chi(\ell=m),\ldots,r^\ell-1\}$.\label{claim:vertices-distinct-iii}
    \end{enumerate}
These three inequalities, along with the fact that~$u^\ell_j \neq u^\ell_{j'}$ for every~$\ell \in [m]_0$ and every \mbox{$j,j'\in [r^\ell-1]_0$} with~$j\neq j'$, directly imply that all these vertices are distinct, as claimed in the statement.

The proof of Part~\ref{claim:vertices-distinct-i} is completely analogous to that of \eqref{eq:ineq-vertices-distinct-u-ell}, but now considering the case~$\ell=m$. We proceed by induction over~$\ell$, so let~$\ell\in[m]$ be an integer and note that, for every~$j\in \{\chi(\ell=m),\ldots,r^{\ell}-1\}$,
\begin{align*}
    (\delta^{*}(u^{\ell}_{j},G^{\ell+1}),u^{\ell}_{j}) & \succ(\delta^{*}(v^{\ell},G^\ell)-(r^\ell-1),v^{\ell}) \\
    & = (\delta^{*}(v^{\ell-1},G^{\ell+1})+1-\chi(v^\ell>v^{\ell-1}),v^{\ell})\\
    & \succ(\delta^{*}(v^{\ell-1},G^{\ell+1}),v^{\ell-1}),
\end{align*}
where the first inequality follows from \eqref{eq:inneighbors-lb}, the equality from \eqref{eq:def-r-ell}, and the last inequality from a simple calculation.
By \eqref{eq:inneighbors-ub}, we also know that~$(\delta^{*}(v^{\ell-1},G^{\ell+1}),v^{\ell-1})\succ(\delta^{*}(u^{\ell-1}_{j'},G^{\ell+1}),u^{\ell-1}_{j'})$ for each~$j'\in [r^{\ell-1}-1]_0$.
We obtain
\begin{equation}
    (\delta^{*}(u^{\ell}_{j},G^{\ell+1}),u^{\ell}_{j}) \succ(\delta^{*}(u^{\ell-1}_{j'},G^{\ell+1}),u^{\ell-1}_{j'}) \qquad \text{for every } j'\in [r^{\ell-1}-1]_0.\label{eq:ineq-degree-ell-ell1}
\end{equation}
Furthermore, if~$\ell\geq 2$ we also know from \eqref{eq:order-ell} that
\[
(\delta^{*}(v^{\ell-1},G^{\ell+1}),v^{\ell-1}) = (\delta^{*}(u^{\ell-2}_{0},G^{\ell+1}),u^{\ell-2}_{0}) \succeq (\delta^{*}(u^{\ell-2}_{j'},G^{\ell+1}),u^{\ell-2}_{j'})
\]
for each~$j'\in [r^{\ell-2}-1]_0$.
Therefore,
\begin{equation}
    (\delta^{*}(u^{\ell}_{j},G^{\ell+1}),u^{\ell}_{j}) \succ(\delta^{*}(u^{\ell-2}_{j'},G^{\ell+1}),u^{\ell-2}_{j'}) \qquad \text{for every } j'\in [r^{\ell-2}-1]_0.\label{eq:ineq-degree-ell-ell2}
\end{equation}
Inequalities \eqref{eq:ineq-degree-ell-ell1} and \eqref{eq:ineq-degree-ell-ell2} prove Part~\ref{claim:vertices-distinct-i} for the base cases~$\ell\in \{1,2\}$, and give the necessary ingredient to prove it for larger values of~$\ell$: If we assume that Part~\ref{claim:vertices-distinct-i} holds for some value of~$\ell\in [m]$, \eqref{eq:ineq-degree-ell-ell1} implies that it also holds for~$\ell+1$ and \eqref{eq:ineq-degree-ell-ell2} implies that it also holds for~$\ell+2$.

To prove Part~\ref{claim:vertices-distinct-ii}, we let~$\ell\in [m]_0$ and~$j\in \{\chi(\ell=m),\ldots,r^\ell-1\}$ be integers, and we distinguish whether~$\ell$ is even or odd. If~$\ell$ is even, we have that
\[
    \delta^*(u^\ell_j,G^{\ell+1}) \geq t > \delta^*(v^0,G^{\ell+1}),
\]
where the first inequality follows from \eqref{eq:ineq-lb-indegrees-uell}, \eqref{eq:g-growth-same}, and \eqref{eq:g-growth-diff} along with Part~\ref{claim:properties-g-i} of \autoref{claim:properties-g}, and the second one from the initial assumption on the indegree of~$v^0$. This immediately implies that~$(\delta^*(u^\ell_j,G^{\ell+1}),u^\ell_j) \succ (\delta^*(v^0,G^{\ell+1}),v^0)$, as claimed. If~$\ell=1$, we have that
\begin{align}
    (\delta^*(u^1_j,G^{2}),u^1_j)
    & \succeq (\delta^*(v^1,G^1)-(r^1-1),v^1)\nonumber \\
    & = (\delta^*(v^0,G^0)+1-\chi(v^1>v^0),v^1)\nonumber \\
    & \succ (\delta^*(v^0,G^{0}),v^0),\label{eq:ineq-lb-u1}
\end{align}
where the first inequality comes from \eqref{eq:inneighbors-lb}, the equality from the definition of~$r^1$, and the last inequality from a simple calculation. Finally, if~$\ell$ is odd and~$\ell\geq 1$, we have that
\[
    (\delta^*(u^\ell_j,G^{\ell+1}),u^\ell_j) \succ (\delta^*(u^1_0,G^{\ell+1}),u^1_0) \succ (\delta^*(v^0,G^{\ell+1}),v^0),
\]
where the first inequality follows from Part~\ref{claim:vertices-distinct-i} and the second one from \eqref{eq:ineq-lb-u1}.

In order to prove Part~\ref{claim:vertices-distinct-iii}, we observe that, for every~$j\in[\tilde{r}-1]_0$,
\begin{align}
    (\delta^{*}(\tilde{u}_j,G^m), \tilde{u}_j) & \succ (\delta^-(\tilde{v},G^m)-(\tilde{r}-1),\tilde{v}) \nonumber \\
    & = (\delta^{*}(v^m,G^m)+1-\chi(\tilde{v}>v^m), \tilde{v}) \nonumber \\
    & \succ (\delta^{*}(v^m,G^m), v^m),\label{eq:ineq-utilde}
\end{align}
where the first inequality follows from \eqref{eq:inneighbors-vtilde}, the equality from \eqref{eq:def-r-tilde}, and the last inequality from a simple calculation. This yields
\[
    (\delta^{*}(\tilde{u}_j,G^m), \tilde{u}_j) \succ (\delta^{*}(v^m,G^m), v^m) \succ (\delta^{*}(u^{m}_{j'},G^{\ell}),u^{m}_{j'})
\]
for every~$j'\in [r^m-1]$, where the first inequality follows from \eqref{eq:ineq-utilde} and the second one from \eqref{eq:inneighbors-ub}. This implies the result for~$\ell=m$. For~$\ell\in [m-1]_0$, we have that
\[
    (\delta^{*}(\tilde{u}_j,G^m), \tilde{u}_j) \succ (\delta^{*}(v^m,G^m), v^m) = (\delta^{*}(u^{m-1}_0,G^m), u^{m-1}_0) \succeq (\delta^{*}(u^{\ell}_{j'},G^m), u^{\ell}_{j'})
\]
for every~$j'\in [r^\ell-1]_0$,
where the first inequality follows from \eqref{eq:ineq-utilde}, the equality from the definition of~$v^m$, and the last inequality from \eqref{eq:order-ell}, if~$\ell=m-1$, and from Part~\ref{claim:vertices-distinct-i}, otherwise.
\end{proof}

\autoref{claim:properties-g} and \autoref{claim:vertices-distinct}, along with \eqref{eq:ineq-lb-indegrees-uell}, imply that for every~$\ell'\in [m]_0$ it is possible to take a subset of~$\{u^\ell_j: \ell\in [\ell']_0, j\in \{\chi(\ell=m),\ldots,r^\ell-1\}\}$ such that some vertex in the subset has a (unique) lower bound on its indegree equal to~$\tau$, for every~$\tau$ between~$t$ and~$\delta^*(v^{\ell'},G^{\ell'})-\chi(\ell'=m)$. We state this conveniently in the following claim.

\begin{claim}\label{claim:levels-covered}
    Let~$\ell'\in [m]_0$ be an integer. Then, there exists a subset
    \[
        \{z_j: j\in [\delta^*(v^{\ell'},G^{\ell'})-\chi(\ell'=m)-t]_0\} \subseteq \{u^\ell_j: \ell\in [\ell']_0, j\in \{\chi(\ell=m),\ldots,r^\ell-1\}\}
    \]
    such that~$\max\{\delta^*(z_j,G_j): G_j\in \{G,G'\}\} \geq t+j$ for all~$j\in [\delta^*(v^{\ell'},G^{\ell'})-\chi(\ell'=m)-t]_0$.
\end{claim}

\begin{proof}
    Let~$\ell'\in [m]_0$ be an integer, and consider the set
    \[
        U = \{u^\ell_j: \ell\in [\ell']_0, j\in \{\chi(\ell=m),\ldots,r^\ell-1\}\}.
    \]
    Inequality \eqref{eq:ineq-lb-indegrees-uell} and \autoref{claim:properties-g} imply that, for every~$\tau\in \{t,t+1,\ldots,\delta^*(v^{\ell'},G^{\ell'})-\chi(\ell'=m)\}$, there is a vertex~$u^\ell_j$ in~$U$ such that~$\delta^*(u^\ell_j,G^{\ell+1}) \geq g(\ell,j)=\tau$. Since all vertices in~$U$ are distinct due to \autoref{claim:vertices-distinct}, we conclude that it is possible to take a subset 
    \[
        \{z_j: j\in [\delta^*(v^{\ell'},G^{\ell'})-\chi(\ell'=m)-t]_0\} \subseteq U
    \]
    such that~$\max\{\delta^*(z_j,G_j): G_j\in \{G,G'\}\} \geq t+j$ for all~$j\in [\delta^*(v^{\ell'},G^{\ell'})-\chi(\ell'=m)-t]_0$.
\end{proof}

We now distinguish two cases to conclude Property \eqref{eq:lem-impartiality-main}. We first consider the case with $\tilde{v} \notin \{u^\ell_j: \ell\in [m]_0, j\in \{\chi(\ell=m),\ldots,r^\ell-1\}\}$. Let \mbox{$Z=\{z_j: j\in [\delta^*(v^{m},G^{m})-1-t]_0\}$} be the set defined in \autoref{claim:levels-covered} for~$\ell'=m$, so that
\begin{equation}
    \max\{\delta^*(z_j,G_j): G_j\in \{G,G'\}\} \geq t+j \qquad \text{for all } j\in [\delta^*(v^{m},G^{m})-1-t]_0.\label{eq:levels-covered-m}
\end{equation}
Note that, by \autoref{claim:vertices-distinct},~$v^0\notin Z$, so it suffices to prove the bound on the indegrees in Property \eqref{eq:lem-impartiality-main}. If~$\delta^*(v^m,G^m)\geq T+1$, Property \eqref{eq:lem-impartiality-main} follows immediately from \eqref{eq:levels-covered-m}. If~$\delta^*(v^m,G^m)=T$, Property \eqref{eq:lem-impartiality-main} follows by \eqref{eq:levels-covered-m} along with the facts that~$\tilde{v}\notin Z$ and~$\delta^-(\tilde{v},G^m) \geq T$. If~$\delta^*(v^m,G^m)\leq T-1$, we know from \eqref{eq:inneighbors-vtilde} that~$\delta^{*}(\tilde{u}_j,G^m) \geq \delta^-(\tilde{v},G^m)-j$ for each~$j\in[\tilde{r}-1]_0$, hence
\[
    \delta^{*}(\tilde{u}_j,G^m) \geq T-j \qquad \text{for all } j\in[T-\delta^{*}(v^m,G^m)-1]_0,
\]
where we used that~$\delta^-(\tilde{v},G^m) \geq T$.
Property \eqref{eq:lem-impartiality-main} thus follows from \eqref{eq:levels-covered-m} and the fact that vertices in~$\{\tilde{u}_j: j\in[T-\delta^{*}(v^m,G^m)-1]_0\}$ do not belong to~$Z$ due to \autoref{claim:vertices-distinct}.

We finally consider the case where~$\tilde{v} \in \{u^\ell_j: \ell\in [m]_0, j\in \{\chi(\ell=m),\ldots,r^\ell-1\}\}$. We let 
\[
    \ell' = \min\{ \ell\in [m]_0: u^\ell_j=\tilde{v} \text{ for some } j\in \{\chi(\ell=m),\ldots,r^\ell-1\} \}
\]
and~$j'\in \{\chi(\ell'=m),\ldots,r^{\ell'}-1\}$ be such that~$\tilde{v}=u^{\ell'}_{j'}$.
We observe that~$\ell'\neq m$, since~$\tilde{v}=u^m_0$ and vertices in~$\{u^m_j: j\in [r^m-1]_0\}$ are all distinct. Thus, 
\eqref{eq:inneighbors-ub} implies that
\begin{equation}
    (\delta^*(v^{\ell'},G^{\ell'}), v^{\ell'}) \succ (\delta^*(\tilde{v},G^{\ell'}), \tilde{v}).\label{eq:ineq-vtilde-u0}
\end{equation}
We now let~$Z=\{z_j: j\in [\delta^*(v^{\ell'},G^{\ell'})-t]_0\}$ be the set defined in \autoref{claim:levels-covered} for this value of~$\ell'$, so that
\begin{equation}
    \max\{\delta^*(z_j,G_j): G_j\in \{G,G'\}\} \geq t+j \qquad \text{for all } j\in [\delta^*(v^{\ell'},G^{\ell'})-t]_0.\label{eq:levels-covered-ell}
\end{equation}
If~$\delta^*(v^{\ell'},G^{\ell'})\geq T$, Property \eqref{eq:lem-impartiality-main} follows immediately. Otherwise, from \eqref{eq:ineq-vtilde-u0} and the fact that $\delta^-(\tilde{v},G^{\ell'}) \geq T$ we have that
\[
    (\delta^-(\tilde{v},G^{\ell'}), \tilde{v}) \succ (\delta^*(v^{\ell'},G^{\ell'}), v^{\ell'}) \succ (\delta^*(\tilde{v},G^{\ell'}), \tilde{v}),
\]
so we can use \autoref{lem:indegree-changes} with~$(d,z) = (\delta^*(v^{\ell'},G^{\ell'}), v^{\ell'})$. This implies the existence of vertices~$\{y_j: j\in [q-1]_0\}$ such that 
\begin{equation}
    (\delta^{*}(y_j,G^{\ell'}),y_j) \succ (\delta^-(\tilde{v},G^{\ell'})-j, \tilde{v}) \qquad \text{for every } j\in[q-1]_0,\label{eq:ineq-indegrees-yj}
\end{equation}
where~$q=\delta^-(\tilde{v},G^{\ell'})-\delta^*(v^{\ell'},G^{\ell'})+\chi(\tilde{v}>v^{\ell'})$. In particular, since~$\delta^-(\tilde{v},G^{\ell'}) \geq T$, this yields
\begin{equation}
    \delta^{*}(y_j,G^{\ell'}) \geq T-j \qquad \text{for every } j\in[T-\delta^*(v^{\ell'},G^{\ell'})-1]_0.\label{eq:ineq-indegrees-yj-mod}
\end{equation}
We finally claim that
\begin{equation}
    (\delta^{*}(y_j,G^{\ell'}),y_j) \succ (\delta^{*}(u^{\ell}_{j'},G^{\ell'}),u^{\ell}_{j'}) \ \ \ \text{for every } j\in [q-1]_0,\ \ell \in [\ell']_0,\ j'\in [r^{\ell'}-1]_0.\label{eq:claim-indegrees-yj}
\end{equation} 
If true, this implies that~$\{y_j: j\in [q-1]_0\}\cap Z = \emptyset$, thus Property \eqref{eq:lem-impartiality-main} follows once again by combining \eqref{eq:ineq-indegrees-yj-mod} and \eqref{eq:levels-covered-ell}.

We finish the proof by showing \eqref{eq:claim-indegrees-yj}. For every~$j\in [q-1]_0$, we have that
\begin{align}
    (\delta^{*}(y_j,G^{\ell'}),y_j) & \succ (\delta^-(\tilde{v},G^{\ell'})-(q-1),\tilde{v})\nonumber\\
    & = (\delta^*(v^{\ell'},G^{\ell'})+1-\chi(\tilde{v}>v^{\ell'}),\tilde{v}) \nonumber\\
    & \succ (\delta^*(v^{\ell'},G^{\ell'}), v^{\ell'})\label{eq:indeq-indegrees-yj-vl},
\end{align}
where the first inequality comes from \eqref{eq:ineq-indegrees-yj}, the equality from the definition of~$q$, and the last inequality from a simple calculation.
From \eqref{eq:inneighbors-ub}, we know that
\[
    (\delta^*(v^{\ell'},G^{\ell'}), v^{\ell'}) \succ (\delta^{*}(u^{\ell'}_{j'},G^{\ell'}),u^{\ell'}_{j'}) \qquad \text{for every } j'\in [r^{\ell'}-1]_0.
\]
Combining this inequality with \eqref{eq:indeq-indegrees-yj-vl} shows \eqref{eq:claim-indegrees-yj} for the case where~$\ell=\ell'$. 
If~$\ell'\geq 1$, from the fact that~$v^{\ell'}=u^{\ell'-1}_0$ and \eqref{eq:order-ell} we know that 
\[
    (\delta^*(v^{\ell'},G^{\ell'}), v^{\ell'}) = (\delta^*(u^{\ell'-1}_0,G^{\ell'}), u^{\ell'-1}_0) \succeq (\delta^{*}(u^{\ell'-1}_{j'},G^{\ell'}),u^{\ell'-1}_{j'}) \ \ \text{for every } j'\in [r^{\ell'-1}-1]_0.
\]
Combining this inequality with \eqref{eq:indeq-indegrees-yj-vl} shows \eqref{eq:claim-indegrees-yj} for the case where~$\ell=\ell'-1$. Finally, if~$\ell'\geq 2$, from the fact that~$v^{\ell'}=u^{\ell'-1}_0$ and Part~\ref{claim:vertices-distinct-i} of \autoref{claim:vertices-distinct}, we know that
\[
    (\delta^*(v^{\ell'},G^{\ell'}), v^{\ell'}) = (\delta^*(u^{\ell'-1}_0,G^{\ell'}), u^{\ell'-1}_0) \succeq (\delta^{*}(u^{\ell}_{j'},G^{\ell'}),u^{\ell}_{j'})
\]
for every~$\ell\in [\ell'-1]_0$ and~$j'\in [r^{\ell}-1]_0$.
Combining this inequality with \eqref{eq:indeq-indegrees-yj-vl} shows \eqref{eq:claim-indegrees-yj} for the case where~$\ell\leq\ell'-2$. This concludes the proof of Property \eqref{eq:lem-impartiality-main} and the proof of the lemma.
\end{proof}

\autoref{fig:lemma-3} illustrates this result by showing a situation where the outgoing edge of a vertex $\tilde{v}$ with $\delta^-(\tilde{v},G_1)=\delta^-(\tilde{v},G_2)=T$ determines whether $\delta^{*}(v)\geq t$ for another vertex~$v$, and thus whether $\tilde{v}$ itself is selected or not. Note that in the example there exist vertices $w_j$ such that $\delta^-(w_j,G)\geq t+j$ for every $j\in [T-t]_0$ and some $G\in \{G_1,G_2\}$. 

\begin{figure}[t]
\centering
\begin{tikzpicture}[scale=0.88]

\Vertex[x=1, y=2.25, Math, shape=circle, color=black, , size=.05, label=u^3_1, fontscale=1, position=above, distance=-.1cm]{A}
\Vertex[x=1.5, y=2.25, Math, shape=circle, color=black, size=.05, label=\tilde{v}, fontscale=1, position=above, distance=-.07cm]{B}
\Vertex[x=.5, y=1.5, Math, shape=circle, color=black, size=.05, label=u^2_1, fontscale=1, position=above left, distance=-.17cm]{C}
\Vertex[y=.75, Math, shape=circle, color=black, size=.05, label=u^1_1, fontscale=1, position=above left, distance=-.17cm]{D}
\Vertex[x=2, y=.75, Math, shape=circle, color=black, size=.05, label=v, fontscale=1, position=above right, distance=-.16cm]{E}

\Edge[Direct, color=black, lw=1pt](B)(A)
\Edge[Direct, color=black, lw=1pt](A)(C)
\Edge[Direct, color=black, lw=1pt](C)(D)
\Edge[Direct, color=black, lw=1pt](D)(E)
\Edge[Direct, color=black, lw=1pt](E)(B)

\draw[] (-.2,-.2) -- (2.2,-.2);
\draw[] (-.2,.55) -- (2.2,.55);
\draw[] (-.2,1.3) -- (2.2,1.3);
\draw[] (-.2,2.05) -- (2.2,2.05);
\draw[] (-.2,2.8) -- (2.2,2.8);

\Vertex[x=4.1, y=2.25, Math, shape=circle, color=black, , size=.05, label=u^3_1, fontscale=1, position=above, distance=-.1cm]{F}
\Vertex[x=4.6, y=2.25, Math, shape=circle, color=black, size=.05, label=\tilde{v}, fontscale=1, position=above, distance=-.07cm]{G}
\Vertex[x=3.6, y=.75, Math, shape=circle, color=black, size=.05, label=u^2_1, fontscale=1, position=above, distance=-.1cm]{H}
\Vertex[x=3.1, y=.75, Math, shape=circle, color=black, size=.05, label=u^1_1, fontscale=1, position=above, distance=-.1cm]{I}
\Vertex[x=5.1, y=.75, Math, shape=circle, color=black, size=.05, label=v, fontscale=1, position=above right, distance=-.16cm]{J}

\Edge[Direct, color=black, lw=1pt](G)(F)
\Edge[Direct, color=black, lw=1pt](H)(I)
\Edge[Direct, color=black, lw=1pt, bend=-25](I)(J)
\Edge[Direct, color=black, lw=1pt](J)(G)

\draw[] (2.9,-.2) -- (5.3,-.2);
\draw[] (2.9,.55) -- (5.3,.55);
\draw[] (2.9,1.3) -- (5.3,1.3);
\draw[] (2.9,2.05) -- (5.3,2.05);
\draw[] (2.9,2.8) -- (5.3,2.8);

\Vertex[x=7.2, y=1.5, Math, shape=circle, color=black, , size=.05, label=u^3_1, fontscale=1, position=above, distance=-.1cm]{K}
\Vertex[x=7.7, y=2.25, Math, shape=circle, color=black, size=.05, label=\tilde{v}, fontscale=1, position=above, distance=-.07cm]{L}
\Vertex[x=6.7, y=.75, Math, shape=circle, color=black, size=.05, label=u^2_1, fontscale=1, position=above, distance=-.1cm]{M}
\Vertex[x=6.2, y=.75, Math, shape=circle, color=black, size=.05, label=u^1_1, fontscale=1, position=above, distance=-.1cm]{N}
\Vertex[x=8.2, y=.75, Math, shape=circle, color=black, size=.05, label=v, fontscale=1, position=above right, distance=-.16cm]{O}

\Edge[Direct, color=black, lw=1pt](M)(N)
\Edge[Direct, color=black, lw=1pt, bend=-25](N)(O)
\Edge[Direct, color=black, lw=1pt](O)(L)

\draw[] (6,-.2) -- (8.4,-.2);
\draw[] (6,.55) -- (8.4,.55);
\draw[] (6,1.3) -- (8.4,1.3);
\draw[] (6,2.05) -- (8.4,2.05);
\draw[] (6,2.8) -- (8.4,2.8);

\Vertex[x=10.3, y=1.5, Math, shape=circle, color=black, , size=.05, label=u^3_1, fontscale=1, position=above, distance=-.1cm]{P}
\Vertex[x=10.8, y=2.25, Math, shape=circle, color=black, size=.05, label=\tilde{v}, fontscale=1, position=above, distance=-.07cm]{Q}
\Vertex[x=9.8, y=.75, Math, shape=circle, color=black, size=.05, label=u^2_1, fontscale=1, position=above, distance=-.1cm]{R}
\Vertex[x=9.3, y=.75, Math, shape=circle, color=black, size=.05, label=u^1_1, fontscale=1, position=above, distance=-.1cm]{S}
\Vertex[x=11.3, Math, shape=circle, color=black, size=.05, label=v, fontscale=1, position=above right, distance=-.16cm]{T}

\Edge[Direct, color=black, lw=1pt](R)(S)
\Edge[Direct, color=black, lw=1pt](T)(Q)

\draw[] (9.1,-.2) -- (11.5,-.2);
\draw[] (9.1,.55) -- (11.5,.55);
\draw[] (9.1,1.3) -- (11.5,1.3);
\draw[] (9.1,2.05) -- (11.5,2.05);
\draw[] (9.1,2.8) -- (11.5,2.8);

\Vertex[x=13.4, y=1.5, Math, shape=circle, color=black, , size=.05, label=u^3_1, fontscale=1, position=above, distance=-.1cm]{U}
\Vertex[x=13.9, y=2.25, Math, shape=circle, color=white, size=.05, label=\tilde{v}, fontscale=1, position=above, distance=-.07cm]{V}
\Vertex[x=12.9, y=.75, Math, shape=circle, color=black, size=.05, label=u^2_1, fontscale=1, position=above, distance=-.1cm]{W}
\Vertex[x=12.4, Math, shape=circle, color=black, size=.05, label=u^1_1, fontscale=1, position=above, distance=-.1cm]{X}
\Vertex[x=14.4, Math, shape=circle, color=black, size=.05, label=v, fontscale=1, position=above right, distance=-.16cm]{Y}

\Edge[Direct, color=black, lw=1pt](Y)(V)

\draw[] (12.2,-.2) -- (14.6,-.2);
\draw[] (12.2,.55) -- (14.6,.55);
\draw[] (12.2,1.3) -- (14.6,1.3);
\draw[] (12.2,2.05) -- (14.6,2.05);
\draw[] (12.2,2.8) -- (14.6,2.8);

\Text[x=-.7, y=-.6, fontsize=\scriptsize]{$G_1$}
\Text[x=-.8, y=.925, fontsize=\scriptsize]{$t$}
\Text[x=-.8, y=2.425, fontsize=\scriptsize]{$T$}

\Text[x=1, y=-.6, fontsize=\scriptsize]{$i=0$}
\Text[x=4.1, y=-.6, fontsize=\scriptsize]{$i=1$}
\Text[x=7.2, y=-.6, fontsize=\scriptsize]{$i=2$}
\Text[x=10.3, y=-.6, fontsize=\scriptsize]{$i=3$}
\Text[x=13.4, y=-.6, fontsize=\scriptsize]{$i=4$}

\Vertex[x=1, y=-3, Math, shape=circle, color=black, , size=.05, label=u^3_1, fontscale=1, position=above, distance=-.1cm]{AA}
\Vertex[x=1.5, y=-2.25, Math, shape=circle, color=black, size=.05, label=\tilde{v}, fontscale=1, position=above, distance=-.07cm]{AB}
\Vertex[x=.5, y=-3, Math, shape=circle, color=black, size=.05, label=u^2_1, fontscale=1, position=above, distance=-.1cm]{AC}
\Vertex[y=-3.75, Math, shape=circle, color=black, size=.05, label=u^1_1, fontscale=1, position=above left, distance=-.17cm]{AD}
\Vertex[x=2, y=-3.75, Math, shape=circle, color=black, size=.05, label=v, fontscale=1, position=above right, distance=-.16cm]{AE}

\Edge[Direct, color=black, lw=1pt](AA)(AC)
\Edge[Direct, color=black, lw=1pt](AC)(AD)
\Edge[Direct, color=black, lw=1pt](AD)(AE)
\Edge[Direct, color=black, lw=1pt](AE)(AB)

\draw[] (-.2,-4.7) -- (2.2,-4.7);
\draw[] (-.2,-3.95) -- (2.2,-3.95);
\draw[] (-.2,-3.2) -- (2.2,-3.2);
\draw[] (-.2,-2.45) -- (2.2,-2.45);
\draw[] (-.2,-1.7) -- (2.2,-1.7);

\Vertex[x=4.1, y=-3, Math, shape=circle, color=black, , size=.05, label=u^3_1, fontscale=1, position=above, distance=-.1cm]{AF}
\Vertex[x=4.6, y=-2.25, Math, shape=circle, color=black, size=.05, label=\tilde{v}, fontscale=1, position=above, distance=-.07cm]{AG}
\Vertex[x=3.6, y=-3, Math, shape=circle, color=black, size=.05, label=u^2_1, fontscale=1, position=above, distance=-.1cm]{AH}
\Vertex[x=3.1, y=-4.5, Math, shape=circle, color=black, size=.05, label=u^1_1, fontscale=1, position=above, distance=-.1cm]{AI}
\Vertex[x=5.1, y=-3.75, Math, shape=circle, color=black, size=.05, label=v, fontscale=1, position=above right, distance=-.16cm]{AJ}

\Edge[Direct, color=black, lw=1pt](AF)(AH)
\Edge[Direct, color=black, lw=1pt](AI)(AJ)
\Edge[Direct, color=black, lw=1pt](AJ)(AG)

\draw[] (2.9,-4.7) -- (5.3,-4.7);
\draw[] (2.9,-3.95) -- (5.3,-3.95);
\draw[] (2.9,-3.2) -- (5.3,-3.2);
\draw[] (2.9,-2.45) -- (5.3,-2.45);
\draw[] (2.9,-1.7) -- (5.3,-1.7);

\Vertex[x=7.2, y=-3, Math, shape=circle, color=black, , size=.05, label=u^3_1, fontscale=1, position=above, distance=-.1cm]{AK}
\Vertex[x=7.7, y=-2.25, Math, shape=circle, color=black, size=.05, label=\tilde{v}, fontscale=1, position=above, distance=-.07cm]{AL}
\Vertex[x=6.7, y=-3.75, Math, shape=circle, color=black, size=.05, label=u^2_1, fontscale=1, position=above, distance=-.1cm]{AM}
\Vertex[x=6.2, y=-4.5, Math, shape=circle, color=black, size=.05, label=u^1_1, fontscale=1, position=above, distance=-.1cm]{AN}
\Vertex[x=8.2, y=-3.75, Math, shape=circle, color=black, size=.05, label=v, fontscale=1, position=above right, distance=-.16cm]{AO}

\Edge[Direct, color=black, lw=1pt](AN)(AO)
\Edge[Direct, color=black, lw=1pt](AO)(AL)

\draw[] (6,-4.7) -- (8.4,-4.7);
\draw[] (6,-3.95) -- (8.4,-3.95);
\draw[] (6,-3.2) -- (8.4,-3.2);
\draw[] (6,-2.45) -- (8.4,-2.45);
\draw[] (6,-1.7) -- (8.4,-1.7);

\Vertex[x=10.3, y=-3, Math, shape=circle, color=black, , size=.05, label=u^3_1, fontscale=1, position=above, distance=-.1cm]{AP}
\Vertex[x=10.8, y=-3, Math, shape=circle, color=black, size=.05, label=\tilde{v}, fontscale=1, position=above, distance=-.07cm]{AQ}
\Vertex[x=9.8, y=-3.75, Math, shape=circle, color=black, size=.05, label=u^2_1, fontscale=1, position=above, distance=-.1cm]{AR}
\Vertex[x=9.3, y=-4.5, Math, shape=circle, color=black, size=.05, label=u^1_1, fontscale=1, position=above, distance=-.1cm]{AS}
\Vertex[x=11.3, y=-3.75, Math, shape=circle, color=black, size=.05, label=v, fontscale=1, position=above, distance=-.07cm]{AT}

\Edge[Direct, color=black, lw=1pt](AS)(AT)

\draw[] (9.1,-4.7) -- (11.5,-4.7);
\draw[] (9.1,-3.95) -- (11.5,-3.95);
\draw[] (9.1,-3.2) -- (11.5,-3.2);
\draw[] (9.1,-2.45) -- (11.5,-2.45);
\draw[] (9.1,-1.7) -- (11.5,-1.7);

\Text[x=-.7, y=-5.1, fontsize=\scriptsize]{$G_2$}
\Text[x=-.8, y=-3.575, fontsize=\scriptsize]{$t$}
\Text[x=-.8, y=-2.075, fontsize=\scriptsize]{$T$}

\Text[x=1, y=-5.1, fontsize=\scriptsize]{$i=0$}
\Text[x=4.1, y=-5.1, fontsize=\scriptsize]{$i=1$}
\Text[x=7.2, y=-5.1, fontsize=\scriptsize]{$i=2$}
\Text[x=10.3, y=-5.1, fontsize=\scriptsize]{$i=3$}

\end{tikzpicture}
\caption{By changing its outgoing edge, $\tilde{v}$ is able to affect whether it is selected by the mechanism or not. \autoref{lem:impartiality-additive-ub} gives a condition over $T,\ t$, and $k$ such that this cannot happen. In this example, the sum of the indegrees of the vertices of $G_1$ that are shown in the figure is $5(t+1)$, so we need that $5(t+1)\leq kn$. On the other hand, when $T=t+2$ \autoref{lem:impartiality-additive-ub} states that impartiality is guaranteed as long as $4t+3>kn+\min\{k,2\}$. The graphs of this example clearly violate this condition.}
\label{fig:lemma-3}
\end{figure}

It is worth noting that \autoref{lem:impartiality-additive-ub} implies the impartiality of both the supermajority rule with threshold $\lfloor n/2\rfloor+1$ and the two contenders mechanism. When taking a single threshold $T=t$, the sufficient condition for impartiality becomes $T>kn/2$, which is clearly satisfied when $k=1$ and $T=\lfloor n/2\rfloor +1$ as in the supermajority rule with threshold $\lfloor n/2\rfloor+1$. When taking thresholds that differ by~$1$, $T=t+1$, the sufficient condition for impartiality becomes $T>kn/3+1$, which is clearly satisfied when $k=1$ and $T=\lfloor n/3\rfloor +2$ as in the two contenders mechanism. We will now prove \autoref{thm:additive-ub} by using this same condition to guarantee impartiality, but choosing the thresholds in such a way that the additive deviation is as small as possible.

\begin{proof}[Proof of \autoref{thm:additive-ub}]
Let $n\in \NN_+,\ G\in \calG_n(k)$ and $\kappa, c>0$ such that $k\leq cn^{\kappa}$.
Let $f$ be the selection mechanism given by the twin threshold mechanism with thresholds $T=\lfloor \sqrt{8c}n^{\frac{1+\kappa}{2}}\rfloor + 1$ and $t=\lceil \sqrt{c}n^{\frac{1+\kappa}{2}}\rceil $. 
 Since $k(n+1)\leq cn^{1+\kappa}+cn^{\kappa}\leq 2cn^{1+\kappa}$, we have from \autoref{lem:impartiality-additive-ub} that a sufficient condition for impartiality is that 
\[
    \frac{1}{2}(T^2+T+3t-t^2) > 2cn^{1+\kappa}.
\]
Replacing $T$ and $t$ yields
\[
    \frac{1}{2}(T^2+T+3t-t^2) > \frac{1}{2}\left( 8cn^{1+\kappa} + \sqrt{8c}n^{\frac{1+\kappa}{2}} + 3\sqrt{c}n^{\frac{1+\kappa}{2}} - \lceil \sqrt{c}n^{\frac{1+\kappa}{2}}\rceil^2 \right) > 2cn^{1+\kappa},
\]
where we used that $cn^{\kappa}>1$ and thus $\lceil \sqrt{c}n^{\frac{1+\kappa}{2}}\rceil^2 < 4 cn^{1+\kappa}$. We conclude that the mechanism is impartial for these values of $T$ and $t$.

In order to obtain the additive bound, first consider a graph $G=(N,E)$ such that $f$ returns the empty set when run with this graph as input. 
Let $v^*$ be such that $\delta^-(v^*)=\Delta(G)$ and note that necessarily $\hat{\delta}^I(v^*)\leq T-1$. 
Since there are at most $\lfloor kn/t\rfloor$ vertices with indegree $t$ or higher, a maximum of $\lfloor kn/t\rfloor-1$ in-neighbors of $v^*$ have their outgoing edges deleted during the algorithm. 
Therefore, we conclude that $\Delta(G)\leq T+\lfloor kn/t \rfloor-2$.

Consider now $G=(N,E)$ such that the mechanism returns a set $\{v\}$, and let $v^*$ be such that $\delta^-(v^*)=\Delta(G)$. 
Once again, a maximum of $\lfloor kn/t\rfloor-1$ in-neighbors of $v^*$ have their outgoing edges deleted during the algorithm. 
Using the fact that $\hat{\delta}^I(v^*)\leq \hat{\delta}^I(v)$ since $v$ is selected, we conclude that 
\[
    \Delta(G)-\delta^-(v) \leq \left(\delta^-(v)+\left\lfloor \frac{kn}{t} \right\rfloor-1\right) -\delta^-(v) = \left\lfloor \frac{kn}{t}\right\rfloor-1.
\]

Since the value obtained in the former case is greater than or equal to the one obtained in the latter for any values of $T$ and $t$, we have that $f$ is $\alpha$-additive for $\alpha = T+\lfloor kn/t \rfloor-2$. 
Therefore, for the specified values of $T$ and $t$ and given the upper bound on $k$, $f$ is $\alpha$-additive for
\[
    \alpha = \left\lfloor \sqrt{8c}n^{\frac{1+\kappa}{2}}\right\rfloor +1 + \left\lfloor \frac{cn^{1+\kappa}}{\left\lceil \sqrt{c}n^{\frac{1+\kappa}{2}} \right\rceil }\right\rfloor -2 \leq  (\sqrt{8}+1)\sqrt{c}n^{\frac{1+\kappa}{2}} =  O(n^{\frac{1+\kappa}{2}}).
\]
We conclude that $f$ is $O(n^{\frac{1+\kappa}{2}})$-additive.

When $k=1$, a more detailed analysis yields the bound of $\sqrt{7.25n}$. First observe that in this case, by \autoref{lem:impartiality-additive-ub}, impartiality holds whenever
\[
    \frac{1}{2}(T^2+T+3t-t^2)-(n+1) > 0
\]
and thus when
\[
    T^2+T-(t^2-3t+2n+2) > 0.
\]
The left-hand side is equal to zero if and only if
\[
    T=\frac{-1 \pm \sqrt{4(t^2-3t+2n+2)+1}}{2} = \pm \sqrt{t^2-3t+2n+\frac{9}{4}} - \frac{1}{2},
\]
and since $T$ has to be non-negative impartiality holds if 
\[
T > \sqrt{t^2-3t+2n+\frac{9}{4}} - \frac{1}{2}.
\]
This is trivially satisfied if we take
\[
t=\left\lceil1.23\sqrt{n}\right\rceil,\quad T = \left\lfloor \sqrt{\left\lceil1.23\sqrt{n}\right\rceil^2 - 3\left\lceil1.23\sqrt{n}\right\rceil +2n+\frac{9}{4}} +\frac{1}{2} \right\rfloor.
\]
As before, 
the twin threshold mechanism with thresholds $T$ and $t$ is $\alpha$-additive for any $\alpha\geq\alpha(n):= T+\lfloor n/t \rfloor -2$. In order to obtain an upper bound on $\alpha(n)$, we start by bounding $T$ from above:
\begin{align*}
    T\ & \leq \sqrt{(1.23\sqrt{n}+1)^2 - 3(1.23\sqrt{n}+1) +2n+\frac{9}{4}} + \frac{1}{2}\\
    & = \sqrt{3.5129n-1.23\sqrt{n}+\frac{1}{4}} + \frac{1}{2},
\end{align*}
where the first inequality holds because $a^2-3a$ is increasing for $a\geq 3/2$. Then,
\begin{align*}
\alpha(n) = T+\lfloor n/t \rfloor-2\  & \leq \sqrt{3.5129n-1.23\sqrt{n}+\frac{1}{4}} + \frac{1}{2} +\left\lfloor \frac{n}{\lceil 1.23\sqrt{n}\rceil} \right\rfloor-2 \\
& \leq \sqrt{3.5129n-1.23\sqrt{n}+\frac{1}{4}} + \frac{\sqrt{n}}{1.23} -\frac{3}{2} =: h(n).
\end{align*}
Let $g(n):=\sqrt{7.25n}-h(n)$. Then, by the previous inequality,
\begin{align*}
g(1) & = \sqrt{7.25} -\sqrt{3.5129-1.23+\frac{1}{4}} - \frac{1}{1.23} +\frac{3}{2} \\
& = \sqrt{7.25} -\sqrt{2.5329} - \frac{1}{1.23} +\frac{3}{2} \approx 1.79 >0.
\end{align*}
Moreover, 
\begin{align*}
g'(n) & = \frac{\sqrt{7.25}}{2\sqrt{n}} - \frac{7.0258\sqrt{n}-1.23}{2\sqrt{n}\sqrt{14.0516n-4.92\sqrt{n}+1}} - \frac{1}{2.46\sqrt{n}} \\
& = \frac{(2.46\sqrt{7.25}-2)\sqrt{14.0516n-4.92\sqrt{n}+1}-17.283468\sqrt{n}+3.0258}{4.92\sqrt{n}\sqrt{14.0516n-4.92\sqrt{n}+1}},\\
& \geq \frac{(2.46\sqrt{7.25}-2)\sqrt{14.05n-4.92\sqrt{n}+1}-17.3\sqrt{n}+3}{4.92\sqrt{n}\sqrt{14.0516n-4.92\sqrt{n}+1}},\\
& = \frac{45(2.46\sqrt{7.25}-2)\sqrt{14.05n-4.92\sqrt{n}+1}-778.5\sqrt{n}+135}{221.4\sqrt{n}\sqrt{14.0516n-4.92\sqrt{n}+1}},
\end{align*}
which is non-negative when $n\geq 1$. To see this, note that the denominator of the last expression is positive and that $45(2.46\sqrt{7.25}-2)>208$, so that $g'(n)\geq 0$ if
\[
208\sqrt{14.05n-4.92\sqrt{n}+1} \geq 778.5\sqrt{n}-135,
\]
\ie if
\[
1796.95n-2663.88\sqrt{n}+25039\geq 0.
\]
This inequality indeed holds for every $n\geq 1$. This is immediate for $n\in \{1,2,3\}$. For $n\geq 4$ we have that $n\geq 2\sqrt{n}$, and thus the expression on the left is greater than or equal to $930.02\sqrt{n}+25039$, which is clearly non-negative. We conclude that $g(n)\geq 0$ for every $n\in \NN$ and thus $\alpha(n)\leq h(n) \leq \sqrt{7.25n}$ for every $n\geq 1$, so $f$ is $\sqrt{7.25n}$-additive.
\end{proof}

Recalling \autoref{fig:plot_alpha_n}, we note that for small values of $n$, the guarantee given by the optimal twin threshold mechanism is much smaller than $\sqrt{7.25n}$, and that the bound becomes significantly better than those given by the other mechanisms as $n$ grows. 
Closing the gap between the lower bound of  \autoref{thm:additive-plurality-lb} and the upper bound given by the optimal twin threshold mechanism is left as our main open question.

\section{A Tight Impossibility Result for Approval}

So far, we have developed a new mechanism for impartial selection and have established an additive performance guarantee for the mechanism relative to the maximum outdegree in the graph. We will now take a closer look at the case where the maximum outdegree is unbounded, \ie at the approval setting.

When applied to the approval setting, \autoref{thm:additive-ub} provides a performance guarantee of $O(n)$. As the maximum indegree in a graph with $n$ vertices is $n-1$, this bound is trivially achieved by any impartial mechanism including the mechanism that never selects. \citet{caragiannis2022impartial} have used a careful case analysis to show that deterministic impartial mechanisms cannot be better than $3$-additive.
We show that the trivial upper bound of $n-1$ is in fact tight for all $n$, which means that the mechanism that never selects provides the best possible additive performance guarantee among all deterministic impartial mechanisms. Our result is in fact more general and again holds relative to the maximum outdegree~$k$.
\begin{theorem}
    \label{thm:additive-lb}
    Let $n\in\NN$ and $k\leq n-1$. Let $f$ be an impartial deterministic selection mechanism such that $f$ is $\alpha$-additive on $\calG_n(k)$. Then $\alpha\geq k$. 
    In particular, if $f$ is $\alpha$-additive on $\calG_n$, then $\alpha\geq n-1$.
\end{theorem}

In the practically relevant case where individuals are not allowed to abstain and the minimum outdegree is therefore at least~$1$, a small improvement can be obtained by selecting a vertex with an incoming edge from a fixed vertex, and again breaking ties by a fixed ordering of the vertices. The selected vertex then has indegree at least $1$, which for the approval setting implies ($n-2$)-additivity. This guarantee is again best possible.
\begin{theorem}
    \label{thm:additive-lb-abstentions}
    Let $n\in\NN$ and $k\leq n-1$. Let $f$ be an impartial deterministic selection mechanism such that $f$ is $\alpha$-additive on $\calG^+_n(k)$. Then $\alpha\geq k-1$.
    In particular, if $f$ is $\alpha$-additive on $\calG^+_n$, then $\alpha\geq n-2$.
\end{theorem}

To prove both theorems we study the performance of impartial, but not necessarily deterministic, selection mechanisms on a particular class of graphs which in the case of $n$ vertices we denote by $\calG^T_n$. 
For each $n\in\NN$, a graph $G=(N,E)\in \calG_n$ belongs to $\calG^T_n$ if and only if there exists an $r$-partition of $N$ for some $r\geq 1$, which we denote by $(S_1,\ldots,S_{r})$, such that (i)~$u<v$ for every $u\in S_i$, and $v\in S_j$ with $i<j$, and (ii)~$E=\{(u,v)\in S_i\times S_j:  i,j \in [r], i\leq j, u\not=v\}$.
In other words, $\calG^T_n$ contains all graphs obtained by taking an ordered partition of a set of $n$ unlabeled vertices and adding edges from each vertex to all other vertices in the same part and in greater parts. We will not be interested in isomorphic graphs within the class and thus only consider partitions of the vertices in increasing order. A graph in~$\calG^T_n$ is thus characterized by the partition $(S_1,\ldots,S_{r})$, or by the tuple $(s_1,\ldots,s_r)$ where $s_i=|S_i|$ for each $i\in [r]$. For a given graph $G\in\calG^T_n$, we denote the former by $S(G)$, the latter by $s(G)$, and the length of $s(G)$ by $r(G)$. Finally, for $G\in \calG^T_n$, let $\lambda_G$ be the number of graphs with $n$ vertices isomorphic to $G$, \ie 
\[
    \lambda_G=\frac{n!}{\prod_{i=1}^{r(G)}{(s(G))_i!}}.
\]
\autoref{fig:transitive-graphs-2-3} shows the graphs in $\calG^T_2$ and $\calG^T_3$, along with their tuple representation $s(G)$ and associated values $\lambda_G$.
The sums, for $n\in\NN$, of the values $\lambda_G$ for all graphs $G\in\calG^T_n$ are known as Fubini numbers and count the number of weak orders on an $n$-element set. The following lemma establishes a property of Fubini numbers that is readily appreciated for the cases shown in \autoref{fig:transitive-graphs-2-3} but in fact holds for all~$n$. The property was known previously~\citep{diagana2017some}, but we provide an alternative proof in \autoref{app:odd-graphs} for the sake of completeness.
\begin{lemma}
\label{lem:odd-graphs}
For every $n\in \NN,\ n\geq 1$, $\sum_{G\in \calG^T_n}\lambda_G$ is an odd number.
\end{lemma}

\begin{figure}[t]
\centering
\begin{tikzpicture}[scale=0.88]

\Vertex[y=1.5, Math, shape=circle, color=black, size=.05]{A}
\Vertex[Math, shape=circle, color=black, size=.05]{B}

\Edge[Direct, color=black, lw=1pt](A)(B)

\Vertex[x=1.3, y=1.5,  Math, shape=circle, color=black, size=.05]{C}
\Vertex[x=1.3, Math, shape=circle, color=black, size=.05]{D}

\Edge[Direct, color=black, lw=1pt, bend=-20](C)(D)
\Edge[Direct, color=black, lw=1pt, bend=-20](D)(C)

\Vertex[x=3.5, y=1.5, Math, shape=circle, color=black, size=.05]{E}
\Vertex[x=2.6, Math, shape=circle, color=black, size=.05]{F}
\Vertex[x=4.4, Math, shape=circle, color=black, size=.05]{G}

\Edge[Direct, color=black, lw=1pt](E)(F)
\Edge[Direct, color=black, lw=1pt](E)(G)
\Edge[Direct, color=black, lw=1pt](F)(G)

\Vertex[x=6.6, y=1.5, Math, shape=circle, color=black, size=.05]{H}
\Vertex[x=5.7, Math, shape=circle, color=black, size=.05]{I}
\Vertex[x=7.5, Math, shape=circle, color=black, size=.05]{J}

\Edge[Direct, color=black, lw=1pt](H)(I)
\Edge[Direct, color=black, lw=1pt](H)(J)
\Edge[Direct, color=black, lw=1pt, bend=-20](I)(J)
\Edge[Direct, color=black, lw=1pt, bend=-20](J)(I)

\Vertex[x=9.7, y=1.5, Math, shape=circle, color=black, size=.05]{K}
\Vertex[x=8.8, Math, shape=circle, color=black, size=.05]{L}
\Vertex[x=10.6, Math, shape=circle, color=black, size=.05]{M}

\Edge[Direct, color=black, lw=1pt, bend=-20](K)(L)
\Edge[Direct, color=black, lw=1pt, bend=-20](L)(K)
\Edge[Direct, color=black, lw=1pt](K)(M)
\Edge[Direct, color=black, lw=1pt](L)(M)

\Vertex[x=12.8, y=1.5, Math, shape=circle, color=black, size=.05]{N}
\Vertex[x=11.9, Math, shape=circle, color=black, size=.05]{O}
\Vertex[x=13.7, Math, shape=circle, color=black, size=.05]{P}

\Edge[Direct, color=black, lw=1pt, bend=-20](N)(O)
\Edge[Direct, color=black, lw=1pt, bend=-20](O)(N)
\Edge[Direct, color=black, lw=1pt, bend=-20](N)(P)
\Edge[Direct, color=black, lw=1pt, bend=-20](P)(N)
\Edge[Direct, color=black, lw=1pt, bend=-20](O)(P)
\Edge[Direct, color=black, lw=1pt, bend=-20](P)(O)

\Text[x=-1.4, y=-.75]{$s(G)$}
\Text[x=-1.4, y=-1.25]{$\lambda_G$}

\Text[y=-.75]{$(1,1)$}
\Text[y=-1.25]{$2$}

\Text[x=1.3, y=-.75]{$(2)$}
\Text[x=1.3, y=-1.25]{$1$}

\Text[x=3.5, y=-.75]{$(1,1,1)$}
\Text[x=3.5, y=-1.25]{$6$}

\Text[x=6.6, y=-.75]{$(1,2)$}
\Text[x=6.6, y=-1.25]{$3$}

\Text[x=9.7, y=-.75]{$(2,1)$}
\Text[x=9.7, y=-1.25]{$3$}

\Text[x=12.8, y=-.75]{$(3)$}
\Text[x=12.8, y=-1.25]{$1$}
\end{tikzpicture}
\caption{Graphs in $\calG^T_2$ and $\calG^T_3$.}
\label{fig:transitive-graphs-2-3}
\end{figure}

For every pair of graphs $G,G'\in \calG^T_n$ and $j\in\{2,\ldots,r(G)\}$,
we say that there is a \textit{$j$-transition} from $G$ to $G'$ if $r(G)=r(G')+1$, $(s(G))_{j}=1$, and
\[
    (s(G'))_i =\left\{ \begin{array}{ll}
             (s(G))_i &   \text{ if } i\leq j-2, \\
             (s(G))_i +1 &  \text{ if } i= j-1, \\
             (s(G))_{i+1} &  \text{ if } i\geq j.
             \end{array}
   \right.
\]
Intuitively, a $j$-transition can be obtained by changing the outgoing edges of the single vertex in the set $(S(G))_{j}$, including not only edges to vertices in $\bigcup_{i=j}^{r(G)}(S(G))_i$ but also to vertices in $(S(G))_{j-1}$. 
When the value of $j$ is not relevant in a particular context, we simply say that there is a transition from $G$ to $G'$ if there exists some $j\in\{2,\ldots,r(G)\}$ such that there is a $j$-transition from $G$ to $G'$. 
Observe that for every pair of graphs $G,G'\in \calG^T_n$ there is at most one $j\in\{2,\ldots,r(G)\}$ such that there is a $j$-transition from $G$ to $G'$, and that if there is a transition from $G$ to $G'$, there cannot be a transition from $G'$ to $G$.
This kind of relation between ordered partitions has been exploited by \citet{insko2017ordered} for studying an expansion of the determinant, giving rise to a partial order on~$\calG_n^T$.
In our context, it turns out to be relevant because whenever there is a $j$-transition from $G$ to $G'$, any impartial mechanism either selects the vertex in $(S(G))_{j}$ both in $G$ and $G'$, or in none of them.

In the case of plurality, impartiality was shown by \citet{holzman2013impartial} to be incompatible with two further axioms: \textit{positive unanimity}, which requires for all $G=(N,E)\in\calG(1)$ that $v\in f(G)$ if $\delta^-(v)=|N|-1$; and \textit{negative unanimity}, which requires for all $G=(N,E)\in\calG(1)$ that $v\not\in f(G)$ if $\delta^-(v)=0$. This result holds even on a restricted class of graphs, consisting of a single cycle and additional vertices with edges onto that cycle, and has immediate and very strong implications on the best multiplicative approximation guarantee that an impartial mechanism can achieve. For additive performance guarantees, however, the incompatibility of impartiality with the other two axioms implies only a lower bound of $2$. We will see in the following that on the class $\calG^T_n$, impartiality is incompatible with a single axiom that weakens positive and negative unanimity. Strong lower bounds regarding additive performance guarantees for approval then follow immediately.

The class $\calG^T_n$ is very different, and has to be very different, from the class of graphs used by \citeauthor{holzman2013impartial}, and will ultimately require a new analysis. We can, however, follow the approach of \citeauthor{holzman2013impartial} to consider randomized mechanisms rather than deterministic ones, which allows us without loss of generality to restrict attention to mechanisms that treat vertices symmetrically. A \textit{randomized selection mechanism} for $\calG$ is given by a family of functions $f\colon\calG_n\to [0,1]^n$ that maps each graph to a probability distribution on the set of its vertices, such that $\sum_{i=1}^n{(f(G))_i}\leq 1$ for every graph $G\in \calG_n$.
Analogously to the case of deterministic mechanisms, we say that a randomized selection mechanism $f$ is \emph{impartial} on $\calG'\subseteq \calG$ if for every pair of graphs $G = (N, E)$ and $G' = (N, E')$ in $\calG'$ and every $v\in N$, $(f(G))_v = (f(G'))_v$ whenever $E \setminus (\{v\} \times N) = E' \setminus (\{v\} \times N)$. We say that a randomized selection mechanism $f$ satisfies \textit{weak unanimity} on $\calG_n$ if for every $G=(N,E)\in \calG_n$ such that $\delta^-(v)=n-1$ for some $v\in N$,
\[
    \sum_{u\in N:  \delta^-(u)\geq 1}(f(G))_u\geq 1.
\]
In other words, weak unanimity requires that a vertex with positive indegree is chosen with probability~$1$ whenever there exists a vertex with indegree~$n-1$. We finally say that a randomized mechanism $f$ is \emph{symmetric} if it is invariant with respect to renaming of the vertices, \ie if for every $G = (N, E)\in \calG$, every $v\in N$ and every permutation $\pi = (\pi_1,\ldots, \pi_{|N|})$ of $N$,
$
(f(G_{\pi}))_{\pi_v} = (f(G))_v$,
where $G_{\pi} = (N, E_{\pi})$ with $E_{\pi} = \{(\pi_u, \pi_v):  (u, v) \in E\}$. For a given randomized mechanism $f$, we denote by $f_{\text{s}}$ the mechanism obtained by applying a random permutation $\pi$ to the vertices of the input graph, invoking $f$, and permuting the result by the inverse of $\pi$. Thus, for all $n\in \NN,\ G \in \calG_n$, and $v \in N$,
\[
    (f_{\text{s}}(G))_v = \frac{1}{n!} \sum_{\pi\in \calS_n}(f(G_{\pi}))_{\pi_v},
\]
where $\calS_n$ is the set of all permutations $\pi = (\pi_1,\ldots, \pi_n)$ of a set of $n$ elements. 

The following lemma, which we prove in \autoref{app:symmetry-axiom}, establishes that $f_{\text{s}}$ is symmetric for every randomized mechanism $f$ and inherits impartiality and weak unanimity from $f$. The lemma is a straightforward variant of a result of \citeauthor{holzman2013impartial} and will allow us to restrict attention to symmetric randomized mechanisms.
\begin{lemma}
\label{lem:symmetry-axiom}
Let $f$ be a randomized selection mechanism that is impartial and weakly unanimous on $\calG_n$. Then, $f_{\text{s}}$ is symmetric, impartial, and weakly unanimous on $\calG_n$.
\end{lemma}

We are now ready to state our axiomatic impossibility result, which can be seen as a stronger version of that of \citeauthor{holzman2013impartial} for the case of unbounded outdegree. Both lower bounds follow easily from this result.
\begin{lemma}
\label{lem:imposibility-selection}
For every $n\in\NN,\ n\geq 2$, there exists no randomized selection mechanism $f$ satisfying impartiality and weak unanimity on $\calG^T_n$.
\end{lemma}
\begin{proof}
Let $n\in \NN,\ n\geq 2$ and suppose that there exists a randomized selection mechanism $f$ satisfying impartiality and weak unanimity on $\calG_n$.
Since we can assume symmetry due to \autoref{lem:symmetry-axiom}, for each graph $G \in \calG^T_n$ we have that for every $i\in [r(G)]$ and every $u,v\in (S(G))_i$, $(f(G))_u=(f(G))_v$, and thus we denote this value simply as $(f(G))_i$. 
We consider for the proof an undirected graph $\calH_n=(\calG^T_n, \calF)$, such that for every pair of graphs $G,G'\in \calG^T_n$ we have that $\{G,G'\}\in \calF$ if and only if there is a transition from $G$ to $G'$ or from $G'$ to $G$. 
By definition of a transition, for each $\{G,G'\}\in \calF$ we have $|r(G)-r(G')|=1$. 
Therefore, $\calH_n$ is bipartite with partition $(L_n, R_n)$ where $L_n=\{G\in \calG^T_n:  r(G) \text{ is even}\}$ and $R_n=\{G\in \calG^T_n:  r(G) \text{ is odd}\}$.
\autoref{fig:graphs-of-graphs-2-3-4} depicts the graphs $\calH_2,\ \calH_3$, and $\calH_4$.
For each graph $G \in \calG^T_n$, we define $i(G)=2$ if $(s(G))_1=1$ and $i(G)=1$ otherwise.
\begin{figure}[t]
\centering

\begin{tikzpicture}[scale=0.85]

\Vertex[color=white, size=1.1, label={(1,1)}]{A}

\Vertex[x=3, color=white, size=1.1, label={(2)}]{B}

\Edge[Direct, color=black, lw=1pt, label=2](A)(B)

\Vertex[x=5.5, y=.8, color=white, size=1.1, label={(1,2)}]{C}
\Vertex[x=5.5, y=-.8, color=white, size=1.1, label={(2,1)}]{D}

\Vertex[x=8.5, y=.8, color=white, size=1.1, label={(1,1,1)}]{E}
\Vertex[x=8.5, y=-.8, color=white, size=1.1, label={(3)}]{F}

\Edge[Direct, color=black, lw=1pt, label=3](E)(C)
\Edge[Direct, color=black, lw=1pt, label=2](E)(D)
\Edge[Direct, color=black, lw=1pt, label=2](D)(F)

\Vertex[x=11, y=2.4, color=white, size=1.1, label={(1,1,1,1)}]{G}
\Vertex[x=11, y=.8, color=white, size=1.1, label={(1,3)}]{H}
\Vertex[x=11, y=-.8, color=white, size=1.1, label={(2,2)}]{I}
\Vertex[x=11, y=-2.4, color=white, size=1.1, label={(3,1)}]{J}

\Vertex[x=14, y=2.4, color=white, size=1.1, label={(1,1,2)}]{K}
\Vertex[x=14, y=.8, color=white, size=1.1, label={(1,2,1)}]{L}
\Vertex[x=14, y=-.8, color=white, size=1.1, label={(2,1,1)}]{M}
\Vertex[x=14, y=-2.4, color=white, size=1.1, label={(4)}]{N}

\Edge[Direct, color=black, lw=1pt, label=4](G)(K)
\Edge[Direct, color=black, lw=1pt, label=3](G)(L)
\Edge[Direct, color=black, lw=1pt, label=2, position={above left=5mm}](G)(M)
\Edge[Direct, color=black, lw=1pt, label=3, position={right=5mm}](L)(H)
\Edge[Direct, color=black, lw=1pt, label=2, position={below left=5mm}](K)(I)
\Edge[Direct, color=black, lw=1pt, label=3](M)(I)
\Edge[Direct, color=black, lw=1pt, label=2](M)(J)
\Edge[Direct, color=black, lw=1pt, label=2](J)(N)

\end{tikzpicture}

\caption{Directed versions of $\calH_2,\ \calH_3$, and $\calH_4$. For $n\in \{2,3,4\}$, each graph $G\in \calG^T_n$ is represented by the tuple $s(G)$ and each edge $(G,G')$ labeled with $j\in \NN$ represents a $j$-transition from $G$ to $G'$.}
\label{fig:graphs-of-graphs-2-3-4}
\end{figure}

Let $G=(N,E)$ and $G'=(N,E')$ be two graphs in $\calG^T_n$ such that there is a $j$-transition from $G$ to $G'$ for some $j\in\{2,\ldots,r(G)\}$.
Denoting by $v$ the unique vertex in $(S(G))_{j}$, which is also in $S(G')_{j-1}$, we have that $E\setminus (\{v\}\times N) = E'\setminus (\{v\}\times N)$, and since $f$ is impartial, $(f(G))_{j} = (f(G'))_{j-1}$. 
Therefore,
\begingroup
\allowdisplaybreaks
\begin{align*}
    \lambda_{G'} \cdot (s(G'))_{j-1} \cdot (f(G'))_{j-1} & = \frac{n!}{\prod_{i=1}^{r(G')}{(s(G'))_i!}} (s(G'))_{j-1} (f(G'))_{j-1}\\
    & = \frac{n!}{((s(G'))_{j-1}-1)!\prod_{i\in [r(G')]\setminus \{j-1\}}{(s(G'))_i!}} (f(G'))_{j-1}\\
    & = \frac{n!}{((s(G))_{j-1})!\prod_{i\in [r(G)]\setminus \{j-1,j\}}{(s(G))_i!}} (f(G))_{j}\\
    & = \frac{n!}{\prod_{i\in [r(G)]}{(s(G))_i!}} (s(G))_{j} \cdot (f(G))_{j}\\
    & = \lambda_G \cdot (s(G))_{j}\cdot  (f(G))_{j}.
\end{align*}
\endgroup
The first two and last two equalities are obtained by replacing known expressions and simple calculations. 
The third equality comes from the fact that $(s(G'))_i=(s(G))_i$ for every $i\leq j-2$, $(s(G'))_{j-1}=(s(G))_{j-1}+1$, and $(s(G'))_i=(s(G))_{i+1}$ for every $i\geq j$. 
Moreover, observe that for each $G\in \calG^T_n$ and for each $j\in \{i(G),\ldots,r(G)\}$, there exists exactly one $G'\in \calG^T_n$ such that there is a $j$-transition from $G$ to $G'$ (if $(s(G))_j=1$) or a $(j+1)$-transition from $G'$ to $G$ (if $(s(G))_j\geq 2$). Therefore,
\begin{equation}\label{eq:sums_bipartition}
    \sum_{G\in L_n}\sum_{i=i(G)}^{r(G)}{\lambda_G \cdot (s(G))_{i} \cdot (f(G))_{i}} = \sum_{G\in R_n}\sum_{i=i(G)}^{r(G)}{\lambda_G \cdot (s(G))_{i} \cdot (f(G))_{i}}.
\end{equation}
We now derive two important sets of inequalities.
From the fact that $f$ is a selection mechanism, for each $G\in \calG^T_n$ we have that
\begin{equation}
    \sum_{i=i(G)}^{r(G)}{(s(G))_{i} \cdot (f(G))_{i}}\leq 1,\label{eq:probs-leq-1}
\end{equation}
where replacing 1 by $i(G)$ on the left-hand side is possible since it can only make the sum smaller and thus keeps the inequality.
On the other hand, from the fact that $f$ satisfies weak unanimity, for each $G\in \calG^T_n$ we have that
\begin{equation}
-\sum_{i=i(G)}^{r(G)}{(s(G))_{i} \cdot (f(G))_{i}}\leq -1,\label{eq:probs-geq-1}
\end{equation}
where we can omit the term for $i=1$ on the left-hand side whenever $(s(G))_1=1$, because the vertex in $(S(G))_1$ has indegree 0 in such case. 

In order to cancel out the left-hand sides of the previous inequalities, we assign a sign to each part of the bipartition of $\calH_n$. Let $\text{sign}(L_n), \text{sign}(R_n)\in \{-1,1\}$ with $\text{sign}(L_n)\cdot \text{sign}(R_n) = -1$, and let $\text{sign}(G)=\text{sign}(L_n)$ for each $G\in L_n$ and $\text{sign}(G)=\text{sign}(R_n)$ for each $G\in R_n$. 
Summing up the inequalities \eqref{eq:probs-leq-1} multiplied by $\lambda_G$ for every $G\in \calG^T_n$ with $\text{sign}(G)=1$ and the inequalities \eqref{eq:probs-geq-1} multiplied by $\lambda_G$ for every $G\in \calG^T_n$ with $\text{sign}(G)=-1$, we obtain
\[
    \sum_{G\in \calG^T_n}\sum_{i=i(G)}^{r(G)}{\text{sign}(G) \cdot \lambda_G \cdot (s(G))_{i} \cdot (f(G))_{i}} \leq \sum_{G\in \calG^T_n} \text{sign}(G)\cdot \lambda_G.
\]
By expression~\eqref{eq:sums_bipartition} the left-hand side is equal to~$0$. However, we know from \autoref{lem:odd-graphs} that $\sum_{G\in \calG^T_n}\lambda_G$ is odd, so the right-hand side cannot be equal to $0$. For one of the two possible choices of $\text{sign}(L_n)$ and $\text{sign}(R_n)$ the right-hand side is negative, and we obtain a contradiction. We conclude that a randomized selection mechanism $f$ satisfying impartiality and weak unanimity on $\calG_n$ cannot exist.
\end{proof}

As an illustration, the counterexamples constructed for $n=3$ and $n=4$ are shown in \autoref{fig:counterexample-weak-unanimity}.

\begin{figure}[t]
\centering
\begin{tikzpicture}

\Vertex[x=1, y=6.73, Math, shape=circle, color=black, size=.05]{A}
\Vertex[Math, y=5, shape=circle, color=black, size=.05, label=p_1, fontscale=1, position=below left, distance=-.16cm]{B}
\Vertex[x=2, y=5, Math, shape=circle, color=black, size=.05, label=p_2, fontscale=1, position=below right, distance=-.16cm]{C}

\Edge[Direct, color=black, lw=1pt](A)(B)
\Edge[Direct, color=black, lw=1pt](A)(C)
\Edge[Direct, color=black, lw=1pt](B)(C)

\Vertex[x=5, y=6.73, Math, shape=circle, color=black, size=.05]{D}
\Vertex[x=4, y=5, Math, shape=circle, color=black, size=.05, label=p_2, fontscale=1, position=below left, distance=-.16cm]{E}
\Vertex[x=6, y=5, Math, shape=circle, color=black, size=.05, label=p_2, fontscale=1, position=below right, distance=-.16cm]{F}

\Edge[Direct, color=black, lw=1pt](D)(E)
\Edge[Direct, color=black, lw=1pt](D)(F)
\Edge[Direct, color=black, lw=1pt, bend=-20](E)(F)
\Edge[Direct, color=black, lw=1pt, bend=-20](F)(E)

\Vertex[x=9, y=6.73, Math, shape=circle, color=black, size=.05, label=p_1, fontscale=1, position=left, distance=-.1cm]{G}
\Vertex[x=8, y=5, Math, shape=circle, color=black, size=.05, label=p_1, fontscale=1, position=below left, distance=-.16cm]{H}
\Vertex[x=10, y=5, Math, shape=circle, color=black, size=.05, label=p_3, fontscale=1, position=below right, distance=-.16cm]{I}

\Edge[Direct, color=black, lw=1pt, bend=-20](G)(H)
\Edge[Direct, color=black, lw=1pt](G)(I)
\Edge[Direct, color=black, lw=1pt, bend=-20](H)(G)
\Edge[Direct, color=black, lw=1pt](H)(I)

\Vertex[x=13, y=6.73, Math, shape=circle, color=black, size=.05, label=p_3, fontscale=1, position=left, distance=-.1cm]{J}
\Vertex[x=12, y=5, Math, shape=circle, color=black, size=.05, label=p_3, fontscale=1, position=below left, distance=-.16cm]{K}
\Vertex[x=14, y=5, Math, shape=circle, color=black, size=.05, label=p_3, fontscale=1, position=below right, distance=-.16cm]{L}

\Edge[Direct, color=black, lw=1pt, bend=-20](J)(K)
\Edge[Direct, color=black, lw=1pt, bend=-20](J)(L)
\Edge[Direct, color=black, lw=1pt, bend=-20](K)(J)
\Edge[Direct, color=black, lw=1pt, bend=-20](K)(L)
\Edge[Direct, color=black, lw=1pt, bend=-20](L)(J)
\Edge[Direct, color=black, lw=1pt, bend=-20](L)(K)

\Text[x=1, y=4]{$p_1+p_2\geq 1\ (-6)$}
\Text[x=5, y=4]{$2p_2\leq 1\ (3)$}
\Text[x=9, y=4]{$2p_1+p_3\leq 1\ (3)$}
\Text[x=13, y=4]{$3p_3\geq 1\ (-1)$}


\Vertex[y=2, Math, shape=circle, color=black, size=.05]{A}
\Vertex[x=2, y=2, Math, shape=circle, color=black, size=.05, label=p_1, fontscale=1, position=above right, distance=-.16cm]{B}
\Vertex[Math, shape=circle, color=black, size=.05, label=p_2, fontscale=1, position=below left, distance=-.16cm]{C}
\Vertex[x=2, Math, shape=circle, color=black, size=.05, label=p_3, fontscale=1, position=below right, distance=-.16cm]{D}

\Edge[Direct, color=black, lw=1pt](A)(B)
\Edge[Direct, color=black, lw=1pt](A)(C)
\Edge[Direct, color=black, lw=1pt](A)(D)
\Edge[Direct, color=black, lw=1pt](B)(C)
\Edge[Direct, color=black, lw=1pt](B)(D)
\Edge[Direct, color=black, lw=1pt](C)(D)

\Vertex[x=4, y=2, Math, shape=circle, color=black, size=.05]{E}
\Vertex[x=6, y=2, Math, shape=circle, color=black, size=.05, label=p_4, fontscale=1, position=above right, distance=-.16cm]{F}
\Vertex[x=4, Math, shape=circle, color=black, size=.05, label=p_3, fontscale=1, position=below left, distance=-.16cm]{G}
\Vertex[x=6, Math, shape=circle, color=black, size=.05, label=p_3, fontscale=1, position=below right, distance=-.16cm]{H}

\Edge[Direct, color=black, lw=1pt](E)(F)
\Edge[Direct, color=black, lw=1pt](E)(G)
\Edge[Direct, color=black, lw=1pt](E)(H)
\Edge[Direct, color=black, lw=1pt](F)(G)
\Edge[Direct, color=black, lw=1pt](F)(H)
\Edge[Direct, color=black, lw=1pt, bend=-20](G)(H)
\Edge[Direct, color=black, lw=1pt, bend=-20](H)(G)

\Vertex[x=8, y=2, Math, shape=circle, color=black, size=.05]{I}
\Vertex[x=10, y=2, Math, shape=circle, color=black, size=.05, label=p_2, fontscale=1, position=above right, distance=-.16cm]{J}
\Vertex[x=8, Math, shape=circle, color=black, size=.05, label=p_2, fontscale=1, position=below left, distance=-.16cm]{K}
\Vertex[x=10, Math, shape=circle, color=black, size=.05, label=p_5, fontscale=1, position=below right, distance=-.16cm]{L}

\Edge[Direct, color=black, lw=1pt](I)(J)
\Edge[Direct, color=black, lw=1pt](I)(K)
\Edge[Direct, color=black, lw=1pt](I)(L)
\Edge[Direct, color=black, lw=1pt, bend=-20](J)(K)
\Edge[Direct, color=black, lw=1pt](J)(L)
\Edge[Direct, color=black, lw=1pt, bend=-20](K)(J)
\Edge[Direct, color=black, lw=1pt](K)(L)

\Vertex[x=12, y=2, Math, shape=circle, color=black, size=.05]{M}
\Vertex[x=14, y=2, Math, shape=circle, color=black, size=.05, label=p_5, fontscale=1, position=above right, distance=-.16cm]{N}
\Vertex[x=12, Math, shape=circle, color=black, size=.05, label=p_5, fontscale=1, position=below left, distance=-.16cm]{O}
\Vertex[x=14, Math, shape=circle, color=black, size=.05, label=p_5, fontscale=1, position=below right, distance=-.16cm]{P}

\Edge[Direct, color=black, lw=1pt](M)(N)
\Edge[Direct, color=black, lw=1pt](M)(O)
\Edge[Direct, color=black, lw=1pt](M)(P)
\Edge[Direct, color=black, lw=1pt, bend=-20](N)(O)
\Edge[Direct, color=black, lw=1pt, bend=-20](N)(P)
\Edge[Direct, color=black, lw=1pt, bend=-20](O)(N)
\Edge[Direct, color=black, lw=1pt, bend=-20](O)(P)
\Edge[Direct, color=black, lw=1pt, bend=-20](P)(N)
\Edge[Direct, color=black, lw=1pt, bend=-20](P)(O)

\Text[x=1, y=-1]{$p_1+p_2+p_3\geq 1\ (-24)$}
\Text[x=5, y=-1]{$2p_3+p_4\leq 1\ (12)$}
\Text[x=9, y=-1]{$2p_2+p_5\leq 1\ (12)$}
\Text[x=13, y=-1]{$3p_5\geq 1\ (-4)$}

\Vertex[y=-2.5, Math, shape=circle, color=black, size=.05, label=p_1, fontscale=1, position=above left, distance=-.16cm]{AA}
\Vertex[x=2, y=-2.5, Math, shape=circle, color=black, size=.05, label=p_1, fontscale=1, position=above right, distance=-.16cm]{AB}
\Vertex[y=-4.5, Math, shape=circle, color=black, size=.05, label=p_6, fontscale=1, position=below left, distance=-.16cm]{AC}
\Vertex[x=2, y=-4.5, Math, shape=circle, color=black, size=.05, label=p_7, fontscale=1, position=below right, distance=-.16cm]{AD}

\Edge[Direct, color=black, lw=1pt, bend=-20](AA)(AB)
\Edge[Direct, color=black, lw=1pt](AA)(AC)
\Edge[Direct, color=black, lw=1pt](AA)(AD)
\Edge[Direct, color=black, lw=1pt, bend=-20](AB)(AA)
\Edge[Direct, color=black, lw=1pt](AB)(AC)
\Edge[Direct, color=black, lw=1pt](AB)(AD)
\Edge[Direct, color=black, lw=1pt](AC)(AD)

\Vertex[x=4, y=-2.5, Math, shape=circle, color=black, size=.05, label=p_4, fontscale=1, position=above left, distance=-.16cm]{AE}
\Vertex[x=6, y=-2.5, Math, shape=circle, color=black, size=.05, label=p_4, fontscale=1, position=above right, distance=-.16cm]{AF}
\Vertex[x=4, y=-4.5, Math, shape=circle, color=black, size=.05, label=p_7, fontscale=1, position=below left, distance=-.16cm]{AG}
\Vertex[x=6, y=-4.5, Math, shape=circle, color=black, size=.05, label=p_7, fontscale=1, position=below right, distance=-.16cm]{AH}

\Edge[Direct, color=black, lw=1pt, bend=-20](AE)(AF)
\Edge[Direct, color=black, lw=1pt](AE)(AG)
\Edge[Direct, color=black, lw=1pt](AE)(AH)
\Edge[Direct, color=black, lw=1pt, bend=-20](AF)(AE)
\Edge[Direct, color=black, lw=1pt](AF)(AG)
\Edge[Direct, color=black, lw=1pt](AF)(AH)
\Edge[Direct, color=black, lw=1pt, bend=-20](AG)(AH)
\Edge[Direct, color=black, lw=1pt, bend=-20](AH)(AG)

\Vertex[x=8, y=-2.5, Math, shape=circle, color=black, size=.05, label=p_6, fontscale=1, position=above left, distance=-.16cm]{AI}
\Vertex[x=10, y=-2.5, Math, shape=circle, color=black, size=.05, label=p_6, fontscale=1, position=above right, distance=-.16cm]{AJ}
\Vertex[x=8, y=-4.5, Math, shape=circle, color=black, size=.05, label=p_6, fontscale=1, position=below left, distance=-.16cm]{AK}
\Vertex[x=10, y=-4.5, Math, shape=circle, color=black, size=.05, label=p_8, fontscale=1, position=below right, distance=-.16cm]{AL}

\Edge[Direct, color=black, lw=1pt, bend=-20](AI)(AJ)
\Edge[Direct, color=black, lw=1pt, bend=-20](AI)(AK)
\Edge[Direct, color=black, lw=1pt](AI)(AL)
\Edge[Direct, color=black, lw=1pt, bend=-20](AJ)(AI)
\Edge[Direct, color=black, lw=1pt, bend=-20](AJ)(AK)
\Edge[Direct, color=black, lw=1pt](AJ)(AL)
\Edge[Direct, color=black, lw=1pt, bend=-20](AK)(AI)
\Edge[Direct, color=black, lw=1pt, bend=-20](AK)(AJ)
\Edge[Direct, color=black, lw=1pt](AK)(AL)

\Vertex[x=12, y=-2.5, Math, shape=circle, color=black, size=.05, label=p_8, fontscale=1, position=above left, distance=-.16cm]{AM}
\Vertex[x=14, y=-2.5, Math, shape=circle, color=black, size=.05, label=p_8, fontscale=1, position=above right, distance=-.16cm]{AN}
\Vertex[x=12, y=-4.5, Math, shape=circle, color=black, size=.05, label=p_8, fontscale=1, position=below left, distance=-.16cm]{AO}
\Vertex[x=14, y=-4.5, Math, shape=circle, color=black, size=.05, label=p_8, fontscale=1, position=below right, distance=-.16cm]{AP}

\Edge[Direct, color=black, lw=1pt, bend=-20](AM)(AN)
\Edge[Direct, color=black, lw=1pt, bend=-20](AM)(AO)
\Edge[Direct, color=black, lw=1pt, bend=-20](AM)(AP)
\Edge[Direct, color=black, lw=1pt, bend=-20](AN)(AM)
\Edge[Direct, color=black, lw=1pt, bend=-20](AN)(AO)
\Edge[Direct, color=black, lw=1pt, bend=-20](AN)(AP)
\Edge[Direct, color=black, lw=1pt, bend=-20](AO)(AM)
\Edge[Direct, color=black, lw=1pt, bend=-20](AO)(AN)
\Edge[Direct, color=black, lw=1pt, bend=-20](AO)(AP)
\Edge[Direct, color=black, lw=1pt, bend=-20](AP)(AM)
\Edge[Direct, color=black, lw=1pt, bend=-20](AP)(AN)
\Edge[Direct, color=black, lw=1pt, bend=-20](AP)(AO)

\Text[x=1, y=-5.5]{$2p_1+p_6+p_7\leq 1\ (12)$}
\Text[x=5, y=-5.5]{$2p_4+2p_7\geq 1\ (-6)$}
\Text[x=9, y=-5.5]{$3p_6+p_8\geq 1\ (-4)$}
\Text[x=13, y=-5.5]{$4p_8\leq 1\ (1)$}

\end{tikzpicture}
\caption{Counterexamples for the existence of a randomized selection mechanism satisfying impartiality and weak unanimity for $n=3$ and for $n=4$. The graphs are depicted with the impartial probabilities assigned to each vertex with indegree at least one.
The inequalities included below each of them are obtained from imposing either that the probabilities sum up to at most 1---since at most one vertex is to be selected---or sum up to at least 1---due to weak unanimity. 
The values in parentheses show that the inequalities written for each set of graphs constitute an infeasible system because if one multiplies the corresponding inequality by them and sums the resulting inequalities, the contradiction $0\leq -1$ is achieved.}
\label{fig:counterexample-weak-unanimity}
\end{figure}

We are now ready to prove Theorems~\ref{thm:additive-lb} and~\ref{thm:additive-lb-abstentions}.
In order to be able to apply \autoref{lem:imposibility-selection} to deterministic mechanisms, we need a simple definition.
For a given deterministic selection mechanism $f\colon\calG_n\to 2^N$, let $f_{\text{rand}}\colon\calG_n\to [0,1]^n$ be the randomized selection mechanism such that $(f_{\text{rand}}(G))_v=1$ if $v\in f(G)$ and $(f_{\text{rand}}(G))_v=0$ otherwise.
It is then easy to see that whenever $f$ is impartial, $f_{\text{rand}}$ is impartial as well.
\begin{proof}[Proof of \autoref{thm:additive-lb}]

\begin{algorithm}[t]
	\SetAlgoNoLine
	\KwIn{Digraph $G=(N,E)\in \calG_{k+1}$, mechanism $f_{k}$ and integer $n\geq k+1$.}
	\KwOut{Set $S$ of selected vertices with $|S|\leq 1$.}
        $N' \gets \bigcup_{j=1}^{n-k-1}\{u_j\}$\;
        $H \gets (N\cup N', E)$\;
	{\bf return} $f_{k}(H)$
	\caption{Selection mechanism $f$ based on $f_{k}$.}
	\label{alg:add-isolated-vertices}
\end{algorithm}

The result is straightforward when $n=1$. Let $n\geq 2$ and $k \leq n-1$, and suppose that there is an impartial deterministic selection mechanism $f_{k}$ with
\[
    \Delta(G)-\delta^-(f_k(G), G) \leq k -1
\]
for every $G\in \calG_n(k)$.
We define the deterministic selection mechanism $f$ based on $f_{k}$ as specified in \autoref{alg:add-isolated-vertices}.
This mechanism receives a graph $G=(N,E)$ in $\calG_{k+1}$ and adds new vertices, if necessary to complete $n$ vertices. 
These vertices are isolated, in the sense that the set of edges in this new graph $H$ remains the same.
For every input graph $G$, the graph constructed belongs to $\calG_n(k)$, so the mechanism finally applies $f_{k}$.
We claim that $f_{\text{rand}}$ is impartial and weakly unanimous on $\calG^T_{k+1}$.
To see that $f_{\text{rand}}$ is weakly unanimous, observe that
\[
    \delta^-(v,H)=\delta^-(v,G) \text{ for every } v\in N, \text{ and } \delta^-(v,H)=0 \text{ for every } v\not\in N.
\]
For each $G\in \calG^T_{k+1}$, we have that $\Delta(G)=k$ and thus $\Delta(H)=k$. 
From the fact that $f_{k}$ is $(k-1)$-additive on $\calG_n(k)$, we conclude that $f$ returns a vertex $v^*$ of $H$ with $\delta^-(v^*,H)\geq k - (k-1) = 1$.
This implies, in the first place, that $v^*\in N$, thus $f$ is indeed a selection mechanism on $\calG^T_{k+1}$. Furthermore, we have $\delta^-(v^*,G)\geq 1$, thus
\[
\sum_{v\in G:  \delta^-(v)\geq 1}(f_{\text{rand}}(G))_v = 1,
\]
\ie $f_{\text{rand}}$ is weakly unanimous on $\calG^T_{k+1}$.
Impartiality of $f_{\text{rand}}$ is straightforward since $f_{k}$ is impartial and the set of edges is not modified in the mechanism.
This contradicts \autoref{lem:imposibility-selection}, so we conclude that mechanism $f_{k}$ cannot exist. 
\end{proof}

\begin{proof}[Proof of \autoref{thm:additive-lb-abstentions}]

\begin{algorithm}[t]
	\SetAlgoNoLine
	\KwIn{Digraph $G=(N,E)\in \calG_{k}$, mechanism $f^+_{k}$ and integer $n\geq k+1$.}
	\KwOut{Set $S$ of selected vertices with $|S|\leq 1$.}
	$N' \gets \bigcup_{j=1}^{n-k}\{u_j\}$\;
	$F \gets E\cup (N'\times N) \cup (\{v\in N:  \delta^+(v,G)=0\} \times \{u_1\})$\;
        $H\gets (N\cup N', F)$\;
	{\bf return} $f^+_{k}(H)$
	\caption{Selection mechanism $f$ based on $f^+_{k}$.}
	\label{alg:add-inneighbors}
\end{algorithm}

The result is straightforward when $n\in \{1,2\}$. Let $n\geq 3$ and $k\leq n-1$, and suppose that there is an impartial deterministic selection mechanism $f^+_{k}$ with
\[
    \Delta(G)-\delta^-(f^+_{k}(G), G) \leq k-2
\]
for every $G\in \calG^+_n(k)$.
We define the deterministic selection mechanism $f$ based on $f^+_{k}$ as specified in \autoref{alg:add-inneighbors}.
This mechanism requires a graph $G$ in $\calG_{k}$ and adds new vertices $N'=\{u_1,\ldots,u_{n-k}\}$, as well as edges from each of these vertices to every vertex of $G$, and edges from every vertex of $G$ with outdegree zero to $u_1$.
We first claim that for every input graph $G$, the graph $H$ constructed in the mechanism belongs to $\calG^+_n(k)$.
Indeed, since $G\in \calG_{k}$ every vertex $v\in N$ satisfies $\delta^+(v,G)\leq k-1$. Moreover, an outgoing edge to $u_1$ is added for every $v$ with $\delta^+(v,G)=0$, thus $1\leq \delta^-(v,H)\leq k-1$ for each $v\in N$.
On the other hand, each node in $N'$ has outdegree $k$.
Therefore, the mechanism is well defined, in the sense that in its last step it applies $f^+_{k}$ to a graph in $\calG^+_n(k)$.
We claim that $f_{\text{rand}}$ is impartial and weakly unanimous on $\calG^T_{k}$, which is a clear contradiction to \autoref{lem:imposibility-selection} and thus implies that mechanism $f^+_k$ cannot exist. We prove the claim in what follows.

For each $G=(N,E)\in \calG^T_{k}$, the set $\{v\in N:  \delta^+(v)=0\}$ is equal to $(S(G))_{r(G)}$ if $(s(G))_{r(G)}=1$, or to the empty set, otherwise. Therefore, for $u_1,\ldots,u_{n-k}$ and $H$ as defined in the mechanism we have that $\delta^-(u_1,H)\leq 1,\ \delta^-(u_j,H)=0$ for every $j\in \{2,\ldots,n-k\}$, and $\delta^-(v,H)=\delta^-(v,G)+n-k$ for every $v\in N$.
Since $\Delta(G)=k-1$ from the definition of the set $\calG^T_{k}$, we have that $\Delta(H)=n-1$.
Using that $f^+_{k}$ is $(k-2)$-additive on $\calG^+_n(k)$, we conclude that $f$ returns a vertex $v^*\in N\cup N'$ with $\delta^-(v^*,H) \geq n-1-(k-2) = n-k+1\geq 2$.
This implies, in the first place, that $v^*\in N$, thus $f$ is indeed a selection mechanism on $\calG^T_{k}$. Furthermore, we have $\delta^-(v^*,G)=\delta^-(v^*,H)-(n-k)\geq 1$, thus
\[
\sum_{v\in G:  \delta^-(v)\geq 1}(f_{\text{rand}}(G))_v = 1,
\]
\ie $f_{\text{rand}}$ is weakly unanimous on $\calG^T_{k}$.

To see that $f$ is impartial on $\calG^T_{k}$, let $G=(N,E), G'=(N,E')\in \calG^T_{k}$ and $\tilde{v}\in N$ be such that $E\setminus (\{\tilde{v}\}\times N) = E'\setminus (\{\tilde{v}\}\times N)$.
Denoting $F$ and $F'$ the edges of the graphs defined in the mechanism $f$ when run with input $G$ and $G'$, respectively, it is enough to show that $F\setminus (\{\tilde{v}\}\times (N\cup N')) = F'\setminus (\{\tilde{v}\}\times (N\cup N'))$, because this would imply 
\[
    f(G)\cap \{\tilde{v}\} = f^+_k(N\cup N',F)\cap \{\tilde{v}\} = f^+_k(N\cup N',F')\cap \{\tilde{v}\} = f(G')\cap \{\tilde{v}\}
\]
where the second equality holds since $f^+_k$ is impartial on $\calG^+_n(k)$ by hypothesis.
Indeed,
\begin{align*}
F\setminus (\{\tilde{v}\}\times (N\cup N')) & = (E\cup (N'\times N) \\
&\qquad \cup (\{v\in N:  \delta^+(v,G)=0\} \times \{u_1\})) \setminus (\{\tilde{v}\}\times (N\cup N'))\\
& = E\setminus (\{\tilde{v}\}\times N) \cup (N'\times N) \\
&\qquad \cup (\{v\in N\setminus \{\tilde{v}\}:  \delta^+(v,G)=0\} \times \{u_1\})\\
& = E'\setminus (\{\tilde{v}\}\times N) \cup (N'\times N) \\
&\qquad \cup (\{v\in N\setminus \{\tilde{v}\}:  \delta^+(v,G')=0\} \times \{u_1\})\\
& = (E'\cup (N'\times N) \\
&\qquad \cup (\{v\in N:  \delta^+(v,G')=0\} \times \{u_1\})) \setminus (\{\tilde{v}\}\times (N\cup N'))\\
& = F'\setminus (\{\tilde{v}\}\times (N\cup N')),
\end{align*}
where the third equality uses the fact that the outgoing edges of every vertex in $(N\cup N')\setminus \{\tilde{v}\}$ are the same in $G$ and $G'$.
This implies that $f_{\text{rand}}$ is impartial as well, concluding the proof of the claim and the proof of the theorem.
\end{proof}

Theorems~\ref{thm:additive-lb} and~\ref{thm:additive-lb-abstentions} provide tight bounds for the approval setting but have very weak implications for plurality. We end with small but nontrivial lower bounds for the latter. 
\begin{theorem}
    \label{thm:additive-plurality-lb}
    Let $n\in \NN$ and let $f$ be an impartial deterministic selection mechanism such that $f$ is $\alpha$-additive on $\calG_n(1)$. Then, $\alpha \geq \alpha_n$, where
    \[
        \alpha_n = \left\{ \begin{array}{ll}
             1 &   \text{if}\ \ n \in \{2,3\}, \\
             2 &   \text{if}\ \ 4\leq n \leq 9, \\
             3 &   \text{if}\ \ n \geq 10. \\
             \end{array}
   \right.
    \]
\end{theorem}

\begin{proof}
Denote $S_1=\{2,3\},\ S_2=\{4,5,\ldots,9\}$, and $S_3=\{10,11,\ldots\}$ for simplicity.
The graphs in \autoref{fig:counterexample-plurality-n2}, \autoref{fig:counterexample-plurality-n4}, and \autoref{fig:counterexample-plurality-n10} show the probabilities that impartial randomized selection mechanisms on $\calG_n(1)$ assign to each vertex for $n\in S_1,\ S_2$, and $S_3$, respectively. Vertices not shown in the figures have no incident edges.
The inequalities below each graph $G\in \calG_n(1)$ are obtained either from (i) imposing that the corresponding mechanism on this set selects with probability $1$ a vertex with indegree at least $\Delta(G)$, $\Delta(G)-1$, or $\Delta(G)-2$ for $n$ in $S_1,\ S_2$, or $S_3$, respectively, or (ii) the fact that the probabilities assigned to the vertices of a single graph add up to at most one.
The values in parentheses show that the inequalities written for each set of graphs constitute an infeasible system because if one multiplies the corresponding inequality by them and sums the resulting inequalities, the contradiction $0\leq -1$ is achieved.
Defining $\alpha_n$ as in the statement of the theorem, this shows that for $n\in \NN$ and for every randomized mechanism $h\colon \calG_n\xrightarrow{} [0,1]^n$, there exists $G=(N,E)\in \calG_n(1)$ such that
\[
    \sum_{v\in N:  \delta^-(v)\geq \Delta-(\alpha_n-1)}(h(G))_v < 1.
\]
\begin{figure}[t]
\centering
\begin{tikzpicture}[scale=0.8]
\Vertex[Math, shape=circle, color=black, size=.05, label=p_1, fontscale=1, position=left, distance=-.1cm]{A}
\Vertex[x=2, Math, shape=circle, color=black, size=.05, label=p_1, fontscale=1, position=right, distance=-.1cm]{B}
\Edge[Direct, color=black, lw=1pt, bend=-20](A)(B)
\Edge[Direct, color=black, lw=1pt, bend=-20](B)(A)
\Vertex[x=5, Math, shape=circle, color=black, size=.05, label=p_1, fontscale=1, position=left, distance=-.1cm]{AA}
\Vertex[x=7, Math, shape=circle, color=black, size=.05,]{AB}
\Edge[Direct, color=black, lw=1pt](AB)(AA)
\Text[x=1, y=-1]{\scriptsize{$2p_1\leq 1\ (1)$}}
\Text[x=6, y=-1]{\scriptsize{$p_1\geq 1\ (-2)$}}
\end{tikzpicture}
\caption{Proof of impossibility of the existence of an impartial randomized selection mechanism selecting vertices with indegree $\Delta(G)$ for each graph $G\in \calG_n(1)$ and $n\in\{2,3\}$. Scaling and summing the inequalities leads to $0\leq -1$, a contradiction.}
\label{fig:counterexample-plurality-n2}
\end{figure}
\begin{figure}[t]
\centering
\begin{tikzpicture}
\Vertex[y=2, Math, shape=circle, color=black, size=.05, label=p_1, fontscale=1, position=above left, distance=-.16cm]{AA}
\Vertex[x=2, y=2, Math, shape=circle, color=black, size=.05, label=p_1, fontscale=1, position=above right, distance=-.16cm]{AB}
\Vertex[Math, shape=circle, color=black, size=.05, label=p_1, fontscale=1, position=below left, distance=-.16cm]{AC}
\Vertex[x=2, Math, shape=circle, color=black, size=.05, label=p_1, fontscale=1, position=below right, distance=-.16cm]{AD}
\Edge[Direct, color=black, lw=1pt](AA)(AC)
\Edge[Direct, color=black, lw=1pt](AC)(AD)
\Edge[Direct, color=black, lw=1pt](AD)(AB)
\Edge[Direct, color=black, lw=1pt](AB)(AA)
\Vertex[x=4, y=2, Math, shape=circle, color=black, size=.05, label=p_1, fontscale=1, position=above left, distance=-.16cm]{BA}
\Vertex[x=6, y=2, Math, shape=circle, color=black, size=.05, label=p_3, fontscale=1, position=above right, distance=-.16cm]{BB}
\Vertex[x=4, Math, shape=circle, color=black, size=.05]{BC}
\Vertex[x=6, Math, shape=circle, color=black, size=.05, label=p_2, fontscale=1, position=below right, distance=-.16cm]{BD}
\Edge[Direct, color=black, lw=1pt](BA)(BD)
\Edge[Direct, color=black, lw=1pt](BC)(BD)
\Edge[Direct, color=black, lw=1pt](BD)(BB)
\Edge[Direct, color=black, lw=1pt](BB)(BA)
\Vertex[x=8, y=2, Math, shape=circle, color=black, size=.05, label=p_2, fontscale=1, position=above left, distance=-.16cm]{CA}
\Vertex[x=10, y=2, Math, shape=circle, color=black, size=.05]{CB}
\Vertex[x=8, Math, shape=circle, color=black, size=.05]{CC}
\Vertex[x=10, Math, shape=circle, color=black, size=.05, label=p_2, fontscale=1, position=below right, distance=-.16cm]{CD}
\Edge[Direct, color=black, lw=1pt, bend=-20](CA)(CD)
\Edge[Direct, color=black, lw=1pt](CC)(CD)
\Edge[Direct, color=black, lw=1pt, bend=-20](CD)(CA)
\Edge[Direct, color=black, lw=1pt](CB)(CA)
\Vertex[x=12, y=2, Math, shape=circle, color=black, size=.05]{DA}
\Vertex[x=14, y=2, Math, shape=circle, color=black, size=.05, label=p_3, fontscale=1, position=above right, distance=-.16cm]{DB}
\Vertex[x=12, Math, shape=circle, color=black, size=.05]{DC}
\Vertex[x=14, Math, shape=circle, color=black, size=.05, label=p_4, fontscale=1, position=below right, distance=-.16cm]{DD}
\Edge[Direct, color=black, lw=1pt](DA)(DD)
\Edge[Direct, color=black, lw=1pt](DC)(DD)
\Edge[Direct, color=black, lw=1pt, bend=-20](DD)(DB)
\Edge[Direct, color=black, lw=1pt, bend=-20](DB)(DD)
\Text[x=1, y=-1]{\scriptsize{$4p_1\leq 1\ (1)$}}
\Text[x=5, y=-1]{\scriptsize{$p_1+p_2+p_3\geq 1\ (-4)$}}
\Text[x=9, y=-1]{\scriptsize{$2p_2\leq 1\ (2)$}}
\Text[x=13, y=-.7]{\scriptsize{$p_3+p_4\leq 1\ (4)$}}
\Text[x=13, y=-1.3]{\scriptsize{$p_4\geq 1\ (-4)$}}
\end{tikzpicture}
\caption{ Proof of impossibility of the existence of an impartial randomized selection mechanism selecting vertices with indegree at least $\Delta(G)-1$ for each graph $G\in \calG_n(1)$ and $4\leq n \leq 9$. Scaling and summing the inequalities leads to $0\leq -1$, a contradiction.}
\label{fig:counterexample-plurality-n4}
\end{figure}
\begin{figure}[!t]
\centering
\begin{tikzpicture}
\Vertex[Math, shape=circle, color=black, size=.05, label=p_2, fontscale=1, position=above left, distance=-.16cm]{AA}
\Vertex[x=1, Math, shape=circle, color=black, size=.05, label=p_2, fontscale=1, position=above right, distance=-.16cm]{AB}
\Vertex[x=1, y=1, Math, shape=circle, color=black, size=.05, label=p_1, fontscale=1, position=below right, distance=-.16cm]{AC}
\Vertex[y=1, Math, shape=circle, color=black, size=.05, label=p_1, fontscale=1, position=below left, distance=-.16cm]{AD}
\Vertex[x=-.8, y=-.5, Math, shape=circle, color=black, size=.05]{AE}
\Vertex[x=1.8, y=-.5, Math, shape=circle, color=black, size=.05]{AF}
\Vertex[x=1.8, y=1.5, Math, shape=circle, color=black, size=.05]{AG}
\Vertex[x=-.8, y=1.5, Math, shape=circle, color=black, size=.05]{AH}
\Vertex[x=.5, y=-.866, Math, shape=circle, color=black, size=.05, label=p_2, fontscale=1, position=left, distance=-.08cm]{AI}
\Vertex[x=.5, y=-1.866, Math, shape=circle, color=black, size=.05]{AJ}
\Edge[Direct, color=black, lw=1pt](AA)(AB)
\Edge[Direct, color=black, lw=1pt](AB)(AI)
\Edge[Direct, color=black, lw=1pt, bend=-20](AC)(AD)
\Edge[Direct, color=black, lw=1pt, bend=-20](AD)(AC)
\Edge[Direct, color=black, lw=1pt](AE)(AA)
\Edge[Direct, color=black, lw=1pt](AF)(AB)
\Edge[Direct, color=black, lw=1pt](AG)(AC)
\Edge[Direct, color=black, lw=1pt](AH)(AD)
\Edge[Direct, color=black, lw=1pt](AI)(AA)
\Edge[Direct, color=black, lw=1pt](AJ)(AI)
\Vertex[x=4, Math, shape=circle, color=black, size=.05, label=p_6, fontscale=1, position=above left, distance=-.16cm]{BA}
\Vertex[x=5, Math, shape=circle, color=black, size=.05, label=p_4, fontscale=1, position=above right, distance=-.16cm]{BB}
\Vertex[x=5, y=1, Math, shape=circle, color=black, size=.05, label=p_3, fontscale=1, position=below right, distance=-.16cm]{BC}
\Vertex[x=4, y=1, Math, shape=circle, color=black, size=.05, label=p_1, fontscale=1, position=below left, distance=-.16cm]{BD}
\Vertex[x=3.2, y=-.5, Math, shape=circle, color=black, size=.05]{BE}
\Vertex[x=5.8, y=-.5, Math, shape=circle, color=black, size=.05]{BF}
\Vertex[x=5.8, y=1.5, Math, shape=circle, color=black, size=.05]{BG}
\Vertex[x=3.2, y=1.5, Math, shape=circle, color=black, size=.05]{BH}
\Vertex[x=4.5, y=-.866, Math, shape=circle, color=black, size=.05, label=p_5, fontscale=1, position=left, distance=-.08cm]{BI}
\Vertex[x=4.5, y=-1.866, Math, shape=circle, color=black, size=.05]{BJ}
\Edge[Direct, color=black, lw=1pt](BA)(BB)
\Edge[Direct, color=black, lw=1pt](BB)(BI)
\Edge[Direct, color=black, lw=1pt](BC)(BD)
\Edge[Direct, color=black, lw=1pt](BD)(BA)
\Edge[Direct, color=black, lw=1pt](BE)(BA)
\Edge[Direct, color=black, lw=1pt](BF)(BB)
\Edge[Direct, color=black, lw=1pt](BG)(BC)
\Edge[Direct, color=black, lw=1pt](BH)(BD)
\Edge[Direct, color=black, lw=1pt](BI)(BA)
\Edge[Direct, color=black, lw=1pt](BJ)(BI)
\Vertex[x=8, Math, shape=circle, color=black, size=.05, label=p_2, fontscale=1, position=above left, distance=-.16cm]{CA}
\Vertex[x=9, Math, shape=circle, color=black, size=.05, label=p_9, fontscale=1, position=above right, distance=-.16cm]{CB}
\Vertex[x=9, y=1, Math, shape=circle, color=black, size=.05, label=p_5, fontscale=1, position=below right, distance=-.16cm]{CC}
\Vertex[x=8, y=1, Math, shape=circle, color=black, size=.05, label=p_7, fontscale=1, position=below left, distance=-.16cm]{CD}
\Vertex[x=7.2, y=-.5, Math, shape=circle, color=black, size=.05]{CE}
\Vertex[x=9.8, y=-.5, Math, shape=circle, color=black, size=.05]{CF}
\Vertex[x=9.8, y=1.5, Math, shape=circle, color=black, size=.05]{CG}
\Vertex[x=7.2, y=1.5, Math, shape=circle, color=black, size=.05]{CH}
\Vertex[x=8.5, y=-.866, Math, shape=circle, color=black, size=.05, label=p_8, fontscale=1, position=left, distance=-.08cm]{CI}
\Vertex[x=8.5, y=-1.866, Math, shape=circle, color=black, size=.05]{CJ}
\Edge[Direct, color=black, lw=1pt](CA)(CD)
\Edge[Direct, color=black, lw=1pt](CB)(CI)
\Edge[Direct, color=black, lw=1pt, bend=-20](CC)(CD)
\Edge[Direct, color=black, lw=1pt, bend=-20](CD)(CC)
\Edge[Direct, color=black, lw=1pt](CE)(CA)
\Edge[Direct, color=black, lw=1pt](CF)(CB)
\Edge[Direct, color=black, lw=1pt](CG)(CC)
\Edge[Direct, color=black, lw=1pt](CH)(CD)
\Edge[Direct, color=black, lw=1pt](CI)(CA)
\Edge[Direct, color=black, lw=1pt](CJ)(CI)
\Vertex[x=12, Math, shape=circle, color=black, size=.05, label=p_{12}, fontscale=1, position=above left, distance=-.16cm]{DA}
\Vertex[x=13, Math, shape=circle, color=black, size=.05, label=p_{11}, fontscale=1, position=above right, distance=-.16cm]{DB}
\Vertex[x=13, y=1, Math, shape=circle, color=black, size=.05, label=p_3, fontscale=1, position=below right, distance=-.16cm]{DC}
\Vertex[x=12, y=1, Math, shape=circle, color=black, size=.05, label=p_3, fontscale=1, position=below left, distance=-.16cm]{DD}
\Vertex[x=11.2, y=-.5, Math, shape=circle, color=black, size=.05]{DE}
\Vertex[x=13.8, y=-.5, Math, shape=circle, color=black, size=.05]{DF}
\Vertex[x=13.8, y=1.5, Math, shape=circle, color=black, size=.05]{DG}
\Vertex[x=11.2, y=1.5, Math, shape=circle, color=black, size=.05]{DH}
\Vertex[x=12.5, y=-.866, Math, shape=circle, color=black, size=.05, label=p_{10}, fontscale=1, position=left, distance=-.08cm]{DI}
\Vertex[x=12.5, y=-1.866, Math, shape=circle, color=black, size=.05]{DJ}
\Edge[Direct, color=black, lw=1pt](DA)(DB)
\Edge[Direct, color=black, lw=1pt](DB)(DI)
\Edge[Direct, color=black, lw=1pt](DC)(DA)
\Edge[Direct, color=black, lw=1pt](DD)(DA)
\Edge[Direct, color=black, lw=1pt](DE)(DA)
\Edge[Direct, color=black, lw=1pt](DF)(DB)
\Edge[Direct, color=black, lw=1pt](DG)(DC)
\Edge[Direct, color=black, lw=1pt](DH)(DD)
\Edge[Direct, color=black, lw=1pt](DI)(DA)
\Edge[Direct, color=black, lw=1pt](DJ)(DI)
\Text[x=.5, y=-2.7]{\scriptsize{$2p_1+3p_2\leq 1$\ (1)}}
\Text[x=4, y=-2.7]{\scriptsize{$p_1+p_3+p_4+p_5+p_6\geq 1\ (-2)$}}
\Text[x=8.5, y=-2.7]{\scriptsize{$p_2+p_5+p_7+p_8+p_9\geq 1\ (-3)$}}
\Text[x=12.5, y=-2.4]{\scriptsize{$2p_3+p_{10}+p_{11}+p_{12}\leq 1\ (5)$}}
\Text[x=12.5, y=-3]{\scriptsize{$p_{10}+p_{11}+p_{12}\geq 1\ (-10)$}}
\Vertex[y=-5.5, Math, shape=circle, color=black, size=.05, label=p_7, fontscale=1, position=above left, distance=-.16cm]{AA}
\Vertex[x=1, y=-5.5, Math, shape=circle, color=black, size=.05, label=p_{13}, fontscale=1, position=above right, distance=-.16cm]{AB}
\Vertex[x=1, y=-4.5, Math, shape=circle, color=black, size=.05, label=p_3, fontscale=1, position=below right, distance=-.16cm]{AC}
\Vertex[y=-4.5, Math, shape=circle, color=black, size=.05, label=p_3, fontscale=1, position=below left, distance=-.16cm]{AD}
\Vertex[x=-.8, y=-6, Math, shape=circle, color=black, size=.05]{AE}
\Vertex[x=1.8, y=-6, Math, shape=circle, color=black, size=.05]{AF}
\Vertex[x=1.8, y=-4, Math, shape=circle, color=black, size=.05]{AG}
\Vertex[x=-.8, y=-4, Math, shape=circle, color=black, size=.05]{AH}
\Vertex[x=.5, y=-6.366, Math, shape=circle, color=black, size=.05, label=p_4, fontscale=1, position=left, distance=-.08cm]{AI}
\Vertex[x=.5, y=-7.366, Math, shape=circle, color=black, size=.05]{AJ}
\Edge[Direct, color=black, lw=1pt](AA)(AB)
\Edge[Direct, color=black, lw=1pt](AB)(AI)
\Edge[Direct, color=black, lw=1pt](AC)(AB)
\Edge[Direct, color=black, lw=1pt](AD)(AA)
\Edge[Direct, color=black, lw=1pt](AE)(AA)
\Edge[Direct, color=black, lw=1pt](AF)(AB)
\Edge[Direct, color=black, lw=1pt](AG)(AC)
\Edge[Direct, color=black, lw=1pt](AH)(AD)
\Edge[Direct, color=black, lw=1pt](AI)(AA)
\Edge[Direct, color=black, lw=1pt](AJ)(AI)
\Vertex[x=4, y=-5.5, Math, shape=circle, color=black, size=.05, label=p_7, fontscale=1, position=above left, distance=-.16cm]{BA}
\Vertex[x=5, y=-5.5, Math, shape=circle, color=black, size=.05, label=p_5, fontscale=1, position=above right, distance=-.16cm]{BB}
\Vertex[x=5, y=-4.5, Math, shape=circle, color=black, size=.05, label=p_{10}, fontscale=1, position=below right, distance=-.16cm]{BC}
\Vertex[x=4, y=-4.5, Math, shape=circle, color=black, size=.05, label=p_4, fontscale=1, position=below left, distance=-.16cm]{BD}
\Vertex[x=3.2, y=-6, Math, shape=circle, color=black, size=.05]{BE}
\Vertex[x=5.8, y=-6, Math, shape=circle, color=black, size=.05]{BF}
\Vertex[x=5.8, y=-4, Math, shape=circle, color=black, size=.05]{BG}
\Vertex[x=3.2, y=-4, Math, shape=circle, color=black, size=.05]{BH}
\Vertex[x=4.5, y=-6.366, Math, shape=circle, color=black, size=.05]{BI}
\Vertex[x=4.5, y=-7.366, Math, shape=circle, color=black, size=.05]{BJ}
\Edge[Direct, color=black, lw=1pt](BA)(BD)
\Edge[Direct, color=black, lw=1pt](BB)(BA)
\Edge[Direct, color=black, lw=1pt](BC)(BB)
\Edge[Direct, color=black, lw=1pt](BD)(BC)
\Edge[Direct, color=black, lw=1pt](BE)(BA)
\Edge[Direct, color=black, lw=1pt](BF)(BB)
\Edge[Direct, color=black, lw=1pt](BG)(BC)
\Edge[Direct, color=black, lw=1pt](BH)(BD)
\Edge[Direct, color=black, lw=1pt](BI)(BA)
\Edge[Direct, color=black, lw=1pt](BJ)(BI)
\Vertex[x=8, y=-5.5, Math, shape=circle, color=black, size=.05, label=p_9, fontscale=1, position=above left, distance=-.16cm]{CA}
\Vertex[x=9, y=-5.5, Math, shape=circle, color=black, size=.05, label=p_8, fontscale=1, position=above right, distance=-.16cm]{CB}
\Vertex[x=9, y=-4.5, Math, shape=circle, color=black, size=.05, label=p_4, fontscale=1, position=below right, distance=-.16cm]{CC}
\Vertex[x=8, y=-4.5, Math, shape=circle, color=black, size=.05, label=p_{12}, fontscale=1, position=below left, distance=-.16cm]{CD}
\Vertex[x=7.2, y=-6, Math, shape=circle, color=black, size=.05]{CE}
\Vertex[x=9.8, y=-6, Math, shape=circle, color=black, size=.05]{CF}
\Vertex[x=9.8, y=-4, Math, shape=circle, color=black, size=.05]{CG}
\Vertex[x=7.2, y=-4, Math, shape=circle, color=black, size=.05]{CH}
\Vertex[x=8.5, y=-6.366, , Math, shape=circle, color=black, size=.05]{CI}
\Vertex[x=8.5, y=-7.366, , Math, shape=circle, color=black, size=.05]{CJ}
\Edge[Direct, color=black, lw=1pt](CA)(CD)
\Edge[Direct, color=black, lw=1pt](CB)(CD)
\Edge[Direct, color=black, lw=1pt, bend=-20](CC)(CD)
\Edge[Direct, color=black, lw=1pt, bend=-20](CD)(CC)
\Edge[Direct, color=black, lw=1pt](CE)(CA)
\Edge[Direct, color=black, lw=1pt](CF)(CB)
\Edge[Direct, color=black, lw=1pt](CG)(CC)
\Edge[Direct, color=black, lw=1pt](CH)(CD)
\Edge[Direct, color=black, lw=1pt](CI)(CB)
\Edge[Direct, color=black, lw=1pt](CJ)(CI)
\Vertex[x=12, y=-5.5, Math, shape=circle, color=black, size=.05, label=p_9, fontscale=1, position=above left, distance=-.16cm]{DA}
\Vertex[x=13, y=-5.5, Math, shape=circle, color=black, size=.05, label=p_8, fontscale=1, position=above right, distance=-.16cm]{DB}
\Vertex[x=13, y=-4.5, Math, shape=circle, color=black, size=.05, label=p_6, fontscale=1, position=below right, distance=-.16cm]{DC}
\Vertex[x=12, y=-4.5, Math, shape=circle, color=black, size=.05, label=p_7, fontscale=1, position=below left, distance=-.16cm]{DD}
\Vertex[x=11.2, y=-6, Math, shape=circle, color=black, size=.05]{DE}
\Vertex[x=13.8, y=-6, Math, shape=circle, color=black, size=.05]{DF}
\Vertex[x=13.8, y=-4, Math, shape=circle, color=black, size=.05]{DG}
\Vertex[x=11.2, y=-4, Math, shape=circle, color=black, size=.05]{DH}
\Vertex[x=12.5, y=-6.366, Math, shape=circle, color=black, size=.05]{DI}
\Vertex[x=12.5, y=-7.366, Math, shape=circle, color=black, size=.05]{DJ}
\Edge[Direct, color=black, lw=1pt](DA)(DD)
\Edge[Direct, color=black, lw=1pt](DB)(DC)
\Edge[Direct, color=black, lw=1pt, bend=-20](DC)(DD)
\Edge[Direct, color=black, lw=1pt, bend=-20](DD)(DC)
\Edge[Direct, color=black, lw=1pt](DE)(DA)
\Edge[Direct, color=black, lw=1pt](DF)(DB)
\Edge[Direct, color=black, lw=1pt](DG)(DC)
\Edge[Direct, color=black, lw=1pt](DH)(DD)
\Edge[Direct, color=black, lw=1pt](DI)(DB)
\Edge[Direct, color=black, lw=1pt](DJ)(DI)
\Text[x=.5, y=-8.2]{\scriptsize{$2p_3+p_4+p_7+p_{13}\geq 1\ (-4)$}}
\Text[x=4.5, y=-8.2]{\scriptsize{$p_4+p_5+p_7+p_{10}\leq 1\ (5)$}}
\Text[x=8.5, y=-8.2]{\scriptsize{$p_4+p_8+p_9+p_{12}\leq 1\ (1)$}}
\Text[x=12.5, y=-8.2]{\scriptsize{$p_6+p_7+p_8+p_9\leq 1\ (2)$}}
\Vertex[x=4, y=-11, Math, shape=circle, color=black, size=.05]{BA}
\Vertex[x=5, y=-11, Math, shape=circle, color=black, size=.05]{BB}
\Vertex[x=5, y=-10, Math, shape=circle, color=black, size=.05, label=p_{11}, fontscale=1, position=below right, distance=-.16cm]{BC}
\Vertex[x=4, y=-10, Math, shape=circle, color=black, size=.05, label=p_{14}, fontscale=1, position=below left, distance=-.16cm]{BD}
\Vertex[x=3.2, y=-11.5, Math, shape=circle, color=black, size=.05]{BE}
\Vertex[x=5.8, y=-11.5, Math, shape=circle, color=black, size=.05]{BF}
\Vertex[x=5.8, y=-9.5, Math, shape=circle, color=black, size=.05]{BG}
\Vertex[x=3.2, y=-9.5, Math, shape=circle, color=black, size=.05]{BH}
\Vertex[x=4.5, y=-11.866, Math, shape=circle, color=black, size=.05]{BI}
\Vertex[x=4.5, y=-12.866, Math, shape=circle, color=black, size=.05]{BJ}
\Edge[Direct, color=black, lw=1pt](BA)(BD)
\Edge[Direct, color=black, lw=1pt](BB)(BD)
\Edge[Direct, color=black, lw=1pt, bend=-20](BC)(BD)
\Edge[Direct, color=black, lw=1pt, bend=-20](BD)(BC)
\Edge[Direct, color=black, lw=1pt](BE)(BA)
\Edge[Direct, color=black, lw=1pt](BF)(BB)
\Edge[Direct, color=black, lw=1pt](BG)(BC)
\Edge[Direct, color=black, lw=1pt](BH)(BD)
\Edge[Direct, color=black, lw=1pt](BI)(BD)
\Edge[Direct, color=black, lw=1pt](BJ)(BI)
\Vertex[x=8, y=-11, Math, shape=circle, color=black, size=.05]{CA}
\Vertex[x=9, y=-11, Math, shape=circle, color=black, size=.05]{CB}
\Vertex[x=9, y=-10, Math, shape=circle, color=black, size=.05, label=p_{13}, fontscale=1, position=below right, distance=-.16cm]{CC}
\Vertex[x=8, y=-10, Math, shape=circle, color=black, size=.05, label=p_{12}, fontscale=1, position=below left, distance=-.16cm]{CD}
\Vertex[x=7.2, y=-11.5, Math, shape=circle, color=black, size=.05]{CE}
\Vertex[x=9.8, y=-11.5, Math, shape=circle, color=black, size=.05]{CF}
\Vertex[x=9.8, y=-9.5, Math, shape=circle, color=black, size=.05]{CG}
\Vertex[x=7.2, y=-9.5, Math, shape=circle, color=black, size=.05]{CH}
\Vertex[x=8.5, y=-11.866, , Math, shape=circle, color=black, size=.05]{CI}
\Vertex[x=8.5, y=-12.866, , Math, shape=circle, color=black, size=.05]{CJ}
\Edge[Direct, color=black, lw=1pt](CA)(CD)
\Edge[Direct, color=black, lw=1pt](CB)(CC)
\Edge[Direct, color=black, lw=1pt, bend=-20](CC)(CD)
\Edge[Direct, color=black, lw=1pt, bend=-20](CD)(CC)
\Edge[Direct, color=black, lw=1pt](CE)(CA)
\Edge[Direct, color=black, lw=1pt](CF)(CB)
\Edge[Direct, color=black, lw=1pt](CG)(CC)
\Edge[Direct, color=black, lw=1pt](CH)(CD)
\Edge[Direct, color=black, lw=1pt](CI)(CD)
\Edge[Direct, color=black, lw=1pt](CJ)(CI)
\Text[x=4.5, y=-13.4]{\scriptsize{$p_{11}+p_{14}\leq 1\ (5)$}}
\Text[x=4.5, y=-14]{\scriptsize{$p_{14}\geq 1\ (-5)$}}
\Text[x=8.5, y=-13.7]{\scriptsize{$p_{12}+p_{13}\leq 1\ (4)$}}
\end{tikzpicture}
\caption{ Proof of impossibility of the existence of an impartial randomized selection mechanism selecting vertices with indegree at least $\Delta(G)-2$ for each graph $G\in \calG_n(1)$ and $n\geq 10$. Scaling and summing the inequalities leads to $0\leq -1$, a contradiction.}
\label{fig:counterexample-plurality-n10}
\end{figure}

Suppose now for contradiction that there exists $n\in \NN$ and $\alpha \leq \alpha_n-1$ such that there is an impartial deterministic selection mechanism $f$ which is $\alpha$-additive on $\calG_n(1)$. It is clear that $f_{\text{rand}}$ is also impartial and, moreover, for every graph $G=(N,E)\in \calG_n(1)$ with $\Delta(G)\geq \alpha_n$,
\[
\sum_{v\in N:  \delta^-(v)\geq \Delta-(\alpha_n-1)}(f_{\text{rand}}(G))_v = 1,
\]
since $f(G)=v^*$ with $\delta^-(v^*)\geq \Delta-(\alpha_n-1)$.
But this contradicts the non-existence of such impartial randomized selection mechanisms.
We conclude that for every $n\in \NN$, if $f$ is an impartial deterministic selection mechanism which is $\alpha$-additive on $\calG_n(1)$, then $\alpha \geq \alpha_n$.\qedhere
\end{proof}

It is worth noting that these bounds apply to the setting without abstentions as well, with the only difference that any deterministic selection mechanism is impartial and $0$-additive on $\calG^+_2(1)$, thus the lower bound of 1 does not hold for $n=2$.

\newpage
\appendix

\section{Proof of Lemma~\ref{lem:indegree-changes}}
\label{app:indegree-changes}
Let $G$ and $v$ be as defined in the statement of the lemma.
We prove the existence of vertices as claimed by induction over $j$. 
For the base case, suppose that for every $u\in N^-(v)$ it holds 
$(\delta^{*}(u),u)\prec(\delta^-(v), v).$
From the definition of the mechanism we thus have that $i^{*}(u)>i^{*}(v)$ for every $u\in N^-(v)$ and therefore $\delta^{*}(v)=\delta^-(v)$, which contradicts the hypothesis of the lemma.

Now, let $j'\in [r-2]_0$ and assume that for every $j\in [j']_0$ there is a vertex $u_j$ such that $(u_j,v)\in E$ and 
$(\delta^{*}(u_j),u_j)\succ(\delta^-(v)-j, v)$.
Suppose that for every vertex $u\in N^-(v) \setminus \{u_1,\ldots,u_{j'}\}$ it holds
$(\delta^{*}(u),u)\prec(\delta^-(v)-(j'+1), v)$.
The expression $\delta^-(v)-(j'+1)$ is exactly the indegree of $v$ after deleting the incoming edges from $u_0,\ldots,u_{j'}$, thus from the definition of the mechanism we have that $i^{*}(u)>i^{*}(v)$ for every $u\in N^-(v) \setminus \{u_0,\ldots,u_{j'}\}$.
This yields
\[
\delta^{*}(v)\geq \delta^-(v)-(j'+1)\geq \delta^-(v)-(r-1),
\]
implying $r\geq \delta^-(v)-\delta^{*}(v)+1$.
In the case $v>z$, this is equivalent to $d\leq \delta^{*}(v)$, but these two inequalities contradict $(d, z)\succeq (\delta^{*}(v), v)$.
If $v<z$, on the other hand, it is equivalent to $d\leq \delta^{*}(v)-1$, which is again a contradiction. This concludes the existence of vertices $u_0,\ldots,u_{r-1}$ as in the statement of the lemma.

To prove the last claim, we fix $(d,z)=(\delta^{*}(v),v)$. We know that taking $r=\delta^-(v)-\delta^{*}(v)$ there are vertices such that for each $j\in [r-1]_0$, $(u_{j},v)\in E$ and $(\delta^{*}(u_{j}),u_j)\succ(\delta^-(v)-j, v)$. 
Suppose that there is a vertex $u\in N^-(v)\setminus \{u_0,\ldots,u_{r-1}\}$ with $(\delta^{*}(u),u)\succ(\delta^{*}(v),v)$.
It is clear that 
\[
|\{u\in N^-(v):  (\delta^{*}(u),u)\succ(\delta^{*}(v),v)\}| \geq r+1,
\]
thus by expressions~\eqref{eq:vertex-descent-characterization-iterations} and~\eqref{eq:equivalence-iterations-indegrees} we conclude that $\delta^-(v)-\delta^{*}(v)\geq r+1$, which is a contradiction.\qed

\section{Proof of Lemma~\ref{lem:indegree-changes-two-graphs}}
\label{app:indegree-changes-two-graphs}

Since $E_1\setminus(\{u\}\times N)=E_2\setminus (\{u\}\times N)$, it is clear that
\begin{equation}
\delta^-(v,G_2)\geq \delta^-(v,G_1)-1 = d+r+(\delta^-(v,G_1)-\delta^{*}(v,G_1))-1 - \chi(v>z). \label{v-indegree-changes}
\end{equation}
In the following, we denote $r'=r+(\delta^-(v,G_1)-\delta^{*}(v,G_1))$ for ease of notation.
If $\delta^-(v,G_2)= \delta^-(v,G_1)$, then $r'= \delta^-(v,G_2)-d+\chi(v>z)$. 
In addition, 
\[
    (\delta^-(v,G_2),v)=(\delta^-(v,G_1),v)\succ(d, z)\succeq (\delta^{*}(v,G_2), v),
\]
thus from \autoref{lem:indegree-changes} there are vertices $u'_0,\ldots,u'_{r'-1}$ such that for every $j\in [r'-1]_0$ we have that $(u'_{j},v)\in E_1\cap E_2$---since $\delta^-(v,G_2)= \delta^-(v,G_1)$ we have $E_1\cap N^-(v)=E_2\cap N^-(v)$---and $(\delta^{*}(u'_{j},G_2),u'_{j})\succ (\delta^-(v,G_2)-j,v)$.
From expression~\eqref{eq:vertex-descent-characterization-indegrees}, we know that 
\[
|\{u \in \{u'_0,\ldots,u'_{r'-1}\}:  (\delta^{*}(u,G_1),u)\succ(\delta^{*}(v,G_1),v)\}|\leq \delta^-(v,G_1)-\delta^{*}(v,G_1),
\]
thus it is possible to take $\{u_0,\ldots,u_{r-1}\}\subseteq \{u'_0,\ldots,u'_{r'-1}\}$ with the property that the inequality $(\delta^{*}(u_j,G_1),u_j)\prec(\delta^{*}(v,G_1),v)$ also holds for $j\in [r-1]_0$. 
Moreover,
\begin{equation}
    (\delta^{*}(u_j,G_2),u_j)\succ(\delta^-(v,G_2)-j-((\delta^-(v,G_1)-\delta^{*}(v,G_1))),v)=(\delta^{*}(v,G_1)-j,v).\label{eq:inneighbors-indegree}
\end{equation}
Relabeling these vertices so that $(\delta^{*}(u_0,G_2),u_0)\succ\dots \succ (\delta^{*}(u_{r-1},G_2),u_{r-1})$ it is clear that the previous inequalities still hold, and in particular $$(\delta^{*}(u_0,G_2),u_0)\succ (\delta^{*}(v,G_1),v) \succ (\delta^{*}(u_0,G_1),u_0).$$
We conclude that Condition~\ref{indegree-changes-two-graphs-alt2} holds.

Similarly, if $\delta^-(v,G_2)= \delta^-(v,G_1)+1$, then $r'+1= \delta^-(v,G_2)-d+\chi(v>z)$. 
As before, 
$(\delta^-(v,G_2),v)\succ(d, z)\succeq (\delta^{*}(v,G_2), v)$,
thus from \autoref{lem:indegree-changes} there are vertices $u'_0,\ldots,u'_{r'}$ such that for every $j\in [r']_0$ we have that $(u'_{j},v)\in E_2$ and $$(\delta^{*}(u'_{j},G_2),u'_{j})\succ (\delta^-(v,G_2)-j,v).$$ 
But $|(E_2\cap N^-(v))\setminus (E_1\cap N^-(v))|= 1$, thus at least $r'$ of these vertices $u'_{j}$ are such that $(u'_{j},v)\in E_1\cap E_2$. 
As before, from expression~\eqref{eq:vertex-descent-characterization-indegrees} 
we conclude that it is possible to take $\{u_0,\ldots,u_{r-1}\}\subseteq \{u'_0,\ldots,u'_{r'}\}$ such that $(u_j,v)\in E_1\cap E_2$ and $(\delta^{*}(u_j,G_1),u_j)\prec(\delta^{*}(v,G_1),v)$ for every $j\in [r-1]_0$. 
Since expression \eqref{eq:inneighbors-indegree} still holds, relabeling the vertices such that $(\delta^{*}(u_0,G_2),u_0)\succ\dots \succ (\delta^{*}(u_{r-1},G_2),u_{r-1})$ we conclude the lemma with Condition~\ref{indegree-changes-two-graphs-alt2} as well. An example for each of the cases addressed so far is included in \autoref{fig:lemma-2.1}.

\begin{figure}[t]
\centering
\begin{tikzpicture}[scale=0.88]

\Vertex[y=1.8, Math, shape=circle, color=black, , size=.05, label=\tilde{v}, fontscale=1, position=above left, distance=-.17cm]{A}
\Vertex[x=.7, y=1.2, Math, shape=circle, color=black, , size=.05]{B}
\Vertex[x=1.4, y=1.2, Math, shape=circle, color=black, size=.05, label=v, fontscale=1, position=above, distance=-.08cm]{C}
\Vertex[x=1.4, y=.6, Math, shape=circle, color=black, size=.05]{D}
\Vertex[x=2.1, y=.6, Math, shape=circle, color=black, size=.05, label=u_0, fontscale=1, position=above right, distance=-.17cm]{E}

\Edge[Direct, color=black, lw=1pt](B)(C)
\Edge[Direct, color=black, lw=1pt](E)(C)
\Edge[Direct, style=dashed, color=black, lw=1pt, bend=-30](C)(D)

\draw[] (-.2,-.2) -- (2.3,-.2);
\draw[] (-.2,.4) -- (2.3,.4);
\draw[] (-.2,1) -- (2.3,1);
\draw[] (-.2,1.6) -- (2.3,1.6);
\draw[] (-.2,2.2) -- (2.3,2.2);

\Vertex[x=4, y=1.8, Math, shape=circle, color=black, , size=.05, label=\tilde{v}, fontscale=1, position=above left, distance=-.17cm]{F}
\Vertex[x=4.7, y=1.2, Math, shape=circle, color=black, , size=.05]{G}
\Vertex[x=5.4, y=1.8, Math, shape=circle, color=black, size=.05, label=v, fontscale=1, position=above, distance=-.08cm]{H}
\Vertex[x=5.4, Math, shape=circle, color=black, size=.05]{I}
\Vertex[x=6.1, y=1.2, Math, shape=circle, color=black, size=.05, label=u_0, fontscale=1, position=above right, distance=-.17cm]{J}

\Edge[Direct, color=black, lw=1pt](F)(H)
\Edge[Direct, color=black, lw=1pt](G)(H)
\Edge[Direct, color=black, lw=1pt](J)(H)
\Edge[Direct, style=dashed, color=black, lw=1pt](H)(I)

\draw[] (3.8,-.2) -- (6.3,-.2);
\draw[] (3.8,.4) -- (6.3,.4);
\draw[] (3.8,1) -- (6.3,1);
\draw[] (3.8,1.6) -- (6.3,1.6);
\draw[] (3.8,2.2) -- (6.3,2.2);

\Text[x=1.05, y=-.6]{$G_1$}
\Text[x=5.05, y=-.6]{$G_2$}

\Vertex[x=9.7, y=1.8, Math, shape=circle, color=black, , size=.05]{A}
\Vertex[x=10.4, y=1.8, Math, shape=circle, color=black, , size=.05, label=v, fontscale=1, position=above, distance=-.08cm]{B}
\Vertex[x=10.4, y=1.2, Math, shape=circle, color=black, size=.05]{C}
\Vertex[x=9, y=.6, Math, shape=circle, color=black, size=.05, label=u_1, fontscale=1, position=above left, distance=-.17cm]{D}
\Vertex[x=11.1, y=.6, Math, shape=circle, color=black, size=.05, label=u_0, fontscale=1, position=above right, distance=-.17cm]{E}

\Edge[Direct, color=black, lw=1pt](A)(B)
\Edge[Direct, color=black, lw=1pt](D)(B)
\Edge[Direct, color=black, lw=1pt](E)(B)
\Edge[Direct, style=dashed, color=black, lw=1pt, bend=-30](B)(C)

\draw[] (8.8,-.2) -- (11.3,-.2);
\draw[] (8.8,.4) -- (11.3,.4);
\draw[] (8.8,1) -- (11.3,1);
\draw[] (8.8,1.6) -- (11.3,1.6);
\draw[] (8.8,2.2) -- (11.3,2.2);

\Vertex[x=13.7, y=1.8, Math, shape=circle, color=black, , size=.05]{F}
\Vertex[x=14.4, y=1.8, Math, shape=circle, color=black, , size=.05, label=v, fontscale=1, position=above, distance=-.08cm]{G}
\Vertex[x=14.4, Math, shape=circle, color=black, size=.05]{H}
\Vertex[x=13, y=.6, Math, shape=circle, color=black, size=.05, label=u_1, fontscale=1, position=above left, distance=-.17cm]{I}
\Vertex[x=15.1, y=1.8, Math, shape=circle, color=black, size=.05, label=u_0, fontscale=1, position=above right, distance=-.17cm]{J}

\Edge[Direct, color=black, lw=1pt](F)(G)
\Edge[Direct, color=black, lw=1pt](I)(G)
\Edge[Direct, color=black, lw=1pt](J)(G)
\Edge[Direct, style=dashed, color=black, lw=1pt](G)(H)

\draw[] (12.8,-.2) -- (15.3,-.2);
\draw[] (12.8,.4) -- (15.3,.4);
\draw[] (12.8,1) -- (15.3,1);
\draw[] (12.8,1.6) -- (15.3,1.6);
\draw[] (12.8,2.2) -- (15.3,2.2);

\Text[x=10.05, y=-.6]{$G_1$}
\Text[x=14.05, y=-.6]{$G_2$}

\end{tikzpicture}
\caption{Illustration of \autoref{lem:indegree-changes-two-graphs} for the case $\delta^-(v,G_2)=\delta^-(v,G_1)+1$, shown on the left, and for the case $\delta^-(v,G_2)=\delta^-(v,G_1)$, shown on the right. In contrast to \autoref{fig:example-mechanism} only the initial iteration $i=0$ is shown, and the overall drop in the indegree of~$v$ is illustrated by a dashed arrow. Observe that $(\delta^{*}(u_0,G_2), u_0)\succ (\delta^{*}(v,G_1), v) \succ (\delta^{*}(u_0,G_1), u_0)$ and $(\delta^{*}(u_1,G_2), u_1)\succ (\delta^{*}(v,G_1)-1, v),\ (\delta^{*}(u_1,G_1), u_1)\prec (\delta^{*}(v,G_1), v)$.}
\label{fig:lemma-2.1}
\end{figure}

In what follows, we suppose that the first inequality in Condition~\eqref{v-indegree-changes} is an equality, so $(\tilde{v},v)\in E_1\setminus E_2$ and $(E_2\cap N^-(v))\subset (E_1\cap N^-(v))$. 
In this case, $r'-1= \delta^-(v,G_2)-d+\chi(v>z)$. 
Suppose first that $(\delta^-(v,G_1)-1,v)\succ(d, z)$, thus
\[
    (\delta^-(v,G_2),v)=(\delta^-(v,G_1)-1,v)\succ(d, z)\succeq (\delta^{*}(v,G_2), v).
\]
In this case, from \autoref{lem:indegree-changes} there are vertices $u'_0,\ldots,u'_{r'-2}$ such that for all $j\in [r'-2]_0$ we have that $(u'_{j},v)\in E_1\cap E_2$ and $(\delta^{*}(u'_{j},G_2),u'_{j})\succ (\delta^-(v,G_2)-j,v)$.
If $(\delta^{*}(\tilde{v},G_1), \tilde{v})\succ(\delta^{*}(v,G_1),v)$, from expression~\eqref{eq:vertex-descent-characterization-indegrees} we have that 
\[
|\{u \in \{u'_0,\ldots,u'_{r'-2}\}:  (\delta^{*}(u,G_1),u)\succ(\delta^{*}(v,G_1),v)\}|\leq \delta^-(v,G_1)-\delta^{*}(v,G_1)-1.
\]
Therefore, we conclude once again that it is possible to take $\{u_0,\ldots,u_{r-1}\}\subseteq \{u'_0,\ldots,u'_{r'-2}\}$ such that $(\delta^{*}(u_j,G_1),u_j)\prec(\delta^{*}(v,G_1),v)$ for every $j\in [r-1]_0$.
We also have that
\[
    (\delta^{*}(u_j,G_2),u_j)\succ (\delta^-(v,G_2)-j-(\delta^-(v,G_1)-\delta^{*}(v,G_1)-1),v)=(\delta^{*}(v,G_1)-j,v),
\]
so after relabeling the vertices with $(\delta^{*}(u_0,G_2),u_0)\succ\dots \succ (\delta^{*}(u_{r-1},G_2),u_{r-1})$ Condition~\ref{indegree-changes-two-graphs-alt2} follows again.
On the other hand, if $(\delta^{*}(\tilde{v},G_1),\tilde{v})\prec (\delta^{*}(v,G_1),v)$, from expression~\eqref{eq:vertex-descent-characterization-indegrees}
\[
|\{u \in \{u'_0,\ldots,u'_{r'-2}\}:  (\delta^{*}(u,G_1),u)\succ(\delta^{*}(v,G_1),v)\}|\leq \delta^-(v,G_1)-\delta^{*}(v,G_1).
\]
Thus, it is possible to take $\{u_1,\ldots,u_{r-1}\}\subseteq \{u'_0,\ldots,u'_{r'-2}\}$ such that both $(\delta^{*}(u_j,G_2),u_j)\succ(\delta^{*}(v,G_1)-j,v)$ and $(\delta^{*}(u_j,G_1),u_j)\prec(\delta^{*}(v,G_1),v)$ hold for every $j\in [r-1]$.
In addition, whenever $(\delta^{*}(v,G_1),v)\prec(\delta^-(\tilde{v},G_1),\tilde{v})$ we know from \autoref{lem:indegree-changes} that, taking
$\tilde{r}'= \delta^-(\tilde{v},G_1)-\delta^{*}(\tilde{v},G_1)+\chi(\tilde{v}>v)$, there exist vertices $\tilde{u}_0,\ldots,\tilde{u}_{\tilde{r}'-1}$, different than $\tilde{v}$, such that for every $j\in [\tilde{r}'-1]_0$ it holds $(\tilde{u}_{j},\tilde{v})\in E$ and $(\delta^{*}(\tilde{u}_{j},G_1),\tilde{u}_{j})\succ (\delta^-(\tilde{v},G_1)-j,\tilde{v})$. Taking the first $\tilde{r}=\tilde{r}'-(\delta^{*}(v,G_1)-\delta^{*}(\tilde{v},G_1))$ such vertices and $u_0=\tilde{v}$, Condition~\ref{indegree-changes-two-graphs-alt1} follows.

Finally, if $\delta^-(v,G_2)=\delta^-(v,G_1)-1$ and $(\delta^-(v,G_1)-1,v)\prec(d, z)$, we obtain from the inequality $(d, z)\prec(\delta^{*}(v,G_1),v)$ that $\delta^{*}(v,G_1)=\delta^-(v,G_1)$ and $r=1$. Therefore, we must have $(\delta^{*}(\tilde{v},G_1),\tilde{v})\prec (\delta^{*}(v,G_1),v)$, and the same analysis above leads to the existence of vertices $\tilde{u}_0,\ldots,\tilde{u}_{\tilde{r}-1}$, different than $\tilde{v}$, such that for every $j\in [\tilde{r}-1]_0$ it holds $(\tilde{u}_{j},\tilde{v})\in E$ and $(\delta^{*}(\tilde{u}_{j},G_1),\tilde{u}_{j})\succ (\delta^-(\tilde{v},G_1)-j,\tilde{v})$. Condition~\ref{indegree-changes-two-graphs-alt1} follows in this case as well.
An example is shown in \autoref{fig:lemma-2.2}.\qed

\begin{figure}[t]
\centering
\begin{tikzpicture}[scale=0.88]

\Vertex[y=2.1, Math, shape=circle, color=black, size=.05, label=\tilde{u}_0, fontscale=1, position=above, distance=-.08cm]{A}
\Vertex[x=1.4, y=2.1, Math, shape=circle, color=black, size=.05, label=\tilde{v}, fontscale=1, position=above, distance=-.08cm]{B}
\Vertex[x=2.1, y=2.1, Math, shape=circle, color=black, size=.05, label=\tilde{u}_1, fontscale=1, position=above, distance=-.08cm]{C}
\Vertex[x=.7, y=.7, Math, shape=circle, color=black, size=.05, label=v, fontscale=1, position=above left, distance=-.16cm]{D}
\Vertex[x=1.4, y=.7, Math, shape=circle, color=black, size=.05]{E}

\Edge[Direct, color=black, lw=1pt](A)(B)
\Edge[Direct, color=black, lw=1pt](C)(B)
\Edge[Direct, color=black, lw=1pt](B)(D)
\Edge[Direct, style=dashed, color=black, lw=1pt](B)(E)

\draw[] (-.2,-.2) -- (2.3,-.2);
\draw[] (-.2,.5) -- (2.3,.5);
\draw[] (-.2,1.2) -- (2.3,1.2);
\draw[] (-.2,1.9) -- (2.3,1.9);
\draw[] (-.2,2.6) -- (2.3,2.6);

\Vertex[x=4, y=2.1, Math, shape=circle, color=black, size=.05, label=\tilde{u}_0, fontscale=1, position=above, distance=-.08cm]{F}
\Vertex[x=5.4, y=2.1, Math, shape=circle, color=black, size=.05, label=\tilde{v}, fontscale=1, position=above, distance=-.08cm]{G}
\Vertex[x=6.1, y=2.1, Math, shape=circle, color=black, size=.05, label=\tilde{u}_1, fontscale=1, position=above, distance=-.08cm]{H}
\Vertex[x=4.7, Math, shape=circle, color=black, size=.05, label=v, fontscale=1, position=above left, distance=-.16cm]{I}
\Vertex[x=5.4, y=.7, Math, shape=circle, color=black, size=.05]{J}

\Edge[Direct, color=black, lw=1pt](F)(G)
\Edge[Direct, color=black, lw=1pt](H)(G)
\Edge[Direct, style=dashed, color=black, lw=1pt](G)(J)

\draw[] (3.8,-.2) -- (6.3,-.2);
\draw[] (3.8,.5) -- (6.3,.5);
\draw[] (3.8,1.2) -- (6.3,1.2);
\draw[] (3.8,1.9) -- (6.3,1.9);
\draw[] (3.8,2.6) -- (6.3,2.6);

\Text[x=1.05, y=-.6]{$G_1$}
\Text[x=5.05, y=-.6]{$G_2$}

\end{tikzpicture}
\caption{Illustration of \autoref{lem:indegree-changes-two-graphs} for the case where $\delta^-(v,G_2)=\delta^-(v,G_1)-1$ with $(\delta^-(v,G_1), v) \prec (\delta^-(\tilde{v},G_1), \tilde{v})$. Only the initial iteration $i=0$ is shown, and the overall drop in the indegree of~$\tilde{v}$ is illustrated by a dashed arrow. Observe that $(\delta^{*}(\tilde{u}_j,G_1),\tilde{u}_{j})\succ (\delta^-(\tilde{v},G_1)-j,\tilde{v})$ for $j\in\{0,1\}$.}
\label{fig:lemma-2.2}
\end{figure}

\section{Proof of Lemma~\ref{lem:odd-graphs}}
\label{app:odd-graphs}

Let $\calG'_n=\{G\in \calG^T_n:  (s(G))_i = (s(G))_{r(G)+1-i} \text{ for every } i\in [r(G)]\}$ be the set of graphs $G\in \calG^T_n$ such that the tuple $s(G)$ is symmetric. 
From the definition of $\calG^T_n$, it is clear that for every graph $G$ with a non-symmetric tuple $s(G)$, \ie for every graph $G\in \calG^T_n \setminus \calG'_n$, there exists a unique $H\in \calG^T_n \setminus (\calG'_n\cup \{G\})$ such that $r(G)=r(H)=:r$ and $(s(G))_i = (s(H))_{r+1-i} \text{ for every } i\in [r]$. 
Moreover, $\lambda_G=\lambda_H$ for such a pair of graphs. 
This implies that 
\[
    \sum_{G\in \calG^T_n\setminus \calG'_n}\lambda_G
\]
is even.
In what follows we show how to conclude the lemma using the following claim: for every $G\in \calG'_n$ with $r(G)\geq 2$, $\lambda_G$ is even. Observe that
\[
    \sum_{G\in \calG^T_n}\lambda_G = \sum_{G\in \calG^T_n\setminus \calG'_n}\lambda_G + \sum_{G\in \calG'_n:  r(G)\geq 2}\lambda_G + \sum_{G\in \calG'_n:  r(G)=1}\lambda_G.
\]
If the claim is true, we have that the first two sums on the right-hand side are even. The third sum only contains one term, namely $\lambda_G$ for the complete graph $G$, for which $s(G)=(n)$ and thus $\lambda_G=n!/n!=1$. 
Therefore, $\sum_{G\in \calG^T_n}\lambda_G$ is the sum of an even term plus 1, and we conclude the result.

We now prove the claim. Let $G\in \calG'_n$ with $r(G)\geq 2$. 
The multiplicity of the prime factor 2 in the numerator of $\lambda_G$, due to Legendre's formula, is simply
\[
    \sum_{\ell=1}^{\infty}{\left\lfloor\frac{n}{2^{\ell}}\right\rfloor},
\]
while its multiplicity in the denominator is
\[
    \sum_{i=1}^{r(G)}\sum_{\ell=1}^{\infty}{\left\lfloor\frac{(s(G))_i}{2^{\ell}}\right\rfloor}.
\]
Therefore, we need to prove that this last term is strictly lower than the former.
It is easy to see that for every $\ell \in \NN,\ \ell\geq 1$ we have $\sum_{i=1}^{r(G)}{\left\lfloor\frac{(s(G))_i}{2^{\ell}}\right\rfloor}\leq \left\lfloor\frac{n}{2^{\ell}}\right\rfloor$, since without the floor functions we would have the equality, and the fractional parts of the terms $(s(G))_i/2^{\ell}$ must add up ---when summing over $i$--- to at least the fractional part of $n/2^{\ell}$. 
We now show that there exists some $\ell'$ for which this inequality is strict, by distinguishing two cases.
If $(s(G))_i\leq n/2$ for every $i\in [r(G)]$, we let $\ell'\in \NN$ such that $2^{\ell'}\leq n<2^{\ell'+1}$, and then the following holds:
\[
    \sum_{i=1}^{r(G)}\left\lfloor \frac{(s(G))_i}{2^{\ell'}}\right\rfloor 
\leq \sum_{i=1}^{r(G)}\left\lfloor  \frac{n/2}{2^{\ell'}}\right\rfloor 
= \sum_{i=1}^{r(G)}\left\lfloor  \frac{n}{2^{\ell'+1}}\right\rfloor 
= 0
< 1
= \left\lfloor\frac{n}{2^{\ell'}}\right\rfloor.
\]
On the other hand, if $(s(G))_{i'}> n/2$ for some $i'\in [r(G)]$, from the symmetry of $s(G)$ we have that $r(G)$ is odd and, moreover, $i'=(r(G)+1)/2$.
We now define $\ell'\in \NN$ such that $2^{\ell'}\leq n-(s(G))_{i'}<2^{\ell'+1}$, which implies $(s(G))_i \leq \frac{n-(s(G))_{i'}}{2} < 2^{\ell'}$ for every $i\not=i'$. Therefore, the following holds:
\[
    \sum_{i=1}^{r(G)}\left\lfloor \frac{(s(G))_i}{2^{\ell'}}\right\rfloor 
= \left\lfloor \frac{(s(G))_{i'}}{2^{\ell'}}\right\rfloor 
\leq \left\lfloor \frac{n-2^{\ell'}}{2^{\ell'}}\right\rfloor 
= \left\lfloor \frac{n}{2^{\ell'}}-1 \right\rfloor
< \left\lfloor \frac{n}{2^{\ell'}} \right\rfloor.
\]
In either case, we obtain that
\[
    \sum_{i=1}^{r(G)}\sum_{\ell=1}^{\infty}{\left\lfloor\frac{(s(G))_i}{2^{\ell}}\right\rfloor} < \sum_{\ell=1}^{\infty}{\left\lfloor\frac{n}{2^{\ell}}\right\rfloor},
\]
and thus $\lambda_G$ is even, which concludes the proof of the claim and the proof of the lemma.\qed

\section{Proof of Lemma~\ref{lem:symmetry-axiom}}
\label{app:symmetry-axiom}

To see that $f_\text{s}$ is impartial, let $G=(N,E),\ G'=(N,E')\in\calG_n$ and $v\in N$ such that $E\setminus (\{v\}\times N)=E'\setminus (\{v\}\times N)$. 
Since $f$ is impartial,
\[
    (f_{\text{s}}(G))_v = \frac{1}{n!} \sum_{\pi\in \calS_n}(f(G_{\pi}))_{\pi_v} = \frac{1}{n!} \sum_{\pi\in \calS_n}(f(G'_{\pi}))_{\pi_v}=(f_{\text{s}}(G'))_v,
\]
and thus $f_{\text{s}}$ is impartial.

To prove that $f_{\text{s}}$ satisfies weak unanimity, let $G=(N,E)\in \calG_n$ and $v\in N$ with $\delta^-(v)=n-1$. 
Since  $f$ satisfies weak unanimity,
\begin{align*}
    \sum_{u\in N:  \delta^-(u)\geq 1}(f_{\text{s}}(G))_u & = \sum_{u\in N:  \delta^-(u)\geq 1} \frac{1}{n!} \sum_{\pi\in \calS_n}(f(G_{\pi}))_{\pi_u} = \frac{1}{n!} \sum_{\pi\in \calS_n} \sum_{u\in N:  \delta^-(u)\geq 1}(f(G_{\pi}))_{\pi_u}\\
    & \geq  \frac{1}{n!} \sum_{\pi\in \calS_n}1 = 1.
\end{align*}
This concludes the proof of the lemma.\qed

\section*{Acknowledgements}

Research was supported by the Deutsche Forschungsgemeinschaft under project number 431465007 and by the Engineering and Physical Sciences Research Council under grant EP/T015187/1.

\bibliographystyle{abbrvnat}
\bibliography{bibliography}

\end{document}